\documentclass[10pt]{article}
\usepackage{amsmath}
\usepackage{amsfonts}
\usepackage{amssymb}
\usepackage{amsthm}
\usepackage{amscd}
\usepackage{latexsym}
\usepackage{hyperref}
\usepackage{tikz}
\usepackage{tikz-cd}
\usepackage{xcolor}
\usepackage{fullpage}
\usepackage[colorinlistoftodos]{todonotes}

\newtheorem{theorem}{Theorem}[section]
\newtheorem{proposition}[theorem]{Proposition}
\newtheorem{lemma}[theorem]{Lemma}
\newtheorem{corollary}[theorem]{Corollary}

\theoremstyle{definition}

\newtheorem{example}[theorem]{Example}
\newtheorem{remark}[theorem]{Remark}

\newcommand{\Pow}{\mathcal{P}}
\newcommand{\src}{s}
\newcommand{\tgt}{t}
\newcommand{\dom}{\mathit{dom}}
\newcommand{\cod}{\mathit{cod}}
\newcommand{\Dom}{\mathit{Dom}}
\newcommand{\Cod}{\mathit{Cod}}
\newcommand{\id}{\mathit{id}}
\newcommand{\Sup}{\bigvee}
\renewcommand{\sup}{\vee}
\newcommand{\Inf}{\bigwedge}
\renewcommand{\inf}{\wedge}
\newcommand{\conv}{\circ}
\newcommand{\inv}{-}
\newcommand{\Cat}{\mathbf{Cat}}
\newcommand{\Rel}{\mathbf{Rel}}

\newcommand{\fDia}[1]{| #1 \rangle}
\newcommand{\bDia}[1]{\langle #1 |}
\newcommand{\fbDia}[1]{\langle #1 \rangle}

\begin{document}

\title{Higher Catoids, Higher Quantales and their Correspondences}
\author{Cameron Calk \and Philippe Malbos \and Damien Pous \and Georg Struth}

\maketitle

\begin{abstract}  We introduce $\omega$-catoids as generalisations of
  (strict) $\omega$-categories and in particular the higher path
  categories generated by computads or polygraphs in
  higher-dimensional rewriting. We also introduce $\omega$-quantales
  that generalise the $\omega$-Kleene algebras recently proposed for
  algebraic coherence proofs in higher-dimensional rewriting. We then
  establish correspondences between $\omega$-catoids and convolution
  $\omega$-quantales. These are related to Jónsson-Tarski-style
  dualities between relational structures and lattices with operators.
  We extend these correspondences to $(\omega,p)$-catoids, catoids
  with a groupoid structure above some dimension, and convolution
  $(\omega,p)$-quantales, using Dedekind quantales above some
  dimension to capture homotopic constructions and proofs in
  higher-dimensional rewriting. We also specialise them to finitely
  decomposable $(\omega, p)$-catoids, an appropriate setting for
  defining $(\omega, p)$-semirings and $(\omega, p)$-Kleene
  algebras. These constructions support the systematic development and
  justification of $\omega$-Kleene algebra and $\omega$-quantale
  axioms, improving on the recent approach mentioned, where axioms for
  $\omega$-Kleene algebras have been introduced in an ad hoc fashion.\\



\end{abstract}


\section{Introduction}\label{S:introduction}

Rewriting systems are fundamental models of computation. Their rules
generate computations as sequences or paths of rewriting steps.
Computational structure, such as confluence or Church-Rosser
properties, is modelled geometrically in terms of rewriting diagrams
and algebraically via inclusions between rewriting relations, which
form abstractions of sets of rewriting paths~\cite{Terese03}.
Coherence properties of rewriting systems, like the Church-Rosser
theorem, Newman's lemma or normal form theorems, can be derived in
algebras that support reasoning with binary relations: Kleene
algebras, quantales or relation
algebras~\cite{DoornbosBW97,Struth06,DesharnaisS11,DesharnaisMS11}.

Higher-dimensional rewriting generalises and categorifies this
approach, using computads or polygraphs as higher-dimensional
rewriting systems and the higher-dimensional paths they generate
instead of rewriting relations, relations between rewriting relations,
and so forth~\cite{Street76,Burroni91}, see also~\cite{Polybook2025}
for a recent textbook.  This approach has been developed mainly in the
context of (strict) $\omega$-categories, where higher-dimensional
rewriting diagrams have globular shape. They are filled with
higher-dimensional cells instead of relational inclusions (two-cells
in $\Rel$) as witnesses for the $\forall\exists$-relationships between
the higher-dimensional paths that form their faces.  In fact, strict
$(\omega,p)$-categories, where cells of dimension greater than $p$ are
invertible, are often used in practice, for instance for showing that
all parallel reduction cells of a higher-dimensional rewriting system
are contractible.  Applications of higher-dimensional rewriting range
from string rewriting~\cite{GuiraudMalbos18,GuiraudMalbos12advances}
to the computational analysis of coherence properties and cofibrant
approximations in categorical
algebra~\cite{GuiraudHoffbeckMalbos19,MalbosRen23}.

It has recently been argued that higher Kleene algebras, which support
algebraic reasoning about sets of higher-dimensional rewriting paths,
can be used for calculating categorical coherence proofs in
higher-dimensional rewriting~\cite{CalkGMS20} -- just as Kleene
algebras in the classical case. To capture such properties, their
axioms must reflect the shapes of globular cells of
$(\omega,p)$-categories and their pasting schemes: the relationships
between their face maps and the interchange laws that relate the cell
compositions in different dimensions and directions.  Yet a systematic
construction of these algebras and a systematic justification of their
axioms relative to the underlying $(\omega,p)$-categories and
polygraphs has so far been missing.

Drawing from a seemingly unrelated field, we use the correspondences,
in the sense of modal logic, associated with the Jónsson-Tarski
duality between $(n+1)$-ary relational structures and boolean algebras
with $n$-ary operators -- a Stone-type dual equivalence -- as a
guide. In light of this duality, we might consider the multiplication
of a Kleene algebra, for instance, as a binary modal operator and the
fact that an arrow in a category is a composition of two others as a
ternary relation, and look at correspondences between identities that
hold in these structures, for instance, whether an associativity law
on the relational structure makes the multiplication of the Kleene
algebra associative, and vice versa.  Balancing such correspondences
leads us from $\omega$-categories to $\omega$-catoids, which are
isomorphic to ternary-relational structures with suitable relational
laws, and from higher Kleene algebras to $\omega$-quantales.  Imposing
axioms on $\omega$-catoids then allows us to derive axioms on
$\omega$-quantales and vice versa, until balance is achieved. While
this could be considered for $\omega$-catoids on a set $X$ and
$\omega$-quantales on the powerset $\Pow X$, we generalise this
construction to correspondence triangles for $\omega$-catoids $C$,
$\omega$-quantales $Q$ and convolution $\omega$-quantales on function
spaces $Q^C$, and further to $(\omega,p)$-structures.

We briefly outline the simplest case to supply some basic intuition. A
powerset quantale on a set $X$ is a quantale on $\Pow X$, that is, the
complete lattice $(\Pow X,\subseteq)$ and at the same time a monoid
$(\Pow X,\cdot,1)$ such that the binary operator $\cdot$ preserves
arbitrary sups in both arguments. In every powerset quantale, the
multiplication on singleton sets, the atoms in this structure, can be
arranged into a ternary relation $\{x\}\subseteq \{y\}\cdot \{z\}$,
which satisfies a certain relational associativity law and has the
elements of the set $1$ as relational units (the laws of a monoid
object in $\Rel$). As each element $x\in X$ has exactly one left and
one right relational unit, they can be assigned to $x$ by a source map
$s:X\to X$ and a target map $t:X\to X$. Using a multioperation
$\odot:X\times X\to \Pow X$ instead of the ternary relation
$X\times X\times X\to 2$, so that
$x\in y\odot z\Leftrightarrow \{x\}\subseteq \{y\}\cdot \{z\}$, and
emphasising source and target maps instead of unit elements, leads to
the definition of a catoid \cite{CranchDS20,FahrenbergJSZ21a} as a
structure $(X,\odot,\src,\tgt)$ where, for all $x,y,z\in C$,
\begin{gather*}
  \bigcup\{x\odot v\mid v \in y \odot z\} = \bigcup\{u\odot z\mid u\in
  x \odot y\},\\
  x\odot y \neq \emptyset \Rightarrow \tgt(x)=\src(y),\qquad
  \src(x)\odot x = \{x\},\qquad x\odot \tgt(x)=\{x\}.
\end{gather*}
Powerset quantales on $\Pow X$ thus give rise to catoids:
associativity of $\odot$ is derived using associativity of
$\cdot$. Conversely, starting from a catoid $(X,\odot,s,t)$, one can
construct a quantale on $\Pow X$ with composition
$A\cdot B=\{c \in a\odot b\mid a \in A, b\in B\}$. Tying these
constructions together yields a dual equivalence between the category
of catoids (with suitable morphisms) and the category of powerset
quantales (with suitable homomorphisms), an instance of the
Jónsson-Tarski duality mentioned. The multioperation of the catoid is
thus simply an alternative encoding of a ternary relation, while the
multiplication of the powerset quantale is seen as binary operator on
a boolean algebra. We are, however, not interested in such a duality
itself, but in the correspondences between equations that are typical
for the modal logics and algebras associated with it, as illustrated
in the example of associativity above.

\begin{figure}[t]
  \center
  
  \begin{tikzpicture}[x=2.5cm, y=2.5cm]
    \node (X) at (0,0) {catoid $C$};
    \node (Q) at (1,0) {quantale $Q$};
    \node (R) at (.5,.8) {quantale $Q^C$};
    \path (X) edge[-] node[coordinate, pos=.2] (QX) {} node[coordinate, pos=.8] (XQ) {} (Q);
    \path (X) edge[-] node[coordinate, pos=.2] (RX) {} node[coordinate, pos=.85] (XR) {} (R);
    \path (Q) edge[-] node[coordinate, pos=.2] (RQ) {} node[coordinate, pos=.85] (QR) {} (R);
  \end{tikzpicture}
  \caption{Basic two-out-of-three correspondence between catoid $C$, quantale
    $Q$ and convolution quantale $Q^C$}
  \label{Fig:triangle1}
\end{figure}

In the above construction, the powerset on $X$ corresponds to a map
$X\to 2$ into the quantale $2$ of booleans, and $2$ can be replaced by
an arbitrary quantale $Q$.  This leads to correspondence triangles
between catoids $C$, value quantales $Q$ and convolution quantales
$Q^C$ on function spaces, as depicted in
Figure~\ref{Fig:triangle1}. In the simplest case, if $C$ is a catoid
and $Q$ a quantale, then $Q^C$ is a quantale; if $Q^C$ and $Q$ are
quantales, then $C$ is a catoid; and if $Q^C$ is a quantale and $C$ a
catoid, then $Q$ is a quantale~\cite{DongolHS21,CranchDS21}.  The last
two of these two-out-of-three properties require mild conditions on
$C$ or $Q$ explained in Section~\ref{S:globular-convolution}. In the
construction of the convolution quantale $Q^C$, if
$\odot:C\times C\to \Pow C$ is the multioperation on the catoid $C$,
if $\cdot:Q\times Q\to Q$ is the composition in the value quantale $Q$
and if $f$, $g$ are maps in $C\to Q$, then the quantalic composition
$\ast : Q^C\times Q^C \to Q^C$ is the convolution
\begin{equation*}
  (f\ast g)(x) = \Sup_{x \in y\odot z} f(y)\cdot g(z),
\end{equation*}
where the sup is taken with respect to $y$ and $z$.  Functions
$\delta_x^\alpha:C\to Q$, which map $y\in C$ to $\alpha\in Q$ if $y=x$
and to the minimal element $\bot$ of the quantale $Q$ otherwise, now
replace singleton sets as atoms.  They allow us to obtain equations in
$C$ from those in $Q^C$ and $Q$, as well as $Q$ from $Q^C$ and
$C$. Units in two of these structures then give rise to a
corresponding unit in the third, see~\cite{CranchDS21} and
Section~\ref{S:globular-convolution} for details. To construct a
convolution or power set quantale, it is thus necessary and sufficient
to understand the structure of the underlying catoid (if the value
quantale is fixed).
 
Our correspondence results for $\omega$-catoids and $\omega$-quantales
are based on two extensions of this basic triangle.  The first is a
correspondence triangle between interchange catoids $C$, interchange
quantales $Q$ and convolution interchange quantales
$Q^C$~\cite{CranchDS21}. Catoids are then equipped with two
multioperations and quantales with two monoidal structures that
interact via interchange laws. This first extension helps us to deal
with the interchange laws in $\omega$-catoids which we wish to reflect
within $\omega$-quantales.

The second extension is a correspondence triangle between local
catoids $C$, which exhibit the typical composition pattern of
categories, modal value quantales $Q$ and modal convolution quantales
$Q^C$~\cite{FahrenbergJSZ21a}.  Here, correspondences between
properties of the source and target maps of catoids and the domain and
codomain maps of modal quantales extend the basic ones. This second
extension justifies the laws of modal Kleene
algebras~\cite{DesharnaisS11} and modal
quantales~\cite{FahrenbergJSZ22a} relative to the catoid and category
axioms. In powerset quantales, for instance, the domain and codomain
operations in a powerset quantale arise as the direct images of the
source and target maps of the corresponding catoid, and source and
target maps in a catoid are obtained by restricting the application of
domain and codomain operations of a powerset quantale to singleton
sets. The second extension thus sets up the correspondence between the
source and target structure in $\omega$-catoids and the domain and
codomain structure on $\omega$-catoids.

In combination, the correspondence for interchange laws and that for
source/target and domain/codomain therefore help us to reflect the
full structure of strict $\omega$-categories in powerset or
convolution algebras, using $\omega$-catoids to balance the equational
axioms in the two kinds of structures in correspondence proofs.

\begin{figure}[t]
    \centering
  \begin{tikzpicture}[x=3.2cm, y=3.2cm]
    \node (X) at (0,0) {local $\omega$-catoid $C$};
    \node (Q) at (1,0) {$\omega$-quantale $Q$};
    \node (R) at (.5,.8) {$\omega$-quantale $Q^C$};
    \path (X) edge[-] node[coordinate, pos=.2] (QX) {} node[coordinate, pos=.8] (XQ) {} (Q);
    \path (X) edge[-] node[coordinate, pos=.2] (RX) {} node[coordinate, pos=.85] (XR) {} (R);
    \path (Q) edge[-] node[coordinate, pos=.2] (RQ) {} node[coordinate, pos=.85] (QR) {} (R);
  \end{tikzpicture}
  \quad
   \begin{tikzpicture}[x=3.5cm, y=3.2cm]
    \node (X) at (0,0) {local $(\omega,p)$-catoid $C$};
    \node (Q) at (1,0) {$(\omega,p)$-quantale $Q$};
    \node (R) at (.5,.8) {$(\omega,p)$-quantale $Q^C$};
    \path (X) edge[-] node[coordinate, pos=.2] (QX) {} node[coordinate, pos=.8] (XQ) {} (Q);
    \path (X) edge[-] node[coordinate, pos=.2] (RX) {} node[coordinate, pos=.85] (XR) {} (R);
    \path (Q) edge[-] node[coordinate, pos=.2] (RQ) {} node[coordinate, pos=.85] (QR) {} (R);
  \end{tikzpicture}   
  \caption{Correspondences between local $\omega$-catoid $C$,
    $\omega$-quantale $Q$ and $\omega$-quantale $Q^C$ on the left, and
    local $(\omega,p)$-catoid $C$,
    $(\omega,p)$-quantale $Q$ and $(\omega,p)$-quantale $Q^C$ on the right}
  \label{Fig:triangle2}
\end{figure}

The correspondence triangles between $\omega$-catoids and
$\omega$-quantales, shown in Figure~\ref{Fig:triangle2}, the main
technical results in this article, require first of all definitions of
these two structures. $\omega$-Catoids are introduced in
Section~\ref{S:2-lr-msg}. They are obtained by adding globular shape
axioms to those of local and interchange catoids and then generalising
beyond two dimensions in light of previous axiomatisations of
(single-set)
$\omega$-categories~\cite{BrownH81,Street87,MacLane98,Steiner04}. Based
on the $\omega$-catoid axioms and on previous axioms for globular
$n$-Kleene algebras~\cite{CalkGMS20}, we then introduce
$\omega$-quantales in Section~\ref{S:2-quantales} as our main
conceptual contribution, and justify their axioms through the
correspondence proofs in Section~\ref{S:globular-convolution}. Results
for $n$-structures can be obtained in the standard way by truncation;
correspondence results for powerset $\omega$-quantales are discussed
in Section~\ref{S:msg-quantale}. In this article, $\omega$-catoids, as
simple generalisations of strict $\omega$-categories, are therefore
merely tools for deriving the $\omega$-quantale axioms and vice
versa. They do not constitute any attempt towards infinity categories.

The $\omega$-catoid axioms depend on rather delicate definedness
conditions for compositions, captured multioperationally by mapping to
non-empty sets. These are sensitive to Eckmann-Hilton-style collapses,
as discussed in Appendix~\ref{A:eckmann-hilton}.
They satisfy natural functorial properties, specialise to
$\omega$-category axioms and yield the globular cell structure
expected, see Section~\ref{S:2-lr-msg}. The Isabelle/HOL proof
assistant with its automated theorem provers and counterexample
generators has allowed us to simplify these structural axioms quite
significantly, which in turn simplified the development of
$\omega$-quantale axioms in Section~\ref{S:2-quantales} and the proofs
of correspondence triangles in Section~\ref{S:globular-convolution}.

The extended correspondence triangles for $(\omega,p)$-catoids and
$(\omega,p)$-quantales in Section~\ref{S:n-structures}, shown in
Figure~\ref{Fig:triangle2}, are compositional with respect to the
$\omega$-correspondences and those for groupoids and quantales with a
suitable notion of {involution or converse. These are established in
  Section~\ref{S:dedekind-correspondence}, see
  Figure~\ref{Fig:triangle3}. They adapt Jónsson and Tarski's
  classical duality between groupoids and relation
  algebras~\cite{JonssonT52} and extend it to convolution
  algebras. Catoids specialise automatically to groupoids when the two
  obvious axioms for inverses are added, see
  Section~\ref{S:lr-groupoids}.  Yet we use Dedekind quantales instead
  of relation algebras here. These are involutive
  quantales~\cite{MulveyWP92} equipped with a variant of the Dedekind
  law from relation algebra, see
  Section~\ref{S:modular-quantales}. Apart from the interaction
  between the converse and the modal structure needed for
  $(\omega,p)$-quantales, we also discuss weaker variants of converse
  useful for semirings or Kleene algebras in this section.

  \begin{figure}[t]
    \centering
    \begin{tikzpicture}[x=3.4cm, y=3.2cm]
      \node (X) at (0,0) {groupoid $C$};
      \node (Q) at (1,0) {Dedekind quantale $Q$};
      \node (R) at (.5,.8) {Dedekind quantale $Q^C$};
      \path (X) edge[-] node[coordinate, pos=.2] (QX) {} node[coordinate, pos=.8] (XQ) {} (Q);
      \path (X) edge[-] node[coordinate, pos=.2] (RX) {} node[coordinate, pos=.85] (XR) {} (R);
      \path (Q) edge[-] node[coordinate, pos=.2] (RQ) {} node[coordinate, pos=.85] (QR) {} (R);
    \end{tikzpicture}
    \caption{Correspondence between groupoid $C$,
      Dedekind quantale $Q$ and convolution Dedekind quantale $Q^C$,
      where both quantales are assumed to carry a complete Heyting
    algebra structure}
  \label{Fig:triangle3}
\end{figure}

The content of the remaining sections of this paper is summarised as
follows. In Section~\ref{S:lr-multisemigroups} and
\ref{S:modal-quantales} we recall the basic properties of catoids and
modal quantales. In Section~\ref{S:modalities} we develop the basic
laws for modal box and diamond operators in $\omega$-quantales in
preparation for more advanced future coherence proofs in
higher-dimensional rewriting. In Section~\ref{S:globular-semiring} and
\ref{S:globular-ka} we specialise the $\omega$- and
$(\omega,p)$-quantale axioms to those of $\omega$-semirings and
$\omega$-Kleene algebras and their $(\omega,p)$-variants. For general
convolution algebras, this requires a finite decomposition property on
$\omega$- and $(\omega,p)$-catoids.
Sections~\ref{S:globular-semiring} and~\ref{S:globular-ka} also
contain a detailed comparison of the $\omega$-Kleene algebras and
$(\omega,p)$-Kleene algebras introduced in this paper with previous
globular $n$-Kleene algebras~\cite{CalkGMS20} and their slightly
different axioms. Overall, $\omega$-quantales offer greater
flexibility when reasoning about higher rewriting diagrams than
$\omega$-Kleene algebras. They admit arbitrary suprema and additional
operations such as residuals, and they support reasoning with least
and greatest fixpoints beyond the Kleene star. Already the proof of
the classical Newman's lemma in modal Kleene
algebras~\cite{DesharnaisMS11} assumes certain suprema that are not
present in all Kleene algebras. Nevertheless, convolution semiring and
Kleene algebras have been widely studied in computer
science~\cite{DrosteKV09} and their higher variants therefore
certainly deserve an exploration.

Most results in this paper have been verified with the Isabelle/HOL
proof assistant, but the development of interactive proof support for
higher categories, higher-dimensional rewriting and categorical
algebra remains part of a larger research programme, which requires
substantial additions. Our Isabelle components can be found in the
Archive of Formal Proofs~\cite{Struth23,CalkS23,CalkS24}. The
components contain specifications and basic libraries for $2$-catoids,
groupoids, $2$-quantales and their $\omega$-variants, as well as for
Dedekind quantales, and they cover the basic properties of these
structures that feature in this paper. All extensions to powersets
have been formalised with Isabelle/HOL, but neither the constructions
of convolution algebras nor the full correspondence triangles. In
addition to the Isabelle proofs, we present the most important proofs
for this article in the relevant sections. All other proofs can be
found in Appendix~\ref{A:proofs} or the references given, unless they
are trivial.

Finally, Appendix~\ref{A:diagrams} provides diagrams for the most
important structures used in this articles and their relationships.


\section{Catoids}\label{S:lr-multisemigroups}

In preparation for the $\omega$-catoids in Section~\ref{S:2-lr-msg} we
start with recalling the definitions and basic properties of catoids
and related structures~\cite{FahrenbergJSZ21a}; see also the Isabelle
theories~\cite{Struth23} for details. General background on
multisemigroups can be found in~\cite{KudryavtsevaM15} and the
references given there. All proofs in this section
have been checked with Isabelle.

A \emph{catoid} $(C,\odot,\src,\tgt)$ consists of a set $C$, a
multioperation $\odot:C\times C\to \Pow C$ and \emph{source} and
\emph{target} maps $\src,\tgt:C\to C$. These satisfy, for all
$x,y,z\in C$,
\begin{gather*}
  \bigcup\{x\odot v\mid v \in y \odot z\} = \bigcup\{u\odot z\mid u\in
  x \odot y\},\\
  x\odot y \neq \emptyset \Rightarrow \tgt(x)=\src(y),\qquad \src(x)\odot
  x = \{x\},\qquad x\odot \tgt(x)=\{x\}. 
\end{gather*}

The first catoid axiom expresses multirelational \emph{associativity}.
If we extend $\odot$ from $C\times C\to \Pow C$ to
$\Pow C\times \Pow C\to \Pow C$ such that, for all $X,Y\subseteq C$,
\begin{equation*}
  X\odot Y= \bigcup_{x\in X,y\in Y} x\odot y,
\end{equation*}
and write $x\odot X$ for $\{x\}\odot X$ and likewise, then the
associativity axiom simplifies to $x\odot (y\odot z)= (x\odot y)\odot z$.  A
\emph{multisemigroup} can thus be defined as a set equipped with an
associative multioperation~\cite{KudryavtsevaM15}.  From now on, we
often write $xy$ instead of $x\odot y$. We also write $s(X)$ and $t(X)$ for the
direct images of $X$ under $s$ and $t$, respectively, for instance, $s(x\odot
y)$ or $s(xy)$ as well as $t(x\odot y)$ or $t(xy)$. 

The second catoid axiom, the \emph{weak locality} axiom, states that
the target $\tgt(x)$ of $x$ and the source $\src(y)$ of $y$ are equal
whenever the composite of $x$ and $y$ is defined, that is,
$x\odot y\neq \emptyset$. We write $\Delta(x,y)$ for
$x\odot y \neq \emptyset$ and call $\Delta$ the \emph{domain of
  definition} of $\odot$.

The third and fourth catoid axioms are \emph{left} and \emph{right
  unit} axioms and we refer to $s(x)$ and $t(x)$ as the \emph{left
  unit} and \emph{right unit} of $x$, respectively.

A catoid $C$ is \emph{functional} if $x,x'\in yz$ imply $x=x'$ for all
$y,z\in C$, and \emph{local} if
$\tgt(x)= \src(y)\Rightarrow xy\neq \emptyset$ for all $x,y\in C$.  A
\emph{category} is then a local functional catoid.

Catoids thus generalise categories beyond locality and
functionality. Local functional catoids formalise categories in
single-set style~\cite{MacLane48}, see also~\cite[Chapter
XII]{MacLane98}, and with arrow composition in diagrammatic order.  In
functional catoids, composition is a partial operation that maps
either to singleton sets or to the empty set. Locality imposes the
standard composition pattern
$\Delta(x,y) \Leftrightarrow \tgt(x)=\src(y)$ of arrows in categories.

Composition in a catoid is \emph{total} if $\Delta=C\times C$. In
every total catoid $C$, there is precisely one element which is the
source and target element of every element in $C$.  Total
multioperations are known as \emph{hyperoperations}.  Total operations
are therefore total functional multioperations.

\begin{remark}
  In the definition of catoids and throughout this text, we are using
  ``set'' indifferently for small sets and classes and ignore the well
  known foundational issues, which do not arise in our
  setting. Distinctions can be made as for standard categories. See,
  for instance, Mac Lane's book~\cite{MacLane98} and
  Example~\ref{E:ExampleCatoidsCategories} for a discussion.
\end{remark}

Catoids form a category with respect to several notions of morphism.
A \emph{catoid morphism} $f:C\to D$ between catoids $C$ and $D$ must
preserve compositions, sources and targets:
\begin{equation*}
  f(x\odot_C y)\subseteq f(x)\odot_D f(y),\qquad
  f\circ \src_C = \src_D\circ f,\qquad
  f\circ \tgt_C= \tgt_D\circ f,
\end{equation*}
where, on the left-hand side of the left identity the image of the set
$x\odot_C y$ with respect to $f$ is taken.  A morphism $f:C\to D$ is
\emph{bounded} if $f(x)\in u\odot_D v$ implies that
$x \in y\odot_C z$, $u=f(y)$ and $v=f(z)$ for some $y,z\in C$.  

Morphisms of categories, as local functional catoids, are
functors. The inclusion in the definition of morphisms reflects that
$x\odot_C y=\emptyset $ whenever the composition is undefined. Bounded
morphisms are widely used in modal and substructural logics. 

\begin{example}\label{ex:bounded-morphism}
  Bounded morphisms between catoids need not satisfy
  $f(x\odot_C y) = f(x)\odot_D f(y)$.  Consider the discrete catoid on
  $C=\{a,b\}$ with $\src=\id_{C} = \tgt$ and
  \begin{equation*}
    xy =
    \begin{cases}
      \{x\} & \text{ if } x=y,\\
    \emptyset & \text{ otherwise.}
    \end{cases}
  \end{equation*}
  The constant map $f_b:x\mapsto b$ on $C$ is clearly a catoid
  endomorphism. It is bounded because every $x\in C$ satisfies
  $f_b(x) \in bb$, $x\in xx$ and $b=f_b(x)$. Nevertheless we have
  $f_b(ab) = \emptyset \neq \{b\} = f_b(a) f_b(b)$.  
\end{example}

The \emph{opposite} of a catoid is defined as for categories. It is a
structure in which $\src$ and $\tgt$ are exchanged and so are the
arguments in $\odot$. The opposite of a (local, functional, total) catoid
is again a (local, functional, total) catoid. Properties of catoids
translate through this duality. 

Properties of catoids and related structures have been collected
in~\cite{FahrenbergJSZ21a} and our Isabelle
theories~\cite{Struth23}. Here we list only some that are structurally
interesting or needed in proofs below.

\begin{lemma}\label{lemma:mm-props}\label{lemma:mm-quiver}
  In every catoid,
  \begin{enumerate}
  \item $\src\circ \src = \src$, $\tgt\circ \tgt=\tgt$, $\src\circ \tgt = \tgt$ and
    $\tgt\circ \src = \src,$
    \item $\src(x) = x \Leftrightarrow x = \tgt(x)$,
  \item $\src(x)\src(x)=\{\src(x)\}$ and $\tgt(x)\tgt(x)=\{\tgt(x)\}$,
  \item $\src(x)\tgt(y)=\tgt(y)\src(x)$,
    \item $\src(\src(x)y) = \src(x)\src(y)$ and $\tgt(x\tgt(y))=\tgt(x)\tgt(y)$.
  \end{enumerate}
\end{lemma}
Note that direct images with respect to $s$ and $t$ are taken in the
left-hand sides of the identities in (5).  According to (2),
  the set of fixpoints of $s$ equals the set of fixpoints of $t$. We
  henceforth write $C_0$ for this set.

\begin{lemma}\label{lemma:fix-id}
  Let $C$ be a catoid. Then $C_0 = \src(C) =\tgt(C)$.
\end{lemma}
The proof is immediate from Lemma~\ref{lemma:mm-props}(1). Hence $C_0$
consists of the units in $C$ which we also call \emph{$0$-cells}, by
analogy with categories.  Similarly, the elements of $C$ can be viewed
as \emph{$1$-cells}, the elements of $C_0$ as degenerate $1$-cells
($\src(x)=\tgt(x)$ holds for all $x\in C_0$ by
Lemma~\ref{lemma:mm-props}(2)) and those of $C-C_0$ as \emph{proper}
or \emph{non-degenerate} $1$-cells. Further, units of catoids can be
seen as orthogonal idempotents.
\begin{lemma}\label{lemma:orth-id}
  Let $C$ be a catoid. For all $x,y\in C_0$,
  $\Delta(x,y)\Leftrightarrow x=y$ and 
\begin{equation*}
  x y =
  \begin{cases}
   \{x\} & \text{ if } x = y,\\
    \emptyset &\text{ otherwise}.
  \end{cases}
\end{equation*}
\end{lemma}

The next two lemmas recall an alternative equational characterisation
of locality, which is important for the correspondence with modal
quantales in Section~\ref{S:globular-convolution}.
  
\begin{lemma}\label{lemma:msg-props1}
  In every catoid,
  \begin{enumerate}
  \item $\src (xy) \subseteq \src(x\src(y))$ and $\tgt(xy)\subseteq
    \tgt(\tgt(x)y)$,
  \item $\Delta(x,y)$ implies $\src(xy)=\{\src(x)\}$ and
    $\tgt(xy)=\{\tgt(y)\}$. 
  \end{enumerate}
\end{lemma}

Item (2) of the following lemma features the equational
characterisation of locality mentioned.

\begin{lemma}\label{lemma:msg-loc}\label{lemma:msg-props2}\label{lemma:msg-loc-var}
  In every catoid $C$, the following statements are then equivalent:
  \begin{enumerate}
  \item $C$ is local,
  \item $\src(xy)=\src(x\src(y))$ and $\tgt(xy)=\tgt(\tgt(x)y)$, for all $x,y\in C$,
  \item $\Delta(x,y) \Leftrightarrow \tgt(x)\src(y)\neq \emptyset $,
    for all $x,y\in C$.
  \end{enumerate}
\end{lemma}

Many examples of catoids and related structures are listed
in~\cite{FahrenbergJSZ21a}.  Here we mention only a few. First we
summarise the relationship with categories.

\begin{example}
\label{E:ExampleCatoidsCategories}
The category of local functional catoids with (bounded) morphisms is
isomorphic to the category of single-set categories à la
MacLane~\cite[Chapter XII]{MacLane98} with (bounded)
morphisms~\cite[Proposition 3.10]{FahrenbergJSZ21a}. The elements of
single-set categories are arrows or $1$-cells of traditional
categories. The objects of $0$-cells of traditional categories
correspond bijectively to identity $1$-cells and thus to units of
catoids.
  \end{example}
  
\begin{example}\label{ex:free-cat}
  The free category or \emph{path category} generated by a given
  digraph $\src,\tgt:E\to V$, for a set $E$ of edges or $1$-generators
  and a set $V$ of vertices or $0$-generators, is a fortiori a
  catoid. Recall that morphisms of path categories are finite paths
  between vertices, represented as alternating sequences of vertices
  and edges. Source and target maps extend from edges to paths, and we
  write $\pi:v\to w$ for a path $\pi$ with source $v$ and target
  $w$. Paths $\pi_1:u\to v$ and $\pi_2:v\to w$ can be composed to the
  path $vw:u\to w$ by gluing their ends. Identities are constant paths
  of length zero such as $v:v\to v$. See~\cite{MacLane98} for details.

  Digraphs themselves can be modelled as single-set structures
  $(X,\src,\tgt)$ satisfying $\src\circ \src = \src$,
  $\tgt\circ \tgt =\tgt$, $\src \circ \tgt = \tgt$ and
  $\tgt \circ \src = \src$.  These conditions make
  $X_0 = \src(X)=\tgt(X)$ the set of vertices.
\end{example}

In higher-dimensional rewriting, digraphs are referred to as
$1$-computads or $1$-polygraphs; in traditional rewriting, they
correspond to abstract rewriting systems. The free category generated
by a digraph supplies rewriting paths, the main object of study in
rewriting systems.  The recursive construction of
higher-dimensional computads or polygraphs has two-steps: for any
dimension $n\geq 0$, assuming that $k$-generators have been supplied
for all $0\leq k\leq n$, form the free $n$-category on $n$-generators,
then add $(n+1)$-generators over this free category.

The next example presents a catoid which is not a category and which
will recur across this text.

\begin{example}\label{ex:shuffle-catoid}
  Let $\Sigma^\ast$ be the free monoid generated by the finite set
  $\Sigma$.  The \emph{shuffle catoid} $(\Sigma^\ast,\|,\src,\tgt)$ on
  $\Sigma^\ast$ has the total commutative multioperation
  $\|:\Sigma^\ast\times \Sigma^\ast \to \Pow\Sigma^\ast$ as its
  composition. For all $a,b\in \Sigma$ and $v,w\in \Sigma^\ast$, it is
  defined recursively as
  \begin{equation*}
    \varepsilon \| v = \{v\} = v\| \varepsilon\qquad\text{
      and }\qquad (av) \| (bw) = a(v\|(bw))\cup
    b((av)\| w),
  \end{equation*}
  where $\varepsilon$ denotes the empty word.  The source and target
  structure of the shuffle catoid is trivial:
  $\src(w)=\varepsilon=\tgt(w)$ for all $w\in \Sigma^\ast$. It is
  local because of totality and triviality of $\src$ and $\tgt$. It is
  obviously not functional and hence not a category.
\end{example}

\begin{remark}\label{R:multimonoids}
  Catoids can be seen as multimonoids, that is, multisemigroups with
  multiple units, which generalise single-set categories with multiple
  units~{\cite[Chapter I]{MacLane98}} to multioperations. Multimonoid
  morphisms are then unit-preserving multisemigroup morphisms, which
  ignore source- and target-preservation. Categories of catoids and
  multimonoids, both with the obvious morphisms, are
  isomorphic; the functional relationship between elements of
  multimonoids and their left and right units determines source and
  target maps~\cite{CranchDS20},\cite[Proposition
  3.10]{FahrenbergJSZ21a}.  Categories of local partial multimonoids
  are therefore isomorphic to categories. See~\cite{CranchDS20} for
  the appropriate notion of locality.
\end{remark}

\begin{remark}
  Multioperations $X\times X\to \Pow X$ are isomorphic to ternary
  relations $X\to X\to X\to 2$.  Writing $R^x_{yz}$ for
  $x\in y \odot z$ allows us to axiomatise catoids alternatively as
  relational structures $(C,R,\src,\tgt)$ that satisfy the relational
  associativity axiom
  $\exists v.\ R^w_{xv}\land R^v_{yz} \Leftrightarrow \exists u.\
  R^u_{xy} \land R^w_{uz}$, the weak locality axiom
  $\exists z.\ R^z_{xy} \Rightarrow \tgt(x)=\src(z)$ and the
  relational unit axioms $R^x_{\src(x)x}$ and $R^x_{xt(x)}$.  Such
  \emph{relational monoids} are, in fact, monoids in the monoidal
  category $\mathbf{Rel}$ with the standard
  tensor~\cite{Rosenthal97,KenneyP11}. Functionality and locality
  translate readily to relations. Relational variants of catoids have
  been studied in~\cite{CranchDS20}.  Morphisms and bounded morphisms
  of ternary relations are standard for modal and substructural
  logics~\cite{Goldblatt89,Givant14}; categories of catoids and
  relational monoids are once again
  isomorphic~\cite{FahrenbergJSZ21a}.

  Ternary and more generally $(n+1)$-ary relations appear as duals of
  binary and more generally $n$-ary modal operators
  on boolean algebras in Jónsson and Tarski's duality theory
  for boolean algebras with
  operators~\cite{JonssonT51,JonssonT52,Goldblatt89,Givant14}. 
\end{remark}

\begin{remark}
  A partial operation $\hat \odot:\Delta\to C$, where
  $y\mathop{\hat \odot} z$ denotes the unique element that satisfies
  $y\mathop{\hat \odot} z\in y\odot z$ whenever $\Delta(y,z)$, can be
  defined in any functional catoid. Using this partial operation,
  $x\in y\odot z$ if and only if $\Delta(y,z)$ and
  $x= y\mathop{\hat{\odot}} z$. 
\end{remark}


\section{Groupoids}\label{S:lr-groupoids}

Higher-dimensional rewriting usually requires rewriting steps to be
invertible above a certain dimension. This
amounts to using groupoids, see~\cite{Brown87} for a survey.
Particularly relevant to us is the work of Jónsson and
Tarski~\cite[Section 5]{JonssonT52} on the correspondence between
groupoids and relation algebras, which we revisit in the slightly
different setting of Dedekind quantales in
Section~\ref{S:dedekind-correspondence}. Several properties that
feature in this section are theirs. In Section~\ref{S:n-structures},
catoids are combined with groupoids to $(n,p)$-catoids and
$(\omega,p)$-catoids, which specialise to the $(n,p)$-categories and
$(\omega,p)$-categories used in higher-dimensional
rewriting. Interestingly, catoids become groupoids and hence
categories when the natural axioms for inverses are added. All
proofs in this section have been checked with
Isabelle~\cite{Struth23}.

A \emph{groupoid} is a catoid $C$ with an \emph{inversion} operation
$(-)^\inv:C\to C$ such that, for all $x\in C$,
\begin{equation*}
  xx^\inv = \{\src(x)\} \qquad\text{ and }\qquad x^\inv x = \{\tgt(x)\}.
\end{equation*}

To justify this definition, we proceed in two steps to derive locality
and functionality, showing selected proofs only. The remaining ones can be
found in Appendix~\ref{A:proofs}.

\begin{lemma}\label{lemma:lr-mgp-local}\label{lemma:lr-mgp-props}
  Let $X$ be a catoid with operation $(-)^\inv:C\to C$ that satisfies
  $\src(x)\in xx^\inv$ and $\tgt(x)\in x^\inv x$ for all $x\in
  C$. Then
  \begin{enumerate}
    \item $C$ is local, 
  \item $\src (x^\inv) = \tgt(x)$ and $\tgt(x^\inv) = \src(x)$,
  \item $xy=\{\src(x)\}$ implies $x^\inv = y$ and $yx=\{\tgt(x)\}$ implies $x^\inv =y$,
    \item $\src(x)^\inv = \src(x)$ and $\tgt(x)^\inv = \tgt(x)$. 
  \end{enumerate}  
\end{lemma}
\begin{proof}
  For (1), suppose $\tgt(x)=\src(y)$. Then
  $\{x\} = x\tgt(x)= x\src(y) = x (yy^-) = (xy)y^-$. Hence ${x}= uy^-$
  and $u\in xy$ hold for some $u\in C$. Thus $\Delta(x,y)$ and $C$ is
  local.
\end{proof}
For proofs of (2)-(4) see Appendix~\ref{A:proofs}.

\begin{lemma}\label{lemma:lr-groupoid-props}
  In every groupoid,
  \begin{enumerate}
  \item $(x^\inv)^\inv = x$,
  \item $x\in yz \Leftrightarrow y \in xz^\inv \Leftrightarrow z \in y^\inv x$.
  \end{enumerate}
\end{lemma}
See Appendix~\ref{A:proofs} for proofs.

\begin{proposition}\label{prop:lr-gpd-cat}
  Every groupoid (as defined above) is a category.
\end{proposition}
\begin{proof}
  In light of Lemma~\ref{lemma:lr-mgp-local}(1) it remains to check
  functionality. Suppose $x,x'\in yz$. Then $z\in y^\inv x$ by
  Lemma~\ref{lemma:lr-groupoid-props}(2) and
  $x'\in yy^\inv x = \src(y)x = \src(x)x=\{x\}$ using the first
  assumption and $ v\in xy \Rightarrow \src(v)=\src(x)$, which holds
  in any catoid by Lemma~\ref{lemma:msg-props2}, in the second
  step. Thus $x'= x$.
\end{proof}

\begin{lemma}\label{lemma:lr-gpd-canc}
  The following cancellation properties hold in every
  groupoid:
  \begin{enumerate}
  \item $\src(x)=\tgt(z)=\src(y)$ and $zx=zy$ imply $x=y$,
  \item $\tgt(x)=\src(z)=\tgt(y)$ and $xz=yz$ imply $x=y$.
  \end{enumerate}
\end{lemma}
See Appendix~\ref{A:proofs} for proofs. The cancellation properties
correspond to the fact that every morphism in a groupoid is both an epi
and a mono.

\begin{lemma}\label{lemma:lr-gpd-jt}
  In every groupoid,
  \begin{enumerate}
  \item $x \mathop{\hat{\odot}} x^\inv \mathop{\hat{\odot}} x=x$,
  \item $\tgt(x)=\src(y)$ implies $x^\inv \mathop{\hat{\odot}}
    x=y^\inv \mathop{\hat{\odot}}  y$,
   \item $(x\mathop{\hat{\odot}} y)^\inv = y^\inv\mathop{\hat{\odot}} x^\inv$.
  \end{enumerate}
\end{lemma}
The last identity follows from Lemma~\ref{lemma:lr-groupoid-props}(2),
the others are straightforward.

The following example is particularly interesting in the context of
Dedekind quantales in Section~\ref{S:modular-quantales}.

\begin{example}\label{ex:pair-groupoid}
 The \emph{pair groupoid} $(X\times X,\odot,s,t,(-)^-)$ on the
    set $X$ is defined, for all $a,b,c,d\in X$,  by
    \begin{gather*}
       (a,b) \odot (c,d) =
    \begin{cases}
      \{(a,d)\} &\text{ if } b=c,\\
      \emptyset &\text{ otherwise},
    \end{cases}\\
    \src((a,b)) = (a,a),\qquad
  \tgt((a,b))= (b,b),\qquad (a,b)^\inv=(b,a).
\end{gather*}
Pair groupoids give rise to algebras of binary relations, see
Example~\ref{ex:rel-dedekind} below.
\end{example}

The final example in this section builds on the path categories from
Example~\ref{ex:free-cat}. It is relevant for modelling higher
homotopies in the context of $(n,p)$-catoids and $(\omega,p)$-catoids
and categories.

\begin{example}\label{ex:free-groupoid}
 The free groupoid generated by a digraph is
  obtained by adding formal inverses to each edge and considering
  paths up to simplification. For instance, if $a,b,c,d$ are
  appropriate edges, then $abc^{-1}cb^{-1}d=ad$, see
  also~\cite[Chapter 4]{Higgins71}.
\end{example}


\section{Modal Quantales}\label{S:modal-quantales}

The second ingredient of the convolution algebras in this article are
quantales, more specifically quantales that extend not only the
composition and unit structure of catoids, categories and groupoids,
but also their source and target structure. It has been shown
in~\cite{FahrenbergJSZ21a} that source and target maps extend to
domain and codomain maps on quantales, which further allow defining
modal diamond and box operators on them. The resulting quantales with
axioms for domain and codomain operators are therefore known as modal
quantales~\cite{FahrenbergJSZ22a}. In this section we recall their
definition and basic properties.  Most of this development has been
formalised with Isabelle~\cite{CalkS23}.

A \emph{quantale} $(Q,\le,\cdot,1)$ is a complete lattice $(Q,\le)$
with an associative composition $\cdot$, which preserves all sups in
both arguments and has a unit $1$.  See~\cite{Rosenthal90} for an
introduction. We write $\Sup$, $\sup$, $\Inf$ and $\inf$ for sups,
binary sups, infs and binary infs in a quantale, and $\bot$ and $\top$
for the smallest and greatest element, respectively. A
\emph{subidentity} of a quantale $Q$ is an element $\alpha\in Q$ such
that $\alpha\le 1$.

A quantale is \emph{distributive} if its underlying lattice is
distributive and \emph{boolean} if its underlying lattice is a boolean
algebra, in which case we write $-$ for boolean complementation.

We henceforth use greek letters for elements of quantales to
distinguish them from elements of catoids, and we often write
$\alpha\beta$ for $\alpha\cdot\beta$.

\begin{example}\label{ex:boolean-quantale}
  We need the \emph{quantale of booleans}
  $2$, which is a boolean quantale with carrier set
  $\{0,1\}$, order $0<1$, $\max$ as binary sup, $\min$ as binary inf
  and composition, $1$ as its unit and $\lambda x.\ 1-x$ as
  complementation. It allows constructing powerset quantales
  over catoids and categories, using $2$ as a value algebra.
\end{example}

A \emph{Kleene star} $(-)^\ast:Q\to Q$ can be defined on any
quantale $Q$, for $\alpha^0 =1$ and $\alpha^{i+1} = \alpha\alpha^i$,
as
\begin{equation*}
  \alpha^\ast = \Sup_{i\ge 0} \alpha^i.
\end{equation*}

A \emph{domain quantale}~\cite{FahrenbergJSZ22a} is a quantale $Q$
with an operation $\dom:Q\to Q$ such that, for all
$\alpha,\beta\in Q$,
\begin{gather*}
  \alpha\le \dom(\alpha)\alpha,\qquad
  \dom(\alpha\dom(\beta))=\dom(\alpha\beta),\qquad \dom(\alpha)\le
  1,\\
  \dom(\bot)=\bot,\qquad \dom(\alpha\lor \beta) = \dom(\alpha)\lor
  \dom(\beta).
\end{gather*}

These domain axioms are known from domain
semirings~\cite{DesharnaisS11}, see also
Section~\ref{S:globular-semiring}. We call the first axiom the
\emph{absorption axiom}. The second expresses \emph{locality} of
domain. The third is the \emph{subidentity} axiom, the fourth the
\emph{bottom} axiom and the final the \emph{(binary) sup} axiom.  Most
properties of interest translate from domain semirings to domain
quantales.  An equational absorption law $\dom(\alpha)\alpha=\alpha$
is derivable.

We define $Q_\dom = \{\alpha\in Q\mid \dom(\alpha)=\alpha\}$.  As for
catoids, $Q_\dom= \dom(Q)$ holds because $\dom\circ\dom =\dom$ is
derivable from the domain quantale axioms. Moreover, $(Q_\dom,\le)$
forms a bounded distributive lattice with $\lor$ as binary sup,
$\cdot$ as binary inf, $\bot$ as least element and $1$ as greatest
element. We call $Q_\dom$ the lattice of \emph{domain elements} or
simply \emph{domain algebra}. In a boolean quantale, $Q_\dom$ is the
set of all subidentities and hence a complete boolean algebra.

Quantales are closed under opposition, which exchanges the arguments
in compositions. A \emph{codomain quantale} $(Q,\cod)$ is then a
domain quantale $(Q^{\mathit{op}},\dom)$.  Further, a \emph{modal
  quantale}~\cite{FahrenbergJSZ22a} is a domain and codomain quantale
$(Q,\le,\cdot,1,\dom,\cod)$ that satisfies the \emph{compatibility}
axioms
\begin{equation*}
  \dom\circ \cod = \cod\qquad\text{ and }\qquad \cod\circ\dom = \dom.
\end{equation*}
These guarantee that $Q_\dom = Q_\cod$, a set which we denote by $Q_0$
by analogy to catoids.
  
\begin{example}\label{ex:catoid-quantale}
  Let $(C,\odot,\src,\tgt)$ be a local catoid. Then $\Pow C$ can be
  equipped with a modal quantale structure. The monoidal structure is
  given by the extended composition
  $\odot: \Pow C\times \Pow C \to \Pow C$ of the catoid and the set
  $C_0$. Its lattice structure is given by $\subseteq$, $\bigcup$ and
  $\bigcap$. Its domain and codomain structure is given by
  $\dom(X)=\src(X)$ and $\cod(X)=\tgt(X)$, the images of any set
  $X\subseteq C$ with respect to $s$ and $t$~\cite{FahrenbergJSZ22a}.
  As a powerset quantale, that is, a quantale on a power set, $\Pow C$
  is in fact an atomic boolean quantale. Note that locality of $C$ is
  needed for locality of $\Pow C$. This result has the following
  instances, among others.
  \begin{enumerate}
  \item Every category $C$ extends to a modal quantale on $\Pow C$.
  \item The path category over any digraph, more specifically,
    extends to a modal quantale at powerset level. Domain and
  codomain elements are sets of vertices, $1$ is the set $V$ of all
  vertices of the digraph. In the context of rewriting, such quantales
  allow reasoning about sets of rewriting paths and in particular
  shapes of rewriting diagrams.
\item Every pair groupoid on the set $X$ extends to a modal quantale
  of binary relations on $X$ with the standard relational domain and
  codomain maps $\dom(R)=\{(x,x)\mid \exists y.\ (x,y)\in R\}$ and
  $\cod(R)=\{(y,y)\mid \exists x.\ (x,y)\in R\}$, and with $\odot$
  extended to the relational composition
  $R\cdot S=\{(x,y)\mid \exists z. (x,z)\in R\land (z,y)\in S\}$.  The
  quantalic unit is the identity relation $\{(x,x)\mid x\in
  X\}$. Quantales of binary relations and similar algebras can
    be used to reason about rewrite relations and once again about
    shapes of rewrite diagrams~\cite{DoornbosBW97,Struth06}.
  \item The shuffle catoid on $\Sigma^\ast$, considered in
    Example~\ref{ex:shuffle-catoid} extends to the commutative
    quantale of shuffle languages on $\Sigma$, a standard model of
    interleaving concurrency in computing. The quantalic composition
    is the shuffle product
    $X\| Y = \bigcup \{x \| y\mid x \in X \land y \in Y\}$ of
    languages, which are subsets of $\Sigma^\ast$. Its monoidal unit
    is $\{\varepsilon\}$, the domain/codomain structure is therefore
    trivial.
  \end{enumerate}
\end{example}
For further examples of modal quantales and their underlying catoids
see Section~\ref{S:globular-convolution} and~\cite{FahrenbergJSZ22a}.

\begin{remark}\label{R:locality}
  Locality of catoids in the form
  $x\odot y \neq \emptyset \Leftrightarrow \tgt(x)\src(y)\neq\emptyset$,
  as in Lemma~\ref{lemma:msg-loc-var}, corresponds to
  \begin{equation*}
    \alpha\cdot \beta\neq \bot \Leftrightarrow \cod(\alpha)\cdot \dom(\beta)\neq \bot
  \end{equation*}
  in modal quantales. This is a consequence of locality of $\dom$ and
  $\cod$ in modal semirings~\cite{DesharnaisS11} and modal
  quantales. In modal quantales, it is even equivalent to
  locality of $\dom$ and $\cod$. Yet the more precise locality
  property
  $\tgt(x)\src(y)\neq\emptyset \Leftrightarrow \tgt(x)=\src(y)$ does not hold in all modal quantales.

  Consider for instance the path category over the digraph
  $v_1\stackrel{e_1}{\longleftarrow}
  v_2\stackrel{e_2}{\longrightarrow} v_3\stackrel{e_3}{\longrightarrow
  } v_4$ and the sets of paths $X=\{(v_2,e_1,v_1),(v_2,e_2,v_3)\}$ and
  $Y=\{(v_3,e_3,v_4)\}$. Then
  $\cod(X)=\{v_1,v_3\}\neq \{v_3\}= \dom(Y)$ whereas
  $X\cdot Y=\{(v_2,e_2,v_3,e_3,v_4)\}\neq \emptyset$.

  Nevertheless, for any local catoid $C$ with elements $a$ and $b$, 
  \begin{align*}
    \cod(\{a\}) \cap \dom(\{b\}) \neq \emptyset
    &\Leftrightarrow \{\tgt(a)\} \cap \{\src(b)\}\neq \emptyset\\
    &\Leftrightarrow \{\tgt(a)\}=\{\src(b)\}\\
    &\Leftrightarrow \cod(\{a\}) = \dom(\{b\})
\end{align*}
holds at least for the atoms $\{a\}$ and $\{b\}$ in the quantale
$\Pow C$.
\end{remark}


\section{Dedekind Quantales}\label{S:modular-quantales}

In the previous section we have seen how catoids and categories give
rise to modal quantales at powerset level. One may therefore wonder
how the inverse structure of groupoids is reflected in quantales. In a
slightly different setting, J\'onsson and Tarski have already given an
answer, extending groupoids to (powerset) relation
algebras~\cite{JonssonT52} along the lines outlined in the previous
section, so that the inverse of the groupoid corresponds to the
converse of the relation algebra. Yet this correspondence ignores the
source/target and domain/codomain structures. To translate their
results to quantales, we consider Dedekind quantales: quantales with
an involution that satisfies the Dedekind law from relation
algebra. Interestingly, domain and codomain operations can be defined
explicitly in Dedekind quantales, whereas they need to be axiomatised
in weaker kinds of quantales. For applications in higher-dimensional rewriting
along the lines of~\cite{CalkGMS20}, quantales are combined with
Dedekind quantales in Section~\ref{S:n-structures}. This yields
$(n,p)$-quantales and $(\omega,p)$-quantales, which are related to
$(n,p)$- and $(n,\omega)$-catoids via correspondence proofs.

Dedekind quantales are single-object versions of the modular
quantaloids studied by Rosenthal~\cite{Rosenthal96}, but much of the
material introduced in this section is new.  All proofs in this
section (except for Lemma~\ref{lemma:residual-law}) can be found in
Appendix~\ref{A:proofs} and our Isabelle theories~\cite{CalkS23}
(including the proofs for this Lemma).  Isabelle has also been
instrumental in finding the counterexamples in this section.

An \emph{involutive quantale}~\cite{MulveyWP92} is a quantale $Q$ with
an operation $(-)^\conv:Q\to Q$ that satisfies
  \begin{equation*}
    \alpha^{\conv\conv} = \alpha,\qquad
     (\Sup A)^\conv = \Sup\{\alpha^\conv \mid \alpha \in A\},\qquad
    (\alpha\beta)^\conv = \beta^\conv \alpha^\conv.
  \end{equation*}

  Involution thus formalises opposition within the language of
  quantales.

\begin{remark}\label{rem:backhouse}
  Replacing the first two axioms by
  $\alpha^\conv \le \beta \Leftrightarrow \alpha \le
  \beta^\conv$ yields an equivalent
  axiomatisation. Involution is therefore self-adjoint.
  \end{remark}

\begin{lemma}\label{lemma:invol-props}
   In every involutive quantale, the following properties hold:
  \begin{enumerate}
  \item $\alpha\le \beta\Rightarrow \alpha^\conv \le \beta^\conv$,
  \item $(\alpha\sup \beta)^\conv = \alpha^\conv \sup \beta^\conv$,
  \item $(\Inf A)^\conv = \Inf\{\alpha^\conv \mid \alpha \in A\}$ and
    $(\alpha\inf \beta)^\conv = \alpha^\conv \inf \beta^\conv$,
  \item $\bot^\circ =\bot$, $1^\circ = 1$ and $\top^\circ = \top$,
    \item $\alpha^\conv \inf  \beta = \bot \Leftrightarrow \alpha \inf \beta^\conv =
      \bot$,
\item $\alpha^{\ast\conv} = \alpha^{\conv\ast}$.
  \end{enumerate}
\end{lemma}

A \emph{Dedekind quantale} is an involutive quantale in which the
\emph{Dedekind law}
\begin{equation*}
  \alpha\beta\inf \gamma \le (\alpha \inf \gamma\beta^\conv)(\beta\inf \alpha^\conv \gamma).
\end{equation*}
holds. It is standard in relation algebra. Next we present an
alternative definition.

\begin{lemma}\label{lemma:modular-dedekind}
  An involutive quantale is a Dedekind quantale if and only if
  the following \emph{modular law} holds:
  \begin{equation*}
    \alpha\beta\inf \gamma \le (\alpha \inf \gamma\beta^\conv)\beta.
  \end{equation*}
\end{lemma}

The modular law is standard in relation algebra as well. 

An extensive list of properties of Dedekind quantales can be found in
our Isabelle theories. Here we only list some structurally important
ones.

\begin{lemma}\label{lemma:modular-props}
  The following properties hold in every Dedekind quantale:
  \begin{enumerate}
  \item the strong Gelfand property
    $\alpha \le \alpha\alpha^\conv \alpha$,
  \item \emph{Peirce's law} $\alpha\beta\inf \gamma^\conv = \bot \Leftrightarrow \beta\gamma \inf
    \alpha^\conv= \bot$,
  \item the \emph{Schröder laws}
    $\alpha\beta\inf \gamma = \bot \Leftrightarrow \beta \inf
    \alpha^\conv \gamma=\bot \Leftrightarrow \alpha \inf
    \gamma\beta^\conv = \bot$.
  \end{enumerate}
\end{lemma}

The strong Gelfand property (the name has been borrowed
from~\cite{PalmigianoR11}) has been used previously by Ésik and
co-workers to axiomatise relational converse in semigroups and Kleene
algebras, where infs are not
available~\cite{BloomES95,EB95}. Similarly, and for the same reason,
it appears in globular $n$-Kleene algebras and their
applications in higher-dimensional rewriting~\cite{CalkGMS20} . 
This is our main reason for
including it here and revisiting it in
Section~\ref{S:globular-ka}. Peirce's law and the Schröder laws are
standard in relation algebra. Indeed, Dedekind quantales bring us
close to relation algebras~\cite{Tarski41,HirschH02,Maddux06}; but see
our more precise comparison below.

\begin{example}\label{ex:rel-dedekind}
  Let $G$ be a groupoid. Then $\Pow G$ forms a relation algebra over
  $G$~\cite{JonssonT52}. As a powerset algebra, the underlying lattice
  of $\Pow G$ is complete (even boolean and atomic). Hence every
  groupoid extends to a Dedekind quantale in which
  $\dom(X)= G_0 \cap XX^\conv$ and $\cod(X)=G_0\cap X^\conv X$ for
  every $X\subseteq G$.  The derivation of the Dedekind law follows
  more or less that of Jónsson and Tarski for relation
  algebras~\cite{JonssonT52}. It needs neither the boolean algebra
  structure present in relation algebras nor the completeness of the
  lattice of the Dedekind quantale. We present a more general
  derivation of this law in
  Theorem~\ref{theorem:groupoid-dedekindquantale} and revisit this
  example more formally in Corollary~\ref{corollary:groupoid-lift}.

  Again there are interesting instances.
  \begin{enumerate}
  \item Each free groupoid generated by some digraph extends to a
    Dedekind quantale on sets of paths.  The converses in this
    quantale are sets of formal inverses in the groupoid.
  \item The pair groupoid on a set $X$ extends to the Dedekind
    quantale of binary relations on $\Pow (X\times X)$ with standard
    relational converse $R^\conv = \{(y,x)\mid (x,y)\in R\} $ extended
    from the inverse operation on the pair groupoid. 
  \end{enumerate}
\end{example}

Next we turn to the domain and codomain structure on Dedekind
quantales. By contrast to general quantales, it can be defined
explicitly in involutive quantales using
\begin{equation*}
  \dom(\alpha)= 1 \inf \alpha\alpha^\conv\qquad\text{ and } \qquad \cod(\alpha)= 1\inf
  \alpha^\conv \alpha,
\end{equation*}
or alternatively $\dom(\alpha)=1\inf \alpha\top$ and
$\cod(\alpha)=1\inf \top \alpha$, as in relation algebra.

But only Dedekind quantales are expressive enough to make these two
definitions coincide, derive the natural domain and codomain laws
needed for defining suitable modal operators (as in
Section~\ref{S:modalities} below) and establish the correspondence
with respect to groupoids in Sections~\ref{S:dedekind-correspondence}
and \ref{S:msg-quantale}.

\begin{proposition}\label{prop:modular-modal}
  Every Dedekind quantale is a modal quantale.
\end{proposition}

A proof using the explicit definitions of $\dom$ and $\cod$ can be
found in Appendix~\ref{A:proofs}.  In addition, we list some
properties that are not available in modal quantales.

\begin{lemma}\label{lemma:quantale-dom-props}
  In every Dedekind quantale, the following properties hold:
  \begin{enumerate}
  \item $\dom(\alpha)=1\inf \alpha\top$,
  \item $\dom(\alpha)\top = \alpha\top$,
  \item $(\dom(\alpha))^\conv = \dom(\alpha)$,
  \item  $\dom(\alpha^\conv) = \cod(\alpha)$. 
  \end{enumerate}
\end{lemma}

Next we present a natural example of a modal quantale that is not
Dedekind.

\begin{example}\label{ex:path}
  The law (1) fails in the modal path quantales from
  Example~\ref{ex:catoid-quantale}, where formal inverses are not
  assumed in the underlying catoid. Recall that, in this model, $1$
  is the set $V$ of all vertices of the digraph. Thus
  $V\cap P\top =\emptyset$ unless the set $P$ of paths contains a path
  of length one and $\dom(P)=\emptyset$ if and only if
  $P=\emptyset$. 
\end{example}

The next lemma shows more directly that Dedekind quantales are more
expressive than involutive quantales

\begin{lemma}\label{lemma:invol-props-counter}
  Neither the strong Gelfand property nor the modular law holds in all
  involutive quantales.
\end{lemma}
\begin{proof}
  In the involutive quantale $\bot < a < \top = 1$ with multiplication
  $aa= \bot$ (the rest being fixed) and $(-)^\conv = \mathit{id}$, the
  strong Gelfand property fails because $aa^\conv a =\bot < a$. The
  modular and (Dedekind law) fail in this involutive quantale because
  \begin{equation*}
    (1 \inf aa)a (1 \inf aa)(a \inf 1a)= 0 <1a\inf a = a.\qedhere
    \end{equation*}
\end{proof}
The following example refines this result. 
\begin{example}
  Adding the strong Gelfand property to the involutive quantale axioms
  does not imply the modular law. The involutive quantale defined by
  
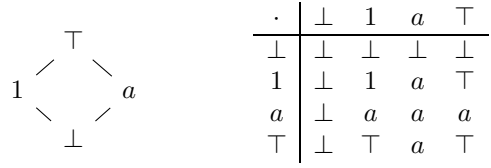
\begin{equation*}
    \begin{tikzcd}[column sep = tiny, row sep = tiny]
      & \top &\\
      1\arrow[ur,-] && a\arrow[ul,-]\\
      &\bot\arrow[ul,-]\arrow[ur,-] &
    \end{tikzcd}
    \qquad\qquad
    \begin{array}{c|cccc}
      \cdot & \bot  & 1 & a & \top\\
      \hline
      \bot & \bot & \bot &\bot  &\bot\\
      1 & \bot & 1 &a &\top \\
      a & \bot & a&a & a\\
      \top & \bot & \top & a & \top
    \end{array}
  \end{equation*}
  and $(-)^\conv = \mathit{id}$ satisfies the strong Gelfand property,
  but
  \begin{equation*}
    1a\inf \top = a> \bot = (1 \inf \top a)a. 
  \end{equation*}
\end{example}

The domain and codomain operations of modal semirings, even finite
ones, need not be uniquely determined~\cite[Lemma
6.4]{DesharnaisS11}. A modal semiring can therefore carry several
domain/codomain structures. Yet they are uniquely determined in
modal semirings over boolean algebras. Finite modal semirings and
finite modal quantales are the same. One may therefore ask whether
there can be other domain/codomain structures on Dedekind quantales
than that given by the explicit definitions above. The answer is the
same as for modal semirings.

For the sake of this argument, we call \emph{modal Dedekind
  quantale} a Dedekind quantale that is also a modal quantale, that
is, it is equipped with a map $\delta^-$ (domain) and a map
$\delta^+$ (codomain) that satisfy axioms from
Section~\ref{S:modal-quantales}. We start with a technical lemma.

\begin{lemma}\label{lemma:dom-conv}
  In every modal Dedekind quantale,
  \begin{equation*}
    \delta^-(\alpha)^\conv =\delta^-(\alpha)\qquad\text{
      and }\qquad \delta^+(\alpha)^\conv = \delta^+(\alpha).
    \end{equation*}
\end{lemma}

\begin{lemma}\label{lemma:dom-not-unique}
 There is a modal distributive Dedekind quantale
 in which $\delta^-\neq \dom$ and $\delta^+\neq \cod$. 
\end{lemma}
\begin{proof}
  In the modal distributive
  Dedekind quantale with $\bot < a < \top = 1$, multiplication $aa=a$,
  $\delta^-(a)=1=\delta^+(a)$ and $(-)^\conv = \id$, we have
  $\delta^-(a) = \delta^+(a) =\top \neq a = 1 \inf aa^\conv = 1 \inf a^\conv
  a$.
\end{proof}

\begin{remark}
  In any boolean modal quantale $Q$, the set $Q_0$ equals the
  boolean subalgebra of all subidentities, and domain and codomain are
  uniquely defined.  The proof for modal semirings~\cite[proof of
  Theorem 6.12]{DesharnaisS11} translates directly. Thus, in any
  boolean modal Dedekind quantale,
  $\delta^-(\alpha) = 1 \inf \alpha\alpha^\conv= \dom(x)$ and
  $\delta^+(\alpha)= 1 \inf \alpha^\conv \alpha= \cod(x)$.
\end{remark}

Boolean Dedekind quantales are strongly related to relation algebras.

\begin{lemma}\label{lemma:residual-law}
  In any boolean Dedekind quantale, $(-\alpha)^\conv = -(\alpha^\conv)$ and the
  \emph{residual law}  $\alpha^\conv -(\alpha\beta) \le -\beta$ is derivable.
\end{lemma}
A residual law appears in Tarski's original axiomatisation of relation
algebra~\cite{Tarski41}. A boolean Dedekind quantale is thus a
relation algebra over a complete lattice, and a relation algebra a
boolean Dedekind quantale in which only finitary sups and infs are
required to exist. Relation algebras are formed over boolean algebras
that need not be complete.

Finally we relate the explicit definition of $\dom$ and $\cod$ in
Dedekind quantales with a definition previously used in higher-dimensional
rewriting in a Kleene-algebraic structure where infs are not
available.

\begin{remark}
  The conditions $\dom(\alpha)\le \alpha\alpha^\conv$ and
  $\cod(\alpha)\le \alpha^\conv \alpha$ have been used for higher-dimensional
  rewriting with globular $2$-Kleene algebras~\cite{CalkGM21}, see
  also Section~\ref{S:globular-ka}.  They are consequences of the
  explicit definition of domain and codomain in Dedekind quantales. In
  involutive modal quantales, each of them implies the strong Gelfand
  property, $\dom(\alpha)\le \alpha\alpha^\conv$ is equivalent to
  $\dom(\alpha)= 1 \inf \alpha\alpha^\conv$ and
  $\cod(\alpha) = 1 \inf \alpha^\conv \alpha$ is equivalent to
  $\cod(\alpha)\le \alpha^\conv \alpha$. Yet none of these laws need
  to hold in such quantales: in the modal distributive Dedekind
  quantale used in the proof of Lemma~\ref{lemma:dom-not-unique},
  \begin{equation*}
    dom(a)=\cod(a)=\top > a = aa^\conv = a^\conv a.
  \end{equation*}
\end{remark}


\section{Higher catoids}\label{S:2-lr-msg}

In this section we present our first conceptual contribution: axioms
for $n$-catoids and $\omega$-catoids that generalise definitions of
strict $n$-categories and $\omega$-categories. These are the
structures from which we develop axioms for higher quantales in
Section~\ref{S:2-quantales}, using the proofs in
Section~\ref{S:globular-convolution}. Mac Lane has outlined
axiomatisations of single-set $2$-categories, $n$-categories and
$\omega$-categories, imposing a $2$-category structure on each pair of
single-set categories $C_i$ and $C_j$ for
$0\le i<j < \omega$~{\cite[Chapter XII]{MacLane98}}.  Here,
$\omega$-category means strict globular $\infty$-category. Similar
single-set approaches appear, for instance,
in~\cite{BrownH81,Street87,Steiner04}. We adapt MacLane's axioms to
catoids. We start from a uniform axiomatisation that includes the case
of $n$ or $\omega$, but then focus mainly on
$\omega$-catoids. As previously, most of the material in this
  section has been formalised with Isabelle~\cite{CalkS24}, and
  Isabelle has been instrumental in analysing and reducing this
  axiomatisation.

For an ordinal $\alpha\in \{0,1,\dots,\omega\}$, an \emph{$\alpha$-catoid}
is a structure $(C,\odot_i,\src_i,\tgt_i)_{0\le i<\alpha}$ such that
each $(C,\odot_i,\src_i,\tgt_i)$ is a catoid and these structures
interact as follows:
\begin{itemize}
\item for all $i\neq j$, 
\begin{gather*}
  \src_i\circ \src_j = \src_j\circ \src_i,\qquad
  \src_i\circ \tgt_j = \tgt_j\circ \src_i,\qquad
  \tgt_i\circ \tgt_j = \tgt_j\circ \tgt_i,\\
   \src_i(x\odot_j y) \subseteq \src_i(x) \odot_j \src_i(y),\qquad
   \tgt_i(x\odot_j y) \subseteq \tgt_i(x) \odot_j \tgt_i(y),
\end{gather*}
\item and for all $i<j$, 
\begin{gather*}
  (w \odot_j x)\odot_i (y\odot_j z) \subseteq (w\odot_i y)\odot_j (x\odot_i z),\\
  \src_j\circ \src_i = \src_i,\qquad \src_j\circ \tgt_i = \tgt_i,\qquad
  \tgt_j\circ \src_i = \src_i,\qquad \tgt_j\circ \tgt_i = \tgt_i.
\end{gather*}
\end{itemize}

An \emph{$\alpha$-category} is a local functional $\alpha$-catoid,
that is, each $(C,\odot_i,\src_i,\tgt_i)$ is local and
functional.

As the $(C,\odot_i,\src_i,\tgt_i)$ are catoids,
$\src_i(C)=\tgt_i(C)$ for each $i<\alpha$ by
Lemma~\ref{lemma:fix-id} . We write $C_i$ for this set of \emph{$i$-cells} of $C$.
We also write $\Delta_i$
for the domain of definition of $\odot_i$ and refer to the source and
target cells of cells as (\emph{lower} and \emph{upper})
\emph{faces}.

The axioms after the first bullet point impose that source and target
maps at each dimension $i$ are catoid (endo)morphisms of the catoid
$(C,\odot_j,\src_j,\tgt_j)$ at each dimension $j\neq i$.  In the local
functional case of $\alpha$-categories, these morphisms become
functors, as expected.  (In the $\omega$-category of globular sets,
the axioms in the first line are known as \emph{globular laws}.)

The \emph{interchange} axioms in the first line of the second bullet
point ensure that $\odot_i$ is a catoid
bi-morphism with respect to $\odot_j$, for all $i<j$.

The \emph{whisker} axioms in the second line of this bullet point,
together with the catoid laws in Lemma~\ref{lemma:mm-props}(1), imply
that $\src_ j(x) = \tgt_j(x) =x$ for all $x\in C_i$, and thus
$x\in C_j$ for all $j\ge i$. Lower dimensional cells thus remain
units in higher dimensions and
\begin{equation*}
  C_0\subseteq C_1\subseteq C_2\subseteq \dots \subseteq C.
\end{equation*}
We may thus regard $i$-cells as degenerate cells or whiskers in which
sources and targets at each dimension greater than $i-1$
coincide. Further, Lemma~\ref{lemma:orth-id} implies that all lower
dimensional cells are orthogonal idempotents with respect to higher
compositions.  For all $i,j\le k$,
\begin{equation*}
    \src_i(x) \odot_k \src_j(y) =
    \begin{cases}
      \{\src_i(x)\} & \text{ if } \src_i(x)=\src_j(y),\\
      \emptyset & \text{ otherwise}.
    \end{cases}
\end{equation*}
All higher-dimensional compositions of lower dimensional cells are
therefore trivial.  The structure of units across dimensions can be
seen in Example~\ref{ex:local-2-catoid-not-strong} below.

We now turn to $\omega$-catoids and mention $n$-catoids only
occasionally.  In this context, one often adds an axiom guaranteeing
that for all $x\in C$ there exists and $i, \omega$ such that
$s_i(x)=x=t_i(x)$, that is all cells in $C$ have finite dimension. In
the special case of strict $\omega$-categories, for instance,
Steiner~\cite{Steiner04}, adds such an axiom, whereas
Street~\cite{Street87} and Mac Lane~\cite{MacLane98} do not feature
this condition and Brown and Higgins mention both
variants~\cite{BrownH81}. Our $\omega$-categories, as local functional
$\omega$-catoids, are therefore the same as $\omega$-categories in the
sense of Street and Mac Lane.  We need the above condition in the
following construction, but not beyond this paragraph.  Because of the
whisker axioms and the catoid axioms, the chain
$C_0\hookrightarrow C_1\hookrightarrow C_2\hookrightarrow\dots
\hookrightarrow C$ is a filtration of $\omega$-catoids (when
$\alpha=\omega$). Catoid $C$ is the (co)limit of an increasing chain
of sub-$\omega$-catoids $C_n$, which themselves are $n$-catoids.
See~\cite[Proposition 2.3]{Steiner04} for a related discussion on
$\omega$-categories.

$\omega$-Catoids, form a category in several ways.  Their morphisms
are catoid morphisms that preserve all source and target maps $\src_i$
and $\tgt_i$ and all compositions $\odot_i$ at each dimension
$i<\omega$. These kind of categories are also categories with respect
to bounded morphisms.

As in the one-dimensional case, $\omega$-catoids and $n$-catoids
generalise $\omega$-categories and $n$-categories to multioperational
compositions, and beyond functionality and locality.  Yet the
underlying $\omega$-graphs or $n$-graphs, where compositions are
forgotten, are the same for both.  The globular cell shape of strict
higher categories is therefore present in the corresponding catoids,
too.

\begin{lemma}\label{lemma:2-msg-props}
  In every $\omega$-catoid, the \emph{globular laws} hold. 
  For all $0\le i<j< \omega$,
\begin{equation*}
     \src_i\circ \src_j = \src_i,\qquad \src_i\circ \tgt_j =
     \src_i,\qquad \tgt_i\circ \tgt_j = \tgt_i,\qquad \tgt_i\circ \src_j= \tgt_i.
   \end{equation*}
 \end{lemma}
 The proofs are immediate from the globular and whisker axioms.

 \begin{example}\label{ex:globular-shape}
   The globular cell shape of $\omega$-catoids can be visualised,
   in two dimensions, as
 \begin{equation*}
 \begin{tikzcd}
  \src_0(x) \arrow[r, bend left, "\src_1(x)", ""{name=U,inner sep=1pt,below}]
  \arrow[r, bend right, "\tgt_1(x)"{below}, ""{name=D,inner sep=1pt}]
  & \tgt_0(x)
  \arrow[Rightarrow, from=U, to=D, "x"]
\end{tikzcd}
\end{equation*}
The relationships between $0$-cells and $1$-cells in this diagram can
be calculated using Lemma~\ref{lemma:2-msg-props}.
\end{example}

In light of Example~\ref{ex:bounded-morphism}, the morphism laws of
$\omega$-categories are rather strong.
 \begin{lemma}\label{lemma:omega-category-strong}
   In every $\omega$-category, for $0\le i< j<\omega$, the following
   \emph{strong morphism} laws hold: 
   \begin{equation*}
     \src_j(x\odot_i y) = \src_j(x) \odot_i \src_j(y) \qquad\text{ and }\qquad
  \tgt_j(x\odot_i y) =\tgt_j(x) \odot_i \tgt_j(y).
  \end{equation*}
 \end{lemma}
 \begin{proof}
   First we derive the law for $\src_j$. Suppose the right-hand side
   is equal to $\emptyset$.  Then
   $\src_j(x\odot_i y) = \src_j(x) \odot_i \src_j(y)$, using the
   morphism axiom
   $\src_j(x\odot_i y)\subseteq \src_j(x)\odot_i
   \src_i(y)$. Otherwise, if $\Delta_i(\src_j(x),\src_j(y))$, then
   $\tgt_i(\src_j(x))= \src_i(\src_j(x))$ and therefore
   $\tgt_i(x)=\src_i(y)$ by Lemma~\ref{lemma:2-msg-props}. Locality
   then implies that $\Delta_i(x,y)$, and
   $\src_j(x\odot_i y) = \src_j(x) \odot_i \src_j(y)$ follows from the
   morphism axiom for $\src_j$ and functionality. The proof for
   $\tgt_j$ follows by opposition.
 \end{proof}
 By contrast to the other morphism axioms, which require that if
 left-hand sides are defined, then so are right-hand sides, the strong
 morphism laws state that one side is defined if and only if the
other is. Accordingly, an $\omega$-catoid is \emph{strong} if it
satisfies the strong morphism laws.

\begin{corollary}\label{corollary:omega-category-strong}
  In every $\omega$-category, $\Delta_i(x,y) \Leftrightarrow
  \Delta_i(\src_j(x),\src_j(y)) \Leftrightarrow
  \Delta_i(\tgt_j(x),\tgt_j(y))$ for all $0\le i<j <\omega$. 
\end{corollary}

\begin{example}\label{ex:strong-morphism}
  The strong morphism laws of $\omega$-categories can be explained
  diagrammatically for $2$-categories and in particular in $\Cat$,
  the category of all small categories, 
  where $\odot_1$ is the vertical composition of $2$-cells or natural
  transformations, and $\odot_0$ the horizontal
  one. Lemma~\ref{lemma:omega-category-strong},
  Corollary~\ref{corollary:omega-category-strong} and locality imply
  that the horizontal composition $x\odot_0 y$ is defined if and only
  if the horizontal compositions $\src_1(x)\odot_0 \src_1(y)$ and
  $\tgt_1(x)\odot_0 \tgt_1(y)$ are defined. In all cases, this means
  that $\tgt_0(x)=\src_0(y)$.
  \begin{equation*} 
    \begin{tikzcd}
      \src_0(x) \arrow[r, bend left, "\src_1(x)", ""{name=U,inner sep=1pt,below}]
      \arrow[r, bend right, "\tgt_1(x)"{below}, ""{name=D,inner sep=1pt}] 
      & \tgt_0(x) \arrow[Rightarrow, from=U, to=D, "x"] \arrow[r, bend left, "\src_1(y)", ""{name=UU,inner sep=1pt,below}]
      \arrow[r, bend right, "\tgt_1(y)"{below}, ""{name=DD,inner sep=1pt}]
      & \tgt_0(y) \arrow[Rightarrow, from=UU, to=DD, "y"]
    \end{tikzcd}
  \end{equation*}
  The morphism axioms
  $\src_1(x\odot_0 y) \subseteq \src_1(x)\odot_0\src_1(y)$ and
  $\tgt_1(x\odot_0 y) \subseteq \tgt_1(x)\odot_0\tgt_1(y)$ then become
  strong because the compositions $x\odot_0 y$,
  $\src_1(x)\odot_0\src_1(y)$ and $\tgt_1(x)\odot_0\tgt_1(y)$ are
  functional, each yields at most a single cell. The upper and lower
  faces $\src_1(x\odot_0 y)$ and $\tgt_1(x\odot_0 y)$ of $x\odot_0 y$
  must thus be equal to $\src_1(x)\odot_0\src_1(y)$ and
  $\tgt_1(x)\odot_0\tgt_1(y)$, respectively.
\end{example}

The next two examples explain the weak morphism laws in the absence of
locality or functionality.
 
\begin{example}\label{ex:local-2-catoid-not-strong}
  In the local $2$-catoid on the set $\{a,b,c\}$ with
  \begin{equation*}
    \begin{array}{c|cccc}
      & \src_0 & \tgt_0 & \src_1 & \tgt_1\\
      \hline
      a & b & b & a & a\\
      b & b & b & b & b\\
      c & b & b &  a & a
    \end{array}
    \qquad\qquad
    \begin{array}{c|ccc}
      \odot_0 & a & b & c\\
      \hline
      a & \{a,b\} & \{a\} & \{c\}\\
      b& \{a\} & \{b\} &\{c\}\\
      c & \{c\} &\{c\} & \{c\}
    \end{array}
    \qquad\qquad
    \begin{array}{c|ccc}
      \odot_1 & a & b & c\\
      \hline
      a & \{a\} & \emptyset & \{c\}\\
      b& \emptyset & \{b\} &\emptyset\\
      c & \{c\} &\emptyset & \{c\}
    \end{array}
  \end{equation*}
  we have
  $\src_1(a \odot_0 c) = \src_1(\{c\})= \{a\}\subset \{a,b\} =
  \{a\}\odot_0\{a\} = \src_1(a)\odot_0 \src_1(a)$ and the same
  inequality holds for $\tgt_1$ because $\src_1=\tgt_1$.  In this
  example, $c$ is a non-degenerate $2$-cell with $1$-faces $a$
  and $0$-faces $b$, while $a$ is a non-degenerate $1$-cell with
  $0$-faces $b$ and $b$ is a $0$-cell. So the strong morphism
  laws fail because $\odot_0$ may map to more than one cell. This
  degenerate situation can be depicted as
  
  
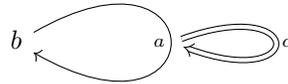
\vspace{-1cm}
  \begin{equation*}
    \begin{tikzcd}
      b \arrow[loop, out=30, in=-30, distance=8em, swap, "a"] &&
      \arrow[loop, Rightarrow, out=20, in=-20, distance=5em,
      "c"]
    \end{tikzcd}
  \end{equation*}
  \vspace{-1cm}
  
  \noindent This $2$-catoid is functional with respect to
  $\odot_1$, but not with respect to $\odot_0$.
  
  It is also worth considering the unit structure given by the
  source and target maps and the whisker axioms, and its
  effect on $\odot_0$ and $\odot_1$. They determine all
  compositions except $a\odot_0 a$, $c\odot_0 c$ and $c\odot_1
  c$. The composition $\odot_1$ is trivial because of the
  whisker axioms.
\end{example}

\begin{example}\label{ex:functional-2-catoid-not-strong}
  In the functional $2$-catoid on the set $\{a,b\}$ with
  \begin{equation*}
    \begin{array}{c|cccc}
      & \src_0 &\tgt_0 & \src_1 & \tgt_1\\
      \hline
      a & b & b & b &b\\
      b & b & b & b & b
    \end{array}
    \qquad\qquad
    \begin{array}{c|ccc}
      \odot_0 & a & b\\
      \hline
      a & \emptyset& \{a\}\\
      b& \{a\} & \{b\} 
    \end{array}
    \qquad\qquad
    \begin{array}{c|ccc}
      \odot_1 & a & b\\
      \hline
      a & \{a\} & \{a\}\\
      b& \{a\} & \{b\}
    \end{array}
  \end{equation*}
  we have
  $\src_1(a\odot_0 a) = \src_1(\emptyset) = \emptyset \subset \{b\}=
  \{b\}\odot_0 \{b\} = \src_1(a)\odot\src_1(a)$. The corresponding
  inequality holds for $\tgt_1$ because $\src_1=\tgt_1$. This example
  is a monoid as a category with respect to $\odot_1$, and a ``broken
  monoid'', hence simply a graph, with respect to $\odot_0$, as
  $a\odot_0 a$ is undefined. Now, the strong morphism laws fail
  because the broken monoid is not local: $\src_0(a)=b=\tgt_0(a)$,
  but $a\odot_0 a=\emptyset$. So $x\odot_0y$ and therefore
  $\src_1(x\odot_0 y)$ may be $\emptyset$, whereas
  $\src_1(x)\odot_0\src_1(y)$ or $\tgt_1(x)\odot_0\tgt_1(y)$ are not.
\end{example}

Next we explain the weakness of the remaining morphism axioms.
First, we need a lemma.
\begin{lemma}\label{lemma:omega-cat-props}
  In every $\omega$-catoid, for all $0\le i < j <\omega$, if
  $\Delta_j(x,y)$, then
  \begin{equation*}
    \src_i(x\odot_j y) = \{\src_i(x)\}=\{\src_i(y)\}\qquad\text{ and }\qquad
    \tgt_i(x\odot_j y) = \{\tgt_i(x)\}=\{\tgt_i(y)\}. 
  \end{equation*}
\end{lemma}
A proof can be found in Appendix~\ref{A:proofs}.

\begin{example}\label{ex:hom2-weak}
  In $2$-catoids, for instance, weakness of
  $\src_0(x\odot_1 y)\subseteq \src_0(x)\odot_1\src_0(y)$ and
  $\tgt_0(x\odot_1 y)\subseteq \tgt_0(x)\odot_1\tgt_0(y)$ allows
  $x\odot_1 y$ to be undefined while both $\src_0(x)\odot_1\src_0(y)$ or
  $\tgt_0(x)\odot_1\tgt_0(y)$ are defined. The last two compositions
  are defined if $\src_0(x)$ is equal to $s_0(y)$, and $\tgt_0(x)$
  is equal to $\tgt_0(y)$, respectively. For the first one,
  $\tgt_1(x)$ must be equal to $\src_1(y)$, from which
  $\src_0(x)=\src_0(\tgt_1(x))=\src_0(\src_1(y))= s_0(y)$ and
  $\tgt_0(x)=\tgt_0(\tgt_1(x))= \tgt_0(\src_1(y))=\tgt_0(y)$ follow
  using the globular laws in Lemma~\ref{lemma:2-msg-props}. The left
  diagram below shows a situation where
  $\src_0(x\odot_1 y)\subset \src_0(x)\odot_1\src_0(y)$ because
  $\src_0(x)=\src_0(y)$ whereas $\tgt_1(x)\neq \src_1(y)$. The right
  diagram shows the opposite situation were
  $\tgt_0(x\odot_1 y)\subset\tgt_0(x)\odot_1\tgt_0(y)$ because
  $\tgt_0(x)=\tgt_0(y)$ whereas $\tgt_1(x)\neq \src_1(y)$. The
  middle diagram shows a situation where both sides are defined
  because $\tgt_1(x)= \src_1(y)$. The globular structure is imposed
  by Lemma~\ref{lemma:omega-cat-props}.
  
  \begin{equation*}
    \begin{tikzcd}
      & \tgt_0(x)\\
      \src_0(x) \arrow[ur, bend left=30, "\src_1(x)", ""{name=Ux, inner
        sep=3pt, below}]\arrow[ur, bend right=30, "\tgt_1(x)"{swap},
      ""{name=Dx, inner sep=2pt}]\arrow[dr, bend
      left=30, "\src_1(y)", ""{name=Uy, inner sep=3pt, below}]\arrow[dr, bend
      right=30, "\tgt_1(y)"{swap}, ""{name=Dy, inner
        sep=2pt}]\arrow[Rightarrow, from=Ux,to=Dx, "x"]\arrow[Rightarrow, from=Uy,to=Dy, "y"{swap}]&\\
      & \tgt_0(y)
    \end{tikzcd}
    \qquad
    \begin{tikzcd}
      \src_0(x) \arrow[r, bend left=70, "\src_1(x)", ""{name=U,inner
        sep=3pt,below}]
      \arrow[r, ""{name=M, inner sep=3pt}]\arrow[r, ""{name=M1, inner
        sep=3pt, below}]
      \arrow[r, bend right=70, "\tgt_1(y)"{swap}, ""{name=D,inner sep=3pt}] 
      & \tgt_0(y) \arrow[Rightarrow, from=U, to=M, "x"] \arrow[Rightarrow,
      from=M1, to=D, "y"] 
    \end{tikzcd}
    \qquad
    \begin{tikzcd}
      \src_0(x) \arrow[dr, bend left=30,  "\src_1(x)", ""{name=Ux,inner
        sep=3pt,below}]\arrow[dr, bend
      right=30,  "\tgt_1(x)"{swap},  ""{name=Dx, inner sep=2pt}]&\\
      & \tgt_0(x) \\
      \src_0(y) \arrow[ur, bend left=30,  "\src_1(y)", ""{name=Uy, inner sep=3pt, below}]\arrow[ur, bend
      right=30, "\tgt_1(y)"{swap}, ""{name=Dy, inner
        sep=2pt}]
      \arrow[Rightarrow, from=Ux,to=Dx, "x" {swap}]\arrow[Rightarrow, from=Uy,to=Dy, "y"]
    \end{tikzcd}
  \end{equation*}
  The weak morphism axioms are thus consistent with vertical
  compositions of globes.  
\end{example}

\begin{example}\label{ex:interchange-weak}
  In $2$-catoids, weakness of
  $(w\odot_1 x)\odot_0 (y\odot_1 z)\subseteq (w\odot_0 y)\odot_1
  (x\odot_0 z)$ allows the right-hand side of the interchange axiom
  to be defined while the left-hand side is undefined. The
  right-hand side is defined if $\tgt_1(w\odot_0 y)$ is equal to
  $\src_1(x\odot_0 z)$ and both sets are nonempty. The globular laws
  in Lemma~\ref{lemma:2-msg-props} and Lemma~\ref{lemma:msg-props1}
  then imply that
  $\src_0(w) = \src_0(w\odot_0 y)= \src_0(\tgt_1(w\odot_0 y))=
  \src_0(\src_1(x\odot_0 z)) = \src_0(x\odot_0 z)= \src_0(x)$, and
  $\tgt_0(y)=\tgt_0(z)$ holds for similar reasons.  The left-hand
  side is defined if $\tgt_0(w\odot_1 x)$ is equal to
  $\src_0(y \odot_1 z)$ and both sets are nonempty. Therefore
  $\src_0(w)=\src_0(w\odot_1 x)= \src_0(x)$ and
  $\tgt_0(y)=\tgt_0(y\odot_1 z)= \tgt_0(z)$, but also
  $\tgt_0(w)=\tgt_0(w\odot_1 x)= \tgt_0(x)$ and
  $\src_0(y)=\src_0(y\odot_1 z)= \src_0(y)$ by
  Lemma~\ref{lemma:omega-cat-props}. Thus in particular
  $\tgt_0(w)=\tgt_0(x)=\src_0(y)=\src_0(z)$. The right diagram below
  shows a situation where
  $(w\odot_1 x)\odot_0 (y\odot_1 z)\subset (w\odot_0 y)\odot_1
  (x\odot_0 z)$, because all compositions on the right are defined,
  but $\tgt_0(w)=\src_0(y)\neq \tgt_0(x)=\src_0(z)$. The left
  diagram below shows a situation where both sides are defined.
  \begin{equation*} 
    \begin{tikzcd}
      \src_0(x) \arrow[r, bend left=70, "\src_1(w)", ""{name=U,inner
        sep=3pt, below}]
      \arrow[r, ""{name=M, inner sep=3pt}]\arrow[r, ""{name=M1, inner
        sep=3pt, below}]
      \arrow[r, bend right=70, "\tgt_1(x)"{below}, ""{name=D,inner sep=3pt}] 
      & \tgt_0(x) \arrow[Rightarrow, from=U, to=M, "w"] \arrow[Rightarrow,
      from=M1, to=D, "x"] \arrow[r, bend left=70, "s_1(y)",
      ""{name=UU,inner sep=3pt,below}]
      \arrow[r, ""{name=MM,inner sep=3pt}] \arrow[r, ""{name=MM1,inner sep=3pt,below}]
      \arrow[r, bend right=70, "\tgt_1(z)"{below}, ""{name=DD,inner sep=3pt}]
      & \tgt_0(y)
      \arrow[Rightarrow, from=UU, to=MM, "y"]\arrow[Rightarrow, from=MM1, to=DD, "z"]
    \end{tikzcd}
    \qquad
    \begin{tikzcd}
      \src_0(x)\arrow[r, ""{name=Ux, inner sep=3pt,below}]\arrow[rr,bend
      left=40, "\src_1(w)", ""{name=Uw, inner sep=3pt, below}]\arrow[r, bend
      right=70, "\tgt_1(x)"{below}, 
      ""{name=Dx, inner sep=3pt}] \arrow[Rightarrow,
      from=Ux, to=Dx, "x"] &
      \tgt_0(x) \arrow[r]\arrow[rr, bend right=40, "\tgt_1(z)"{below}, ""{name=Dz, inner sep=3pt}] \arrow[Rightarrow, from=Uw, "w"] & \src_0(y) \arrow[r,
      ""{name=Dy, inner sep =3pt}]\arrow[r,
      bend left=70, "\src_1(y)", ""{name=Uy, inner sep=3pt, below}]\arrow[Rightarrow,
      from=Uy, to=Dy, "y"]\arrow[Rightarrow, to=Dz, "z"]& \tgt_0(y)
    \end{tikzcd}
  \end{equation*}
  The interchange axiom is thus consistent with horizontal and
  vertical compositions of globes. The difference to the standard
  equational interchange laws of category theory is that, using
  multioperations, we express partiality by mapping to the empty set.
\end{example}

The following example confirms that the interchange axiom and the
morphism axioms for $\src_0$ and $\tgt_0$ remain inclusions in
$\omega$-categories.

\begin{example}\label{ex:cat-inclusions}
 Consider the  $2$-category with $X=\{a,b\}$ and
 \begin{equation*}\label{ex:interchange-eq}
     \begin{array}{c|cccc}
     & \src_0 & \tgt_0 &\src_1 & \tgt_1\\
     \hline
     a & b & b & a & a\\
     b& b & b& b& b
     \end{array}
     \qquad\qquad
  \begin{array}{c|cc}
    \odot_0 & a & b\\
    \hline
    a &\{b\}&\{a\}\\
    b & \{a\}&\{b\}
  \end{array}
\qquad\qquad
   \begin{array}{c|cc}
    \odot_1 & a & b\\
    \hline
    a &\{a\}&\emptyset\\
    b & \emptyset &\{b\}
   \end{array}
 \end{equation*}
 It is actually a monoid as a category with $1$-cell $a$, $0$-cell
 $b$, and composition $a\odot_0a=b$, where $b$ is seen as a unit
 arrow.  Further,  $\odot_1$ is trivial because of the
 whisker axioms. Because of this,
 $(b\odot_1 a) \odot_0 (b\odot_1 a) = \emptyset \subset \{b\} =
 (b\odot_0 b)\odot_1(a\odot_0a)$ as well as
 $\src_0(a\odot_1 b) = \emptyset \subset \{b\} =
 \src_0(a)\odot_1\src_0(b)$ and
 $\tgt_0(a\odot_1 b) = \emptyset \subset \{b\} =
 \tgt_0(a)\odot_1\tgt_0(b)$.
\end{example}

\begin{remark}
  The inclusions in the morphism axioms for $\src_i$ and $\tgt_i$
  cannot be strengthened to equations. We show in
  Appendix~\ref{A:eckmann-hilton} that this would collapse the entire
  structure. A similar collapse happens with an equational interchange
  law.
\end{remark}

The $\alpha$-catoid axioms contain redundancy. We have used Isabelle's
SAT-solvers and automated theorem provers to analyse them. For
irredundancy of a formula $\varphi$ with respect to a set $X$ of
formulas, we ask the SAT-solvers for a model of
$X\cup\{\neg \varphi\}$. For redundancy, we ask the automated theorem
provers for a proof of $X\vdash\varphi$. This proofs-and-refutations
game often succeeds in practice. Because of the set-up of
$\omega$-catoids as pairs of $2$-catoids, an analysis of $2$-catoids
suffices.

\begin{proposition}\label{prop:reduced-props}
  The following $\alpha$-catoid axioms are irredundant and imply the
  other $\alpha$-catoid axioms from the beginning of this
  section. For all $0\le i < j \le \alpha$,
  \begin{gather*}
    \src_j(x\odot_i y)\subseteq \src_j(x) \odot_i \src_j(y),\qquad
    \tgt_j(x\odot_i
    y) \subseteq \tgt_j(x) \odot_i \tgt_j(y),\\
    (w \odot_j x)\odot_i (y\odot_j z) \subseteq (w\odot_i y)\odot_j
    (x\odot_i z).
  \end{gather*}
\end{proposition}

A proof can be found in Appendix~\ref{A:proofs}. This reduction is
convenient for relating structures, and it streamlines our
correspondence proofs below. More generally, the single-set approach
makes $n$-categories accessible to SMT-solvers and first-order
automated theorem provers, using $\Delta$ and $\hat{\odot}$ in
specifications like in Section~\ref{S:lr-multisemigroups}.

\begin{remark}
We cannot replace the morphism laws for $\src_j$ and $\tgt_j$ by
  those for $\src_i$ and $\tgt_i$ in the reduced axiomatisation for
  $\alpha$-catoids. Otherwise the first morphism laws would no longer
  be derivable: Isabelle produces counterexamples. Likewise, if we use
  only the whisker axioms and the morphism axioms in the first line of
  the non-reduced axiomatisation that commute source and target maps,
  we can neither derive the interchange axiom nor the morphism axioms
  for $\src_1$ and $\tgt_1$ in the special case of
  $2$-categories. Isabelle produces once again
  counterexamples. Coincidentally, the morphism laws for $\src_0$ and
  $\tgt_0$ are derivable; proofs can be found in our Isabelle
  theories.
\end{remark}

\begin{example}\label{ex:shuffle}
  Let $(\Sigma^\ast,\odot_0,\varepsilon)$ denote the free monoid
  generated by the finite alphabet $\Sigma$, with word concatenation
  $\odot_0$ modelled as a multirelation and the empty word
  $\varepsilon$. It can be viewed as a category with
  $\src_0(w)=\varepsilon =\tgt_0(w)$ for all $w\in \Sigma^\ast$.
  Further, let $(\Sigma^\ast \odot_1,\varepsilon)$ be the shuffle
  multimonoid on $\Sigma^\ast$ from Example~\ref{ex:shuffle-catoid}
  with $\odot_1=\|$.  Then
  $(\Sigma^\ast,\odot_0,\odot_1,\{\varepsilon\})$ forms a
  $2$-catoid. It has one single $0$-cell, $\{\varepsilon\}$, which is
  also the only $1$-cell. The source/target structure is therefore
  trivial, but an interchange law between word concatenation and word
  shuffle holds.
\end{example}


\section{Higher Quantales}\label{S:2-quantales}

As our main conceptual contribution, we now define the quantalic
structures that match higher catoids in Jónsson-Tarski-type
correspondences. We start with an axiomatisation of
$\alpha$-quantales, which correspond to $\alpha$-catoids, but then
turn our attention mainly to $\omega$-quantales.  Once again we have
checked all proofs in this section with Isabelle~\cite{CalkS24}.

An \emph{$\alpha$-quantale} is a structure
$(Q,\le,\cdot_i,1_i,\dom_i,\cod_i)_{0\le i < \alpha}$, for an ordinal
$\alpha \in \{0,1,\dots,\omega\}$, such that each
$(Q,\le,\cdot_i,1_i,\dom_i,\cod_i)$ is a modal quantale and the
structures interact as follows: 
\begin{itemize}
 \item  for all $i\neq j$,
  \begin{equation*}
    \dom_i(\alpha\cdot_j \beta) \le  \dom_i(\alpha)\cdot_j
    \dom_i(\beta)\qquad\text{ and }\qquad
 \cod_i(\alpha\cdot_j \beta) \le \cod_i(\alpha)\cdot_j
 \cod_i(\beta),
\end{equation*}
\item and for all $i<j$
\begin{gather*}
  (\alpha\cdot_j \beta) \cdot_i (\gamma\cdot_j \delta) \le (\alpha \cdot_i \gamma)\cdot_j (\beta\cdot_i
  \delta)\qquad\text{ and }\qquad
 \dom_j(\dom_i(\alpha)) = \dom_i(\alpha).
\end{gather*}
\end{itemize}

An $\alpha$-quantale is \emph{strong} if for all $i <j$, 
\begin{equation*}
  \dom_j(\alpha\cdot_i \beta) = \dom_j(\alpha)\cdot_i
  \dom_j(\beta) \qquad\text{ and }\qquad
 \cod_j(\alpha\cdot_i \beta) = \cod_j(\alpha)\cdot_i \cod_j(\beta).
\end{equation*}

These axiom systems are already reduced and irredundant in the sense
of Section~\ref{S:2-lr-msg}.

\begin{example}\label{ex:d10-counter}
  In the double modal quantale on $\bot<a<\top$ with
  $\cdot_0= \cdot_1= \land$, $1_0=1_1=\top$, $\dom_0=\mathit{id}$ and
  $\dom_1(a)=\top=\cod_1(a)$ (the rest being fixed), the first five
  globular $2$-quantale axioms hold, but the last one does not:
  $\dom_1(\dom_0(a)) = \top \neq a = \dom_0(a)$.

  Irredundancy of the two weak morphism laws is established using
  similar counterexamples with $5$-elements found by Isabelle. Their
  particular form is of little interest.
\end{example}

\begin{remark}
  The strong resemblance of the $\alpha$-quantale axioms and the
  $\alpha$-catoid ones is caused by correspondences that are developed
  in the following sections.  Nevertheless, there is a mismatch:
  $\dom_j\circ \dom_i = \dom_i$ is an $\alpha$-quantale axiom while
  $\src_j\circ \src_i = \src_i$ is derivable in $\alpha$-catoids and
  the same holds for the two morphism axioms of
  $\alpha$-quantales. For $\dom_j\circ \dom_i = \dom_i$ with $i<j$,
  this can be explained as follows. Our proof of
  $\src_j\circ \src_i = \src_i$ relies on
  $\Delta_k(x,y)\Rightarrow \tgt_k(x)=\src_k(y)$, but
  Remark~\ref{R:locality} shows that the corresponding property is not
  available for quantales. The related properties
  $\alpha\cdot_k\beta\neq \bot \Rightarrow \cod_k(\alpha)\inf
  \dom_k(\beta)\neq \bot$ \emph{are} available, but too weak to
  translate the proof of $\src_j\circ \src_i = \src_i$ to quantales.
  Any formal proof is of course ruled out by
  Example~\ref{ex:d10-counter}.
\end{remark}

We now turn to $\omega$-quantales.

\begin{lemma}\label{lemma:2-quantale-props}
  In every $\omega$-quantale $Q$, for $0\le i<j<\omega$, 
  \begin{enumerate}
  \item $\dom_j\circ \cod_i = \cod_i$, $\cod_j\circ \dom_i = \dom_i$
    and $\cod_j\circ \cod_i = \cod_i$,
  \item $1_j\le 1_j\cdot_i 1_j$, $1_j\cdot_i 1_j= 1_j$ if $Q$ is
    strong, and $1_i\cdot_j 1_i = 1_i$,
  \item $1_i\le 1_j$,
  \item $\dom_j (1_i) = 1_i$, $\dom_i (1_j) = 1_j$, $\cod_j (1_i) = 1_i$
    and $\cod_i(1_j) = 1_j$,
  \item $\dom_i\circ \dom_j = \dom_j \circ  \dom_i$, $\dom_i\circ
    \cod_j = \cod_j\circ\dom_i$, $\cod_i\circ \dom_j = \dom_j\circ
    \cod_i$ and $\cod_i\circ\cod_j = \cod_j\circ \cod_i$,
  \item $\dom_i(\alpha\cdot_j \beta) = \dom_i (\alpha\cdot_j \dom_j(\beta))$ and
    $\cod_i(\alpha\cdot_j \beta) = \cod_i(\cod_j(\alpha)\cdot_j \beta)$. 
  \end{enumerate}
\end{lemma}

A proof can be found in Appendix~\ref{A:proofs}. By (1), the sets
$Q_i = Q_{\dom_i}$ form a chain:
$Q_0\subseteq Q_1\subseteq Q_2 \subseteq \dots \subseteq Q$. Each
$Q_i$ is a complete distributive lattice with $+$ as binary sup and
$\odot_i$ as binary inf according to the properties of domain
quantales recalled in Section~\ref{S:modal-quantales}. Similarly to
the situation for $\omega$-catoids in Section~\ref{S:2-lr-msg}, the
elements of $Q_i$ remain domain elements in all higher dimensions,
hence each $Q_i$ is a distributive sublattice of $Q_j$ for all
$j\ge i$ and of $Q$. We have
\begin{align*}
  \dom_i(x)\odot_k \dom_j(y) &= dom_i(x) \inf \dom_j(y),\\
  \cod_i(x)\odot_k \cod_j(y) &= cod_i(x) \inf \cod_j(y)
\end{align*}
for all $i,j\le k$. At the same time, all truncations $Q_i$ of $Q$ are
$i$-quantales, so that the chain of the $Q_i$ is a filtration of
$\omega$-quantales. The $\omega$-quantale $Q$ is the union of the
quantales $Q_i$ if we add an axiom guaranteeing that for all $x\in Q$
there exists and $i\le \omega$ such that $\dom_i(x)=x=\cod_i(x)$. As
for $\omega$-catoids, we keep this optional and do not require this
property in the considerations that follow. The same results hold for
strong $\omega$-quantales.

\begin{remark}\label{R:cka}
  An interchange law
  $(\alpha\cdot_1 \beta) \cdot_0 (\gamma\cdot_1 \delta) \le (\alpha
  \cdot_0 \gamma)\cdot_1(\beta\cdot_0\delta)$ features in concurrent
  semirings~\cite{Gischer88} and concurrent Kleene algebras and
  quantales~\cite{HoareMSW11}. It has often been contrasted with the
  seemingly equational interchange laws of category theory. Yet this
  ignores the weak nature of equality in categories, which may depend
  on definedness conditions of terms, and which is captured
  explicitly and precisely by the multioperational language.
  Example~\ref{ex:cat-inclusions} shows that the interchange laws of
  $\omega$-categories are as weak as those of $\omega$-quantales and
  $\omega$-Kleene algebras (defined below), of which concurrent
  quantales and Kleene algebras are special cases. See
  Appendix~\ref{A:eckmann-hilton} for pitfalls of strong interchange
  laws and related morphisms.
\end{remark}
  
\begin{example}\label{ex:id-id}
  The identity $1_i=1_j$ need not hold for $i<j$ in strong
  $\omega$-quantales.  There is a strong $2$-quantale on $0<1_0<1_1$
  in which $1_1\cdot_0 1_1=1_1$ and $1_0\cdot_1 1_1=1_1$ (the rest is
  fixed), and $\dom_0(1_1)=1_0=\cod_0(1_1)$ and
  $\dom_1(1_0)=1_0=\cod_1(1_0)$ (the rest is again fixed). This makes
  strong $2$-quantales different from the original concurrent
  quantales mentioned and prevents smaller interchange laws with two
  or three variables.  See~\cite{CranchDS21} for a discussion of how
  the condition $1_0=1_1$ leads to a partial Eckmann-Hilton-style
  collapse.
\end{example}

Next we consider the interactions of the Kleene stars with the
$\omega$-structure.
\begin{lemma}\label{lemma:2-quantale-star-props}
  In every $\omega$-quantale $Q$, for $0\le i<j<\omega$, 
  \begin{enumerate}
  \item
    $\dom_i (\alpha) \cdot_i \beta^{\ast_j} \le (\dom_i
    (\alpha)\cdot_i \beta)^{\ast_j}$ and $\alpha^{\ast_j}\cdot_i \cod_i(\beta) \le (\alpha\cdot_i \cod_i
    (\beta))^{\ast_j}$, 
     \item
    $\dom_j(\alpha) \cdot_i \beta^{\ast_j} \le (\dom_j
    (\alpha)\cdot_i \beta)^{\ast_j}$ and $\alpha^{\ast_j}\cdot_i \cod_ j(\beta) \le (\alpha\cdot_i \cod_j
    (\beta))^{\ast_j}$ if $Q$ is strong,
 \item $(\alpha \cdot_j \beta)^{\ast_i} \le \alpha^{\ast_i} \cdot_j \beta^{\ast_i}$. 
  \end{enumerate}
\end{lemma}

See Appendix~\ref{A:proofs} for proofs. The properties in (1) and (2)
feature as axioms of globular $n$-Kleene algebras
in~\cite{CalkGMS20}. In sum, all axioms of these $n$-Kleene algebras
have now been derived from our smaller, but slightly different set of
axioms for $\omega$-quantales and $n$-quantales.  However, these
quantales presuppose that arbitrary joins and meets exist, while
globular $n$-Kleene algebras are based on globular $n$-semirings,
where only finite sups are assumed to exist and meets are not part of
the language. See Sections~\ref{S:globular-semiring} and
\ref{S:globular-ka} for a detailed discussion. We summarise
  this discussion as follows.

  \begin{proposition}\label{P:omega-quantale-vs-globular-ka}
    Every strong $\omega$-quantale is a globular $\omega$-Kleene
    algebra à la~\cite[Definition 3.2.7]{CalkGMS20} with Kleene stars
    $\alpha^{\ast_j} = \Sup_{i\ge 0} \alpha^{i_j}$.
  \end{proposition}
  Strictly speaking, only globular $n$-Kleene algebras are considered
  in~\cite{CalkGMS20}, but the axioms for $\omega$ are the same.

Finally we list further properties of domains and codomains that are
useful below.
\begin{lemma}\label{lemma:mod-props}
  In every $\omega$-quantale, for $0\le i< j<\omega$, the following properties hold:
  \begin{enumerate}
  \item $\dom_i(\alpha)\cdot_j \dom_i(\alpha) = dom_i(\alpha)$ and
    $\cod_i(\alpha)\cdot_j \cod_i(\alpha) = cod_i(\alpha)$,
  \item $\dom_i (\alpha \cdot_j \beta) = \dom_i ( \alpha \cdot_j \dom_j(\beta))$ and
    $\cod_i (\alpha \cdot_j\beta) = \cod_i (\cod_j(\alpha) \cdot_j \beta)$,
  \item $\dom_i (\alpha \cdot_j \beta) = \dom_i ( \cod_j(\alpha) \cdot_j \beta)$ and
    $\cod_i (\alpha\cdot_j \beta) = \cod_i (\alpha \cdot_j \dom_j(\beta))$,
  \item $\dom_i (\alpha \cdot_i \beta) = \dom_i (\alpha \cdot_i \dom_j(\beta))$ and
    $\cod_i (\alpha \cdot_i \beta) = \cod_i ( \cod_j(\alpha) \cdot_i \beta)$,
  \item $\dom_i (\alpha \cdot_i \beta) \le \dom_i (\cod_j(\alpha)\cdot_i \beta)$ and
    $\cod_i (\alpha \cdot_i \beta) \le \cod_i (\alpha \cdot_i
    \dom_j(\beta))$,  and equalities hold if $Q$ is strong,
  \item
    \begin{align*}
      \dom_i(\alpha) \cdot_i (\beta \cdot_j \gamma) &\le (\dom_i(\alpha) \cdot_i \beta)
                                                      \cdot_j (\dom_i \alpha)\cdot_i \gamma),\\
      (\alpha\cdot_j \beta)\cdot_i \dom_i(\gamma)
                                                    &\le (\alpha\cdot_i \dom_i(\gamma))\cdot_j (\beta \cdot_i
                                                      \dom_i(\gamma)),
    \end{align*}
  \item
    $\dom_i(\dom_j(\alpha)\cdot_j \beta) \le \dom_i(\alpha)\cdot_j
    \dom_i(\beta)$ and
    $\cod_i(\alpha\cdot_j \cod_j(\beta)) \le \cod_i(\alpha)\cdot_j
    \cod_i(\beta)$,
  \item $\dom_j(\dom_i(\alpha) \cdot_j \beta) = \dom_i(\alpha)\cdot_j
    \dom_j(\beta)$ and $\cod_j(\alpha\cdot_j \cod_i(\beta)) = \cod_i(\alpha)\cdot_i
    \cod_j(\beta)$,
  \item $\dom_i(\alpha)\cdot_j \dom_i(\beta) = dom_i(\alpha)\cdot_i \dom_i(\beta)$ and
    $\cod_i(\alpha)\cdot_j \cod_i(\beta) = cod_i(\alpha)\cdot_i \cod_i(\beta)$,
  \item
    \begin{multline*}
      (\dom_i(\alpha)\cdot_j \dom_i(\beta))\cdot_i
      (\dom_i(\gamma)\cdot_j\dom_i(\delta)) \\= (\dom_i(\alpha)\cdot_i
      \dom_i(\gamma))\cdot_j (\dom_i(\beta)\cdot_i\dom_i(\delta)),\\
      (\cod_i(\alpha)\cdot_j \cod_i(\beta))\cdot_i
      (\cod_i(\gamma)\cdot_j\cod_i(\delta)) = (\cod_i(\alpha)\cdot_i
      \cod_i(\gamma))\cdot_j (\cod_i(\beta)\cdot_i\cod_i(\delta)).
    \end{multline*}
  \end{enumerate}
\end{lemma}

See again Appendix~\ref{A:proofs} for proofs. The laws in (2)-(5) are
extended locality laws, those in (6) are weak distributivity laws for
compositions, those in (7) and (8) extended export laws. Note that
export laws $\dom(\dom(\alpha)\beta) = \dom(\alpha)\dom(\beta)$ and
$\cod(\alpha\cod(\beta))= \cod(\alpha)\cod(\beta)$ hold in any modal
semiring and quantale. Finally, the laws in (9) are useful for proving
the strong interchange laws in (10).

\begin{example}\label{ex:shuffle-language}
  The shuffle $2$-catoid on $\Sigma$ in Example~\ref{ex:shuffle}
  extends to the shuffle language $2$-quantale on $\Sigma$ under the
  standard language product
  \begin{equation*}
    XY=\{vw\mid v\in X,w\in Y\}
  \end{equation*}
  and the shuffle product of languages discussed in
  Example~\ref{ex:catoid-quantale}. The domain/codomain structures are
  trivial, as the empty word language $\{\varepsilon\}$ is the joint
  identity of the two underlying quantales, but an interchange law
  $(W\| X)\cdot (Y\| Z)\subseteq (W\cdot Y)\| (X\cdot Z)$ at language
  level can be derived. $2$-Quantales satisfying such more restrictive
  conditions are known as \emph{interchange
    quantales}~\cite{CranchDS21}.
\end{example}


\section{Higher Convolution Quantales and their
  Correspondences}\label{S:globular-convolution}

We can now extend the correspondence triangles between local catoids
$C$, modal quantales $Q$ and modal convolution quantales
$Q^C$~\cite{FahrenbergJSZ21a} as well as those for interchange
multimonoids, interchange quantales and interchange convolution
quantales~\cite{CranchDS21}, which we mentioned briefly in the
introduction, to local $\omega$-catoids, $\omega$-quantales and
convolution $\omega$-quantales, as well as their truncations at
dimension $n$.  These constitute the main
technical contribution in this article.  In
Section~\ref{S:msg-quantale} we specialise these results to modal
powerset $\omega$-quantales, which, at dimension $n$, have
applications in higher-dimensional rewriting~\cite{CalkGMS20}.  We
start with a formal summary of the $1$-dimensional and
$2$-dimensional cases considered in~\cite{FahrenbergJSZ21a,
  CranchDS21}.

Let $C$ be a catoid and $Q$ a quantale,  henceforth called
\emph{value} or \emph{weight quantale}. We write $Q^C$ for the set of
functions from $C$ to $Q$.  We define, for $f,g:C\to Q$, the
\emph{convolution operation} $\ast:Q^C\times Q^C\to Q^C$ as
\begin{equation*}
  (f\ast g)(x) = \Sup_{x\in y\odot z} f(y)\cdot g(z)
\end{equation*}
and the \emph{unit function} $\id_0:C\to Q$ as

\begin{equation*}
  \id_0(x)=
  \begin{cases}
    1 & \text{ if } x\in Q_0,\\
    \bot & \text{ otherwise}.
  \end{cases}
\end{equation*}
We define $\Sup F:C\to Q$ by pointwise extension,
$(\Sup F)(x) = \Sup \{f(x)\mid f \in F\}$ for $F\subseteq Q^C$, and in
particular $\bot:C\to Q$ by $\bot(x)=\bot$ for all $x\in C$,
overloading notation.  We also extend the order on $Q$ pointwise to a
relation on $Q^C$. It is consistent with the standard order on the
lattice $C\to Q$.

Following~\cite{CranchDS21,FahrenbergJSZ21a} we introduce further
notation. For any predicate $P$ we define
\begin{equation*}
  [P]=
  \begin{cases}
    1 & \text{ if } P,\\
    \bot & \text{ otherwise},
  \end{cases}
\end{equation*}
and then $\delta_x(y) = [x=y]$.  Any $f:C\to Q$ can now be written as
\begin{equation*}
  f =\bigvee_{x\in C} f(x)\cdot \delta_x,
\end{equation*}
where $\cdot$ is a module-style action between the ``scalars''
$f(x) \in Q$ and functions $\delta_x\in Q^C$.  More generally, we often
write $\delta^\alpha_x$ for $\alpha\cdot \delta_x$. Then
\begin{align*}
  \id_{Q_0} &= [e\in Q_0] =\Sup_{e\in Q_0}\delta_e,\\
\Sup F &= \Sup_{x \in C}\Sup\{f(x)\mid f \in F\} \cdot \delta_x,\\
f\lor g &= \Sup_{x\in C} (f(x)\lor g(x))\cdot \delta_x.
\end{align*}
Also, for
convolution,
\begin{align*}
  (f\ast g)(x) &= \Sup_{y,z\in C} f(y)\cdot g(z)\cdot [x\in y\odot
  z],\\
f\ast g &= \Sup_{x,y,z\in C} f(y)\cdot g(z)\cdot [x\in y\odot z]\cdot
\delta_x.
\end{align*}

Finally, whenever $Q$ is a modal quantale, we define
$\Dom,\Cod:Q^C\to Q^C$ as
\begin{align*}
  \Dom(f) &= \Sup_{x\in C} \dom(f(x))\cdot
            \delta_{\src(x)},\\ \Cod(f) &= \Sup_{x\in C} \cod(f(x))\cdot
                                          \delta_{\tgt(x)}.
\end{align*}

Correspondence triangles for relational monoids and quantales as well
as relational interchange monoids and interchange quantales are
already known~\cite{CranchDS21}. They have been extended to local
catoids and modal quantales~\cite{FahrenbergJSZ21a}.  The resulting
$2$-out-of-$3$ laws between catoids $C$, value algebras $Q$ and
convolution algebras $Q^C$ require mild non-degeneracy conditions on
$C$ or $Q$. They have previously been given at a fine level of
granularity to explain correspondences between individual laws in $C$,
$Q$ and $Q^C$. Here we only summarise those results relevant to
higher-dimensional extensions.

The following fact translates results for relational monoids and
related structures to the setting of catoids.

\begin{theorem}[{\cite[Theorem 4.8]{DongolHS21},\cite[Proposition 16, Corollary 21]{CranchDS21}}]\label{theorem:quantale-corresp}~
  \begin{enumerate}
  \item Let $C$ be a catoid and $Q$ a quantale. Then $Q^C$ is a
    quantale with the convolution and unit structure defined
    above.
  \item Let $X$ be a set, let $Q^X$ and $Q$ be quantales such that
    $\bot\neq 1$ and $\alpha\cdot (\beta\cdot \gamma) \neq \bot$ for
    some $\alpha,\beta,\gamma\in Q$. Then $X$ can be
    equipped with a catoid structure.
  \item Let $Q$ be a complete lattice equipped with a
    multiplication that preserves arbitrary sups and has a unit.
    Let $Q^C$ be a quantale and $C$ a catoid such that
    $C_0\neq \emptyset$ and $w\in (x \odot y)\odot z$ for some
    $w,x,y,z\in C$. Then $Q$ is a quantale.
  \end{enumerate}
\end{theorem}

We henceforth refer to quantales on function spaces, such as $Q^C$, as
\emph{convolution quantales}. The two-out-of-three correspondence in
Theorem~\ref{theorem:quantale-corresp} is illustrated in
Figure~\ref{Fig:triangle1} in the introduction.

In the following lemma, \emph{double catoid} refers to a set equipped
with two catoid structures that do not interact. Likewise,
\emph{double quantale} refers to a complete lattice equipped with two
monoidal structures that do not interact. In particular, there are no
interchange laws.

\begin{theorem}[{\cite[p,~934]{CranchDS21}}]\label{theorem:interchange-corresp}
  Let $(C,\cdot_0,\src_0,\tgt_0,\cdot_1,\src_1,\tgt_1)$ be a double
  catoid, and further let $(Q,\le,\cdot_0,1_0,\cdot_1,1_1)$ be a double quantale.
  Let $(Q^C,\le,\ast_0,\id_0,\ast_1,\id_1)$ be the associated double convolution quantale.
\begin{enumerate}
\item The interchange law holds in $Q^C$ if it holds in $C$ and $Q$.
\item The interchange law holds in $C$ if it holds in $Q$ and $Q^C$,
  and if $(\alpha\cdot_1 \beta)\cdot_0 (\gamma\cdot_1 \delta)\neq \bot$ for some $\alpha,\beta,\gamma,\delta\in Q$.
\item The interchange law holds in $Q$ if it holds in $C$ and $Q^C$,
  and if $y \in u\odot_1 v$, $z \in w \odot_1 x$ and $\Delta(y,z)$
  hold for some $u,v,w,x,y,z\in C$.
\end{enumerate}
\end{theorem}

Once again we have translated the original statement
in~\cite{CranchDS21} from relational monoids to catoids along the
isomorphism between these structures. This two-out-of-three
correspondence is illustrated in Figure~\ref{Fig:triangle4}. Double
quantales in which the interchange law holds are known as
\emph{interchange quantales}. The results in~\cite{CranchDS21}
establish in fact correspondences between $2$-catoids and interchange
quantales, but only the interchange law of the double catoid
contributes to the interchange law in the double quantale and vice
versa.

\begin{figure}[t]
    \centering
  \begin{tikzpicture}[x=4cm, y=3cm]
    \node (X) at (0,0) {interchange catoid $C$};
    \node (Q) at (1,0) {interchange quantale $Q$};
    \node (R) at (.5,.8) {interchange quantale $Q^C$};
    \path (X) edge[-] node[coordinate, pos=.2] (QX) {} node[coordinate, pos=.8] (XQ) {} (Q);
    \path (X) edge[-] node[coordinate, pos=.2] (RX) {} node[coordinate, pos=.85] (XR) {} (R);
    \path (Q) edge[-] node[coordinate, pos=.2] (RQ) {} node[coordinate, pos=.85] (QR) {} (R);
  \end{tikzpicture}
  \caption{Correspondence between interchange catoid $C$,
    interchange quantale $Q$ and interchange quantale $Q^C$
    from~\cite{CranchDS21}}
  \label{Fig:triangle4}
\end{figure}

\begin{theorem}[{\cite[Theorems 7.1, 8.4, 8.5]{FahrenbergJSZ21a}}]\label{theorem:modal-corresp}
  Let $(C,\odot,\src,\tgt)$ be a catoid and $(Q,\le,\cdot,1)$ be a quantale.
  Let $(Q^C,\le,\ast,\id_0)$ be the associated convolution quantale.
  \begin{enumerate}
  \item $Q^C$ is modal if $C$ is local and $Q$ modal.
    \item $C$ is local if $Q$ and $Q^C$ are modal, and if $1\neq \bot$
      in $Q$.
    \item $Q$ is modal if $C$ is local and $Q^C$ modal, and if
      $\Delta(\ell(x), r(y))$ and $\Delta(z,w)$ for some
      $w,x,y,z\in C$.
  \end{enumerate}
\end{theorem}

This two-out-of-three correspondence is illustrated in
Figure~\ref{Fig:triangle5}.

\begin{figure}[t]
    \centering
  \begin{tikzpicture}[x=3.4cm, y=3cm]
    \node (X) at (0,0) {local catoid $C$};
    \node (Q) at (1,0) {modal quantale $Q$};
    \node (R) at (.5,.8) {modal quantale $Q^C$};
    \path (X) edge[-] node[coordinate, pos=.2] (QX) {} node[coordinate, pos=.8] (XQ) {} (Q);
    \path (X) edge[-] node[coordinate, pos=.2] (RX) {} node[coordinate, pos=.85] (XR) {} (R);
    \path (Q) edge[-] node[coordinate, pos=.2] (RQ) {} node[coordinate, pos=.85] (QR) {} (R);
  \end{tikzpicture}
  \caption{Correspondence between local catoid $C$, modal quantale $Q$
    and modal convolution quantale $Q^C$}
  \label{Fig:triangle5}
  \end{figure}

\begin{example}[\cite{CranchDS21,FahrenbergJSZ21a}]\label{ex:cat-quantale}
 Theorems~\ref{theorem:interchange-corresp}
 and~\ref{theorem:modal-corresp} have the following 
 instances that generalise Example~\ref{ex:catoid-quantale}. Let $Q$
 be a modal value quantale. 
 \begin{enumerate}
 \item Every category $C$ extends to a modal convolution quantale
   $Q^C$. It is similar to a category algebra, using a quantale as
   value algebra instead of a ring or field and extending the source
   and target structure of the category, which category algebras
   ignore.
 \item The modal powerset quantales from
   Example~\ref{ex:catoid-quantale} are convolution quantales with
   $Q=2$.
 \item Every path category over a digraph extends to a modal convolution
   quantale of $Q$-valued paths.
 \item Every pair groupoid extends to a modal convolution quantale of
   $Q$-valued relation.
 \item Every shuffle $2$-catoid extends to a convolution $2$-quantale
   of $Q$-weighted shuffle languages if $Q$ is a $2$-quantale.
 \end{enumerate}
\end{example}

We now extend the constructions in the proofs of
Theorems~\ref{theorem:interchange-corresp}
and~\ref{theorem:modal-corresp} to proofs of $2$-out-of-$3$
correspondence triangles between $\omega$-catoids and
$\omega$-quantales, as shown in the diagram on the left of
Figure~\ref{Fig:triangle2}. Given the proofs in these theorems it
remains to consider the globular
structure. Theorem~\ref{theorem:interchange-corresp} guarantees a
$2$-out-of-$3$ correspondence between the interchange laws in
$\omega$-catoids and $\omega$-quantales, while
Theorem~\ref{theorem:modal-corresp} supplies such a correspondence
between the source and target structure in $\omega$-catoids and the
domain and codomain structure in $\omega$-quantales in each
dimension. Locality of $\omega$-catoids, in particular, is needed to
reflect locality of the domain and codomain structure. This, in turn,
is needed for defining modal operators as actions in
Section~\ref{S:modalities}.

In the $\omega$-setting, our notation for $[P]$ or $\delta$-functions
requires dimension indices, strictly speaking.  But in practice, we
can always pick and swap such indices. Hence we usually suppress them,
as will become clear in the proofs below.

\begin{theorem}\label{theorem:globular-corresp1}
  Let $C$ be a local $\omega$-catoid and $Q$ an
  $\omega$-quantale. Then $Q^C$ is an $\omega$-quantale. 
\end{theorem}
\begin{proof}
  In light of the proof of Theorems~\ref{theorem:interchange-corresp}
  and \ref{theorem:modal-corresp} it remains to extend the morphism
  axioms as well as the axiom $\Dom_j\circ \Dom_i= \Dom_i$ for
  $0\le i<j<\omega$. Although the proof of the interchange law is
  covered, for $2$-catoids and $2$-quantales, by
  Theorem~\ref{theorem:interchange-corresp}, we state it explicitly
  for the case of $\omega$.  We omit indices related to square
  brackets as mentioned. For $0\le i< j < \omega$,
   \begin{align*}
    &((f\ast_j g) \ast_i (h\ast_j k))(x)\\
    &= \Sup_{y,z} (f\ast_j g)(y) \cdot_i (h\ast_j k)(z) \cdot [x \in y
      \odot_i z]\\
       &= \Sup_{y,z} \left(\Sup_{t,u}f(t)\cdot_j g(u)\cdot [y
         \in t\odot_j u]\right) \cdot_i \left(\Sup_{v,w}h(v)\cdot_j
         k(w)\cdot [z \in v\odot_j w] \right) \cdot [x \in y
         \odot_i z]\\
         &= \Sup_{t,u,v,w} (f(t) \cdot_j g(u))\cdot_i (h(v)\cdot_j
           k(w)) \cdot  [x \in (t \odot_j u)\odot_i(v\odot_j w)]\\
    &\le \Sup_{t,u,v,w} (f(t) \cdot_i h(v))\cdot_j (g(u)\cdot_i
           k(w)) \cdot [x \in (t \odot_i v)\odot_j(u\odot_i w)]\\
          &= \Sup_{y,z} \left(\Sup_{t,v}f(t)\cdot_i h(v)\cdot [y
         \in t\odot_i v]\right) \cdot_j \left(\Sup_{u,w}g(u)\cdot_i
         k(w)\cdot [z \in u\odot_i w]\right) \cdot [x \in y
            \odot_j z]\\
    & = ((f \ast_i h)\ast_j (g\ast_i k))(x).
   \end{align*}
   
   For the morphism axiom
   $\Dom_j(f\ast_i g) \le \Dom_j(f)\ast_i \Dom_j(g)$,
\begin{align*}
  &\Dom_j(f\ast_i g)(x)\\
  &=\Sup_u\dom_j\left(\Sup_{v,w}f(v)\cdot_i g(w))\cdot [u\in v\odot_i w]\right)\cdot
  \delta_{\src_j(u)}(x)\\
&=\dom_j\left(\Sup_{v,w}f(v)\cdot_i g(w)\right)\cdot [x\in \src_j(v\odot_i w)]\\
&=\Sup_{v,w}\dom_j(f(v)\cdot_i g(w))\cdot [x\in \src_j(v\odot_i w)]\\
&\le \Sup_{v,w}\dom_j(f(v))\cdot_i \dom_j(g(w))\cdot [x\in
  \src_j(v)\odot_i \src_j(w)]\\
&=\Sup_{t,u} \left(\Sup_v\dom_j(f(v))\cdot_j\delta_{\src_j(v)}(t)\right)\cdot_i
  \left(\Sup_w\dom_j(g(w))\cdot_j \delta_{\src_j(w)}(u)\right)\cdot  [x\in t\odot_i u]\\
&= (\Dom_j(f)\ast_i \Dom_j(g))(x).
       \end{align*}
 The proof of $\Cod_j(f\ast_i g) \le \Cod_j(f)\ast_i \Cod_j(g)$ follows by opposition.

 The proofs of
 \begin{align*}
   \Dom_i(f\ast_j g) &\le \Dom_i(f)\ast_j \Dom_i(g),\\
   \Cod_i(f\ast_j g) &\le \Cod_i(f)\ast_j \Cod_i(g)
 \end{align*}
 are obtained by re-indexing these proofs.

 Finally, for $\Dom_j\circ \Dom_i = \Dom_i$,
  \begin{align*}
    \Dom_j(\Dom_i(f))
     & = \Sup_u \dom_j\left(\Sup_v \dom_i(f(v))\cdot
       \delta_{\src_i(v)}(u) \right)\cdot \delta_{\src_j(u)}\\
    &= \Sup_v\dom_j(\dom_i(f(v))) \cdot \delta_{\src_j(\src_i(v))}\\
    &= \Sup_v\dom_i(f(v))\cdot \delta_{\src_i(v)}\\
    &= \Dom_i(f).\qedhere
  \end{align*}
\end{proof}

For strong $\omega$-catoids, and hence $\omega$-categories, we obtain a
stronger result.

\begin{corollary}\label{corollary:globular-corresp1}
  Let $C$ be a strong local $\omega$-catoid and $Q$ a strong
  $\omega$-quantale. Then $Q^C$ is a strong $\omega$-quantale.
\end{corollary}
\begin{proof}
  It suffices to replay the proofs for the two strong morphism laws
  for $\Dom_j$ and $\Cod_j$ with an equality step in the fourth proof
  step for $\Dom_j$ above. The proof for $\Cod_j$ follows by
  opposition.
\end{proof}

For the following two theorems we tacitly assume the non-degeneracies
needed for Theorem~\ref{theorem:interchange-corresp} and
\ref{theorem:modal-corresp}, calling the respective structures
\emph{sufficiently supported}, and mention those on $C$ and $Q$
explicitly in the proof. See~\cite{CranchDS21} for the general
construction and more detailed explanations of the mechanics of
proofs.  First we recall three properties from~\cite[Lemma 8.2]{FahrenbergJSZ21a}
that generalise readily beyond one dimension.

By analogy to the double catoids and quantales considered in
Theorem~\ref{theorem:interchange-corresp}, we define an $\omega$-fold
catoid as a set equipped with $\omega$ catoid structures which do not
interact.  Likewise, an \emph{$\omega$-fold quantale} is a complete
lattice equipped with $\omega$ sup-preserving monoidal structures
which do not interact.

\begin{lemma}\label{lemma:corresp-lem}
  Let $C$ be an $\omega$-fold catoid and $Q$ be an $\omega$-fold modal
  quantale, all without globular structure.
  Let $Q^C$ be the associated $\omega$-fold modal quantale.
  Then, for $i\neq j$,
\begin{enumerate}
 \item $\Dom_i(\delta^\alpha_x) = \Sup_y \dom_i(
  \delta^\alpha_x(y))\delta_{\src_i(y)} = \dom_i(\alpha) \cdot
  \delta_{\src_i(x)}$,
\item
  $\Dom_i(\delta^\alpha_x\ast_j \delta^\beta_y)(z) = \dom_i(\alpha
  \cdot_j \beta) \cdot_i [z\in \src(x\odot_j y)]$,
 \item $ (\Dom_i(\delta^\alpha_x)\ast_j \Dom_i(\delta^\beta_y))(z) =
  \dom_i(\alpha) \cdot_j \dom_i(\beta)\cdot_j [z\in \src_i(x)\odot_j
  \src_i(y)]$,
\item $(\delta_x^\alpha\ast \delta_y^\beta)(z) = \alpha\cdot
  \beta\cdot [z \in x \odot y]$. 
\end{enumerate}
\end{lemma}

\begin{theorem}\label{theorem:globular-corresp2}
  Let $X$ be a set, let $Q^X$ and $Q$ be $\omega$-quantales
  with $Q$ sufficiently supported. Then $X$ can be equipped
    with a local $\omega$-catoid structure.
\end{theorem}
\begin{proof}
  Given Theorems~\ref{theorem:interchange-corresp} and
  \ref{theorem:modal-corresp} we need to check the homomorphism
  axioms. As in Theorem~\ref{theorem:globular-corresp1}, all proofs
  are similar and we show only one. Suppose
  \begin{equation*}
    \Dom_i(\delta^\alpha_x\ast_j\delta^\beta_y)(z) \le
    (\Dom_i((\delta^\alpha_x)\ast_j \Dom_i(\delta^\beta_y))(z)
  \end{equation*}
  holds in $Q^C$ and
  $\dom_i(\alpha\ast_j\beta) \le \dom_i(\alpha)\ast_j \dom_i(\beta)$
  in $Q$. Suppose also, for non-degeneracy, that
  $\dom_i(\alpha\cdot_j\beta) \neq \bot$.  Then
  $\src_i(x\odot_j y)\le \src_i(x)\odot_j \src_i(y)$ using
  Lemma~\ref{lemma:corresp-lem}(2) and (3).
\end{proof}

Once again we get stronger results for strong quantales. The proofs
are obvious.

\begin{corollary}\label{corollary:globular-corresp2}
 Let $Q^C$ and $Q$ be strong $\omega$-quantales with $Q$ sufficiently
  supported. Then $C$ is a strong local $\omega$-catoid.
\end{corollary}

Additional assumptions are needed to obtain $\omega$-categories. We do
not explain them in this article.

\begin{theorem}\label{theorem:globular-corresp3}
  Let $Q$ be a complete lattice equipped with a multiplication
    that preserves arbitrary sups and has a unit. Let $Q^C$ be an
  $\omega$-quantale and $C$ a local $\omega$-catoid that is
  sufficiently supported. Then $Q$ is an $\omega$-quantale.
\end{theorem}
\begin{proof}
  Given Theorems~\ref{theorem:interchange-corresp} and
  \ref{theorem:modal-corresp} we need to check the homomorphism axioms
  and $\dom_j\circ \dom_i=\dom_i$. As in
  Theorem~\ref{theorem:globular-corresp1}, all proofs of homomorphism
  axioms are similar and we show only one. Suppose
  $\Dom_i(\delta^\alpha_x\ast_j \delta^\beta_y)(z) \le
  (\Dom_i(\delta^\alpha_x)\ast_j \Dom_i(\delta^\beta_y))(z)$ holds in
  $Q^C$ and $\src_i(x\odot_j y) \le \src_i(x)\odot_j \src_i(y)$ in
  $C$. Suppose also, for non-degeneracy, that
  $z \in \src_i(x)\odot_j \src_i(y)$. Then, using
  Lemma~\ref{lemma:corresp-lem}(2) and (3),
\begin{align*}
  \dom_i(\alpha\cdot_j \beta) &= \Dom_i(\delta^\alpha_x\ast_j \delta^\beta_y)(z)\\
  &\le (\Dom_i(\delta^\alpha_x)\ast_j
    \Dom_i(\delta^\beta_y))(z)\\
  &=\dom_i(\alpha) \cdot_j \dom_i(\beta). 
\end{align*}

Finally, using Lemma~\ref{lemma:corresp-lem}(1),
\begin{equation*}
  \dom_j(\dom_i(\alpha)) = \Dom_j (\Dom_i(\delta^\alpha_x))(\src_j(\src_i(x))) =
  \Dom_i(\delta^\alpha_x)(\src_i(x)) = \dom_i(\alpha).
\end{equation*}
\end{proof}

The two-out-of-three correspondence captured by
Theorems~\ref{theorem:globular-corresp1},
\ref{theorem:globular-corresp2} and \ref{theorem:globular-corresp2} is
depicted in the left diagram of Figure~\ref{Fig:triangle2} in the
introduction.

\begin{corollary}\label{corollary:globular-corresp3}
  Let $Q^C$ be a strong $\omega$-quantale and $C$ a strong local
  $\omega$-catoid that is sufficiently supported. Then $Q$ is a strong
  $\omega$-quantale.
\end{corollary}

\begin{example}\label{ex:omega-conv-quantale}~
  \begin{enumerate}
  \item The category $\Cat$ extends to a
    convolution $2$-quantale $Q^\Cat$ for every value $2$-quantale
    $Q$.
  \item The $\omega$-category of globular sets with the standard
    globular compositions extends to a convolution $\omega$-quantale
    for every value $\omega$-quantale $Q$.
  \end{enumerate}
\end{example}

Applications of the powerset case $Q=2$ relevant to
higher-dimensional rewriting are discussed in
Sections~\ref{S:msg-quantale}, \ref{S:globular-ka} and
\ref{S:n-structures}. A convolution $2$-quantale relevant to
concurrency theory based on weighted languages of isomorphism
classes of labelled posets, equipped with a so-called serial and a
parallel composition is discussed in~\cite{CranchDS21}. In this
case, the only unit is the empty poset, and the domain/codomain
structure is trivial, as for weighted languages with shuffle.


\section{Dedekind Convolution Quantales and their
  Correspondences}\label{S:dedekind-correspondence}

In this section we study correspondence triangles between groupoids,
Dedekind value quantales and Dedekind convolution quantales, adapting
the correspondence triangles between groupoids and relation algebras
established by Jónsson and Tarski~\cite[Section 5]{JonssonT52}.

We define, for every $f:C\to Q$ from a groupoid $C$ into an involutive
quantale $Q$, 
\begin{equation*}
  f^\conv(x) = (f(x^\inv))^\conv.
\end{equation*}
Alternatively, we write
$f^\conv= \Sup_{x\in C} (f(x))^\conv \cdot \delta_{x^\inv}$.

\begin{proposition}\label{prop:groupoid-involquantale}
Let $C$ be a groupoid and $Q$ an involutive quantale. Then $Q^C$ is an
involutive quantale. 
\end{proposition}
\begin{proof}
  Given Theorem~\ref{theorem:quantale-corresp} it remains to check the
  three involution axioms.

For sup-preservation, 
\begin{align*}
  (\Sup F)^\conv
  &= \Sup_x ((\Sup F)(x))^\conv \cdot \delta_{x^\inv}\\
  &= \Sup_x ((\Sup \{f(x)\mid f\in F\})^\conv \cdot \delta_{x^\inv}\\
  &= \Sup \{\Sup_xf(x)^\conv \cdot \delta_{x^\inv} \mid f\in F\}\\
    &= \Sup \{f^\conv\mid f\in F\}.
\end{align*}

For involution  proper, 
\begin{align*}
  f^{\conv\conv}
  &= \Sup_x\left(\Sup_y (f(y))^\conv \cdot
    \delta_{y^\inv}(x)\right)^\conv \cdot \delta_{x^\inv}\\
   &= \left(\Sup_y (f(y))^\conv \cdot
     \delta_{y^{\inv\inv}}\right)^\conv\\
  &= \Sup_y (f(y))^{\conv\conv} \cdot \delta_{y^{\inv\inv}}\\
  &= \Sup_y (f(y))\cdot \delta_{y}\\
  &= f.
\end{align*}

For contravariance, 
\begin{align*}
  (f\ast g)^\conv
  &= \Sup_x \left(\Sup_{y,z}f(y)\cdot g(z)\right)^\circ \cdot [x\in
    y\odot z]\cdot \delta_{x^\inv}\\
  &= \Sup_{x,y,z} (f(y)\cdot
    g(z))^\conv\cdot[x \in (y\odot z)^\inv]\cdot \delta_{x}\\
  &= \Sup_{x,y,z}g(z)^\conv\cdot
    f(y)^\conv\cdot [x \in z^\inv \odot y^\inv]\cdot \delta_{x}\\
   &= \Sup_{x,y,z}g(z^\inv)^\conv\cdot
    f(y^\inv)^\conv\cdot [x \in z \odot y]\cdot \delta_{x}\\
  &= \Sup_x\left(\Sup_{z,y} g^\conv(z)\cdot
    f^\conv(y)\cdot [x\in z\odot y]\right)\cdot \delta_x\\
  & = \Sup_x (g^\circ \ast f^\circ)(x)\cdot \delta_x\\
  & = (g^\circ \ast f^\circ).
\end{align*}
\end{proof}

\begin{theorem}\label{theorem:groupoid-dedekindquantale}
  Let $C$ be a groupoid and $Q$ a Dedekind quantale in which binary
  inf distributes over all sups. Then $Q^C$ is a Dedekind quantale (in
  which binary inf distributes over all sups).
\end{theorem}
\begin{proof}
\begin{align*}
  f\ast g \inf h
  &= \left(\left(\Sup_{y,z}f(y)\cdot g(z)\cdot [x\in y\odot z]\right) \inf
    h(x)\right)\cdot \delta_x\\
  &= \Sup_{x,y,z}(f(y)\cdot g(z) \inf h(x))\cdot [x\in y \odot z]\cdot
  \delta_x\\
  &\le  \Sup_{y,z} (f(y) \inf h(x) \cdot g(z)^\conv )\cdot
    g(z) \cdot [x \in y\odot z]\\
  &= \Sup_{x,y,z} (f(y) \inf h(x) \cdot g^\conv(z^\inv) )\cdot
    g(z)\cdot [x \in y\odot z] \cdot \delta_x\\
  & \le \Sup_{x,y,z} \left(f(y) \inf \Sup_{v,w} h(v) \cdot
    g^\conv(w)\cdot [y \in v\odot w]\right)\cdot
    g(z)\cdot [x\in y\odot z]\cdot \delta_x \\
  &= \Sup_{x,y,z} (f(y)\inf (h\ast g^\conv)(y))\cdot g(z)\cdot [x \in
    y\odot z]\cdot \delta_x\\
  &= \Sup_{x,y,z} (f\inf h\ast g^\conv)(y)\cdot g(z)\cdot [x \in y \odot
    z]\cdot \delta_x\\
  &=  (f\inf h\ast g^\conv) \ast g.
\end{align*}
The distributivity law is used in the first step of the proof.  The
sixth step works because $x\in y\odot z$ if and only if
$y\in x\odot z^\inv$ by Lemma~\ref{lemma:lr-groupoid-props}(2), so
that the pair $(x,z^\inv)$ is considered in the sup introduced.
\end{proof}

The condition that binary infs distribute over all sups is well known
from frames or locales, that is, complete Heyting algebras. It holds
in every power set Dedekind quantale. 

\begin{theorem}\label{theorem:conv-quantale-groupoid}
  Let $X$ be a set, let $Q^X$ and $Q$ be Dedekind quantales
  such that $\bot\neq 1$ in $Q$. Then $X$ can be equipped
    with  a groupoid structure.
\end{theorem}
\begin{proof}
  We show that $x \odot x^\inv = \{\src(x)\}$. Suppose
  $\Dom(\delta_x^\alpha) \le \delta_x^\alpha \ast
  (\delta_x^\alpha)^\conv$ holds in $Q^C$ and
  $\dom(\alpha)\le \alpha\cdot \alpha^\conv$ holds in $Q$. Suppose
  also for non-degeneracy that $\dom(\alpha)\neq \bot$; we can pick
  $\alpha=1$ and simply require $\bot\neq 1$.  By
  Lemma~\ref{lemma:corresp-lem}(1) and (4),
  \begin{equation*}
    \delta_{\src(x)}(y) = \Dom(\delta_x)(y) \le (\delta_x \ast
    (\delta_x)^\conv)(y) = [y \in x\odot x^\inv].
  \end{equation*}
  Thus $\src(x) \in x \odot x^\inv$. Moreover,
  $\{\src(x)\}= x \odot x^\inv$ because, in fact,
  $\Dom(\delta_x) = \delta_x \ast (\delta_x)^\conv$, that is,
  $\delta_x$ is functional.
\end{proof}

\begin{theorem}\label{theorem:conv-quantale-corresp2}
  Let $Q$ be a complete lattice with a binary multiplication that
  preserves sups in both arguments and has a unit. Let $Q^C$ be a
  Dedekind quantale and $C$ a groupoid. Then $Q$ is a Dedekind
  quantale.
\end{theorem}
\begin{proof}
  We start with the involutive quantale axioms. First, suppose
  $f^{\conv\conv} = f$ in $Q^C$ and $x^{\inv\inv} = x$ in $X$. Then
  $\alpha^{\conv\conv} = (\delta_x^\alpha)^{\conv\conv}(x^{\inv\inv})
  = \delta_x^\alpha(x)=\alpha$.

Second, suppose 
$(\Sup \{\delta_x^{\alpha}\mid \alpha \in A)^\conv =
\Sup\{(\delta_x^{\alpha})^\conv\mid \alpha\in A\}$ in $Q^X$. Then
\begin{equation*}
(\Sup A)^\conv =
(\Sup\{\delta_x^\alpha(x^\inv)\mid \alpha\in A\})^\conv =
\Sup\{(\delta_x^\alpha)^\conv(x^\inv)\mid \alpha\in A\} =
\Sup\{\alpha^\conv\mid \alpha \in A\}.
\end{equation*}

Third, suppose
$(\delta_y^\alpha \ast \delta_z^\beta)^\conv = (\delta_z^\beta)^\conv
\ast (\delta_y^\alpha)^\conv$ in $Q^C$ and
$(y\odot z)^\inv = z^\inv \odot y^\inv$ in $C$. Also assume
$z\in x \odot y$ for non-degeneracy. Then
\begin{equation*}
  (\alpha\cdot \beta)^\conv = (\delta_y^\alpha \ast
\delta_z^\beta)^\conv(x^\inv) = ((\delta_y^\beta)^\conv \ast
(\delta_z^\alpha)^\conv))(x^\inv) = \beta^\conv\cdot \alpha^\conv.
\end{equation*}

Finally, for the modular law, suppose $\delta_u^\alpha\ast \delta_v^\beta
\inf \delta_w^\gamma\le (\delta_u^\alpha \inf
\delta_w^\gamma\ast (\delta_v^\beta)^\conv)\ast
\delta_v^\beta$. Assume, for non-degeneracy,  that $x\in u\odot
v$. Then
\begin{align*}
  \alpha \cdot \beta \inf \gamma
  &= \alpha\cdot \beta\cdot [x\in u\odot v] \inf 
    \gamma \cdot \delta_x(x)\\
  &= (\delta_u^\alpha\ast \delta_v^\beta\inf \delta_x^\gamma)(x)\\
  &\le ((\delta_u^\alpha \inf
    \delta_x^\gamma\ast (\delta_v^\beta)^\conv)\ast \delta_v^\beta)(x)\\
  &= (\alpha \inf \gamma\cdot \beta^\conv \cdot [u\in
    x\odot v^\inv]) \cdot \beta\cdot [x\in u\odot v]\\
  &= (\alpha \inf \gamma\cdot \beta^\inv) \cdot \beta,
\end{align*}
where $x\in u\odot v \Leftrightarrow u \in x \odot v^\inv$, which
holds by Lemma~\ref{lemma:lr-groupoid-props}(2), is used in the fourth step.
\end{proof}

It is straightforward to check that $Q$ carries a complete Heyting
structure if $Q^C$
does. Theorem~\ref{theorem:groupoid-dedekindquantale} and
\ref{theorem:conv-quantale-corresp2} clearly depend on the quantales
$Q$ or $Q^C$ in their hypotheses being Dedekind, while
Theorem~\ref{theorem:conv-quantale-groupoid} uses slightly more
general properties of the shape $\dom(x)\le xx^\conv$ and
$\cod(x)\le x^\conv x$, which are further discussed in the context of
$\omega$-semirings and $\omega$-Kleene algebras in
Section~\ref{S:globular-ka}.

The two-out-of-three correspondence captured by
Theorems~\ref{theorem:groupoid-dedekindquantale},
\ref{theorem:conv-quantale-groupoid} and
\ref{theorem:conv-quantale-corresp2} is illustrated by the left
diagram of Figure~\ref{Fig:triangle3} in the introduction.

\begin{example}\label{ex:weighted-relational-dedekind-quantales}
  Example~\ref{ex:rel-dedekind} generalises from powerset quantales
  to convolution quantales. Let $Q$ be a Dedekind value quantale. 
  \begin{enumerate}
  \item Every free groupoid extends to a modal $Q$-valued path
    Dedekind quantale in which formal inverses extend to converses.
  \item Every pair groupoid extends to a modal Dedekind quantale of
    $Q$-valued relations.
  \end{enumerate}
\end{example}


\section{$(\omega,p)$-Catoids and
  $(\omega,p)$-Quantales}\label{S:n-structures}

In this section we combine the structures of $\omega$-catoids and
$\omega$-quantales from Sections~\ref{S:2-lr-msg}, \ref{S:2-quantales}
and \ref{S:globular-convolution} with the structures of groupoids and
Dedekind quantales from Sections~\ref{S:lr-groupoids},
\ref{S:modular-quantales} and \ref{S:dedekind-correspondence} to
define the corresponding $(\omega,p)$-structures in higher-dimensional
rewriting theory.  We also show how $2$-out-of-$3$ correspondence
triangles for the $(\omega,p)$-structures can be obtained from those
of their component structures, see also Figure~\ref{Fig:triangle2} in
the introduction. We start from $\omega$-catoids which have a groupoid
structure above some dimension $p<\omega$. These $(\omega,p)$-catoids
and the corresponding $(\omega,p)$-quantales, which are Dedekind
quantales above dimension $p$, are suitable for defining higher
homotopies.
 
An \emph{$(\omega,p)$-catoid} is an $\omega$-catoid $C$ with
operations $(-)^{\inv_i}:C\to C$ for all $p< i \leq \omega$ such that
$(x)^{\inv_i} \in C_i$ for all $x\in C_i$ and the groupoid axioms
hold, for all $p< i\leq \omega$ and $x\in C$:
 \begin{equation*}
 x \odot_{i-1} x^{\inv_i} = \{\src_{i-1}(x)\}
 \quad\text{and}\quad
 x^{\inv_i} \odot_{i-1} x = \{\tgt_{i-1}(x)\}.
\end{equation*}
An \emph{$(\omega,p)$-category} is a local functional
$(\omega,p)$-catoid.

Likewise, an \emph{$(\omega,p)$-quantale} is an $\omega$-quantale $Q$
with operations $(-)^{\conv_i}:Q\to Q$ for all $p< i\le\omega $ such
that $(\alpha)^{\conv_j}\in Q_i$ for all $\alpha\in Q_i$ and that
satisfy the involution axioms and the modular law with respect to the
composition $\cdot_{i-1}$.  Strong $(\omega,p)$-quantales are defined
as for higher quantales in Section~\ref{S:2-quantales}.

All inverses in $(\omega,p)$-catoids and all converses in
$(\omega,p)$-quantales are trivial on elements of lower dimensions.
It follows from Lemma~\ref{lemma:lr-mgp-local} that $x^{\inv_ j} = x$
for all $x\in C_i$ with $i< j$. Similarly, it follows from
Lemma~\ref{lemma:dom-conv} that $\alpha^{\conv_ j} = \alpha$ for all
$\alpha\in Q_i$ with $i< j$.
 
The correspondence results from
Sections~\ref{S:globular-convolution} and
\ref{S:dedekind-correspondence} can then be combined. We use the
notion of sufficient support as for $\omega$-catoids,
$\omega$-groupoids and $\omega$-quantales defined in these sections.

\begin{theorem}\label{theorem:np-correspondences}~
  \begin{enumerate}
  \item Let $C$ be a local $(\omega,p)$-catoid and $Q$ an
    $(\omega,p)$-quantale in which binary inf distributes over all
    sups. Then $Q^C$ is an $(\omega,p)$-quantale.
  \item Let $X$ be a set, let $Q^X$ be an $(\omega,p)$-quantale and
    $C$ a local $(\omega,p)$-catoid that is sufficiently
    supported. Then $X$ can be equipped with an $(\omega,p)$-catoid
    structure.
  \item Let $Q$ be a complete lattice equipped with a
    multiplication that preserves arbitrary sups and has a unit. Let
    $Q^C$ and $Q$ be $(\omega,p)$-quantales with $Q$ sufficiently
    supported. Then $C$ is a local $(\omega,p)$-quantale.
  \end{enumerate}
\end{theorem}
These results specialise to correspondence triangles for strong
$(\omega,p)$-catoids and strong $(\omega,p)$-quantales as
usual. The complete Heyting structure can be used as part of the
definition of $(\omega,p)$-quantale.  Further, we obtain
correspondences for powerset quantales as instances.

\begin{corollary}\label{corollary:np-powerset}~
  \begin{enumerate}
  \item Let $C$ be a local $(\omega,p)$-catoid. Then $\Pow C$ is an
    $(\omega,p)$-quantale.
  \item Let $X$ be a set and $\Pow C$ an $(\omega,p)$-quantale in
    which $\id_i\neq \emptyset$ and the atoms are functional. Then $X$
    can be equipped with a local $(\omega,p)$-catoid structure.
  \end{enumerate}
\end{corollary}
Again there are special cases for strong structures, and
Proposition~\ref{P:omega-quantale-vs-globular-ka} can be extended to
show that every strong $(\omega,p)$-quantale is a globular
$(\omega,p)$-Kleene algebra à la~\cite[Definition
4.4.2]{CalkGMS20}. The laws of converse used in the definition of
globular $(\omega,p)$-Kleene algebras are subsumed by the laws of of
Dedekind quantales, as shown in Section~\ref{S:modular-quantales}; see
Sections~\ref{S:globular-semiring} and~\ref{S:globular-ka} for further
discussion of globular Kleene algebras.


\section{Correspondences for Powerset Quantales}\label{S:msg-quantale}

The correspondence triangles of Section~\ref{S:globular-convolution}
and \ref{S:dedekind-correspondence} specialise to powerset algebras
for the value quantale $Q=2$ of booleans
(Example~\ref{ex:boolean-quantale}), using the isomorphism between the
map $C\to 2$ and $\Pow C$.  We have already seen several examples of
powerset extensions in previous sections. Here we give standalone
correspondence proofs for the globular and the converse structure,
because they are interesting for higher-dimensional
rewriting. Most of the results in this sections have also been
  checked with Isabelle~\cite{CalkS24}.

The first corollary is an immediate instance of
Theorems~\ref{theorem:globular-corresp1} and
\ref{theorem:globular-corresp2}. 

\begin{corollary}\label{corollary:globular-lift}~
  \begin{enumerate}
  \item Let $C$ be a local $\omega$-catoid. Then
    $(\Pow C,\subseteq,\odot_i,C_i,\src_i,\tgt_i)_{0\le i<\omega}$ is
    an $\omega$-quantale.
  \item Let $X$ be a set and $\Pow X$ be an $\omega$-quantale in which
    $\id_i\neq \emptyset$. Then $X$ can be equipped with a local
    $\omega$-catoid structure.
    \item If $C$ is strong, then $\Pow C$ is strong and vice versa. 
  \end{enumerate}
\end{corollary}

We present a set-theoretic proof in Appendix~\ref{A:proofs} because it
might provide additional intuition. We emphasise the one-to-one
relationship between laws of $\omega$-catoids and those of
$\omega$-quantales in the proofs in Appendix~\ref{A:proofs} to
highlight the underlying correspondence, including the proof of this
corollary. The extension from $C$ to a modal quantale $\Pow C$ has
been described in Example~\ref{ex:catoid-quantale}, see
also~\cite[Theorem 6.5]{FahrenbergJSZ21a}. In the converse direction,
one can recover a catoid from the atom structure of the powerset
structure; its singleton sets, as discussed in the introduction:
$x\in y\odot z \Leftrightarrow \{x\} \subseteq \{y\}\ast \{z\}$, while
source and target maps correspond to domain and codomain maps on
singleton sets. Note that the relation between $\odot$ and
  $\ast$ fixes also $\subseteq$ as the order on the powerset quantale,
  and it is compatible with the standard definition of a lattice order
  in terms of sups and infs.

\begin{remark}
  A classical result by Gautam~\cite{Gautam57} shows that equations
  extend to the powerset level if each variable in an equation occurs
  precisely once in each side. These results have later been
  generalised by Grätzer and Whitney~\cite{GraetzerW84} (see
  also~\cite{Brink93} for an overview).  It is therefore no surprise
  that all the (unreduced) axioms of $\omega$-catoids extend directly
  to corresponding properties, which we have already derived from the
  $\omega$-quantale axioms in
  Lemma~\ref{lemma:2-quantale-props}. Nevertheless Gautam's result
  does not prima facie cover multioperations, let alone constructions
  of convolution algebras.
\end{remark}

Next we show how groupoids extend to Dedekind quantales and relation
algebras. This mainly reproduces Jónsson and Tarski's
results~\cite{JonssonT52} with groupoids based on catoids. In
addition, we provide explicit extensions for $\dom$ and $\cod$, which
can otherwise be derived in powerset quantales or relation algebras.
In any groupoid $C$, we write $X^\inv$ for $X\subseteq C$. First we
list some properties that are not directly related to the extension.

\begin{lemma}\label{lemma:groupoid-dom-lift}
  Let $C$ be a groupoid. For all $X,Y\in\Pow(C)$ and all $\mathcal X\in\Pow(\Pow(C))$,
  \begin{enumerate}
  \item $(\bigcup \mathcal{X})^\inv = \bigcup\{X^\inv\mid X
    \in\mathcal{X}\}$ and $(X \cup Y)^\inv = X^\inv\cup X^\inv$,
  \item $\src(X) \subseteq XX^\inv$ and $\tgt(X)\subseteq X^\inv X$,
  \item $X\subseteq XX^\inv X$,
  \item $\src(X)= C_0 \cap XX^\inv$ and $\tgt(X) = C_0\cap X^\inv X$,
  \item $\src(X)= C_0\cap X\top$ and $r(X) = C_0\cap \top X$,
    \item $X\top = \src(X)\top$ and $\top X= \top \tgt(X)$.
  \end{enumerate}
\end{lemma}

A proof can be found in Appendix~\ref{A:proofs}. Properties
(2) and (3) are of course the conditions for involutive quantales from
Section~\ref{S:modular-quantales}. Next we specialise the results of
Section~\ref{S:dedekind-correspondence},  revisiting and
  Example~\ref{ex:rel-dedekind} more formally and extending it to a
  correspondence.

\begin{corollary}\label{corollary:groupoid-lift} ~
  \begin{enumerate}
 \item   Let $C$ be a groupoid. Then $(\Pow C,\odot,\src,\tgt,(-)^\inv)$ is a
   Dedekind quantale.
 \item Let $X$ be a set and $\Pow X$ be a Dedekind quantale in
   which $\id_0\neq \emptyset$ and the atoms are functional. Then $X$
   can be equipped with a groupoid structure.
   \end{enumerate}
 \end{corollary}
 As mentioned in Example~\ref{ex:rel-dedekind}, the proofs are due to
 Jónsson and Tarski; we have also checked (1) with Isabelle.  The
 functionality requirement on atoms in (2), in particular, has been
 noticed and used by Jónsson and Tarski.

 \begin{remark}
   Jónsson and Tarski have considered relation algebras based on
   boolean algebras instead of complete lattices. They also use the
   residual law mentioned in Section~\ref{S:modular-quantales}
   instead of the modular or Dedekind law. This makes no
   difference. See~\cite{HirschH02,Maddux06} for details.
 \end{remark}
 
 Examples of powerset extensions for specific groupoids have been
 given in Example~\ref{ex:rel-dedekind}. This section provides a
 formal development that subsumes these results.


\section{Modal Operators and their Laws}\label{S:modalities}
 
We have already mentioned that modal semirings and modal quantales
carry their name because modal operators, akin to those of modal
logics, can be defined and related using the domain and codomain
structure. As already mentioned, the main application of the higher
quantales introduced in Section~\ref{S:2-quantales} and the Dedekind
quantales studied in Section~\ref{S:modular-quantales} are proofs in
higher-dimensional rewriting~\cite{CalkGMS20}. These require the modal
operators that can be defined on these quantales.  In this section, we
briefly recall these modal structures and present some new modal laws
that hold in higher quantales. The entire content of this section has
been formalised with Isabelle~\cite{CalkS23,CalkS24}.
 
Modal diamond operators can be defined in modal semirings or
quantales~\cite{DesharnaisS11, FahrenbergJSZ22a} as
 \begin{equation*}
   |\alpha\rangle \beta = \dom(\alpha\beta)\qquad\text{ and }\qquad \langle \alpha| \beta = \cod(\beta\alpha).
 \end{equation*} 
 By domain and codomain locality,
$|\alpha\rangle \beta = |\alpha\rangle \dom(\beta)$ and
$\langle\alpha| \beta = \langle\alpha| \dom(\beta)$.  Accordingly, we
mostly use those modal operators with $\beta$ a domain element.
See~\cite{FahrenbergJSZ22a,FahrenbergJSZ21a} for additional properties
for modal operators on quantales. In particular, we can ``demodalise''
diamonds at left-hand sides of inequalities: for $p,q\in Q_0$,
\begin{equation*}
  |\alpha\rangle p \le q \Leftrightarrow \alpha p\le q\alpha \qquad
  \text{ and }\qquad
  \langle \alpha| p \le q \Leftrightarrow p\alpha\le \alpha q.
\end{equation*}

In boolean modal quantales, and hence in modal powerset quantales, we
can define forward and backward modal box operators
\begin{equation*}
  [\alpha|p = \Sup\{q\mid |\alpha\rangle q \le p\}\qquad\text{ and }\qquad |\alpha]p
  = \Sup\{q\mid \langle \alpha|q \le p\},
\end{equation*}
for all $\alpha\in Q$ and $p\in Q_0$. (In arbitrary modal quantales,
we cannot guarantee that the sups, which feature in the definienda,
are again elements of $Q_0$.)

The modal box and diamond operators are then adjoints in the Galois
connections
\begin{equation*}
  |\alpha\rangle p \le q \Leftrightarrow p\le |\alpha]q\qquad \text{ and }\qquad
  \langle \alpha | p\le q \Leftrightarrow p \le [\alpha|q,
\end{equation*}
for all $\alpha\in Q$ and $p,q\in Q_0$. The additional demodalisation
laws $p \le [\alpha|q\Leftrightarrow \alpha p\le q\alpha$ and
$p \le |\alpha]q \Leftrightarrow p\alpha\le \alpha q$ are helpful in
deriving them. Further, in boolean modal quantales, boxes and diamonds
are related by De Morgan duality,
\begin{equation*}
  |\alpha]p = -|\alpha\rangle -p,\qquad |\alpha\rangle p = -|\alpha]-p,\qquad [\alpha|p = -
  \langle \alpha| -p,\qquad \langle \alpha|p = -[\alpha|-p.
\end{equation*}
Finally,  in any modal Dedekind quantale, $\langle \alpha| =|\alpha^\conv \rangle$,
$|\alpha\rangle = \langle \alpha^\conv |$, $[\alpha| = |\alpha^\conv ]$ and
$|\alpha]=[\alpha^\conv|$.

Modal diamond operators satisfy the following module-style laws.

\begin{lemma}\label{lemma:fdia-module}
  In every modal quantale,
  \begin{enumerate}
  \item $|\alpha\beta\rangle p = |\alpha\rangle |\beta\rangle p$,
  \item $|\alpha\sup \beta\rangle p = |\alpha\rangle p \sup |\beta\rangle p$,
  \item
    $|\alpha\rangle (p\sup q) = |\alpha\rangle p \sup |\alpha\rangle
    q$ and $|\alpha\rangle \bot = \bot$,
  \item $|\bot\rangle p = \bot$, $|1\rangle p = p$ and $|\alpha\rangle 1 =
    \dom(\alpha)$.
    \end{enumerate}
  \end{lemma}
  These laws hold already in domain semirings~\cite{DesharnaisS11},
  hence we do not present proofs.  Dual properties hold for backward
  diamonds, and for boxes, if the modal quantale is boolean.

  Beyond these one-dimensional properties, we list some new modal laws
  for the $\omega$-structure.  We write $\langle \alpha\rangle_i$ for
  either $|\alpha\rangle_i$ or $\langle \alpha|_i$ in the following
  lemma.

\begin{lemma}\label{lemma:fdia-glob}
  In every modal $\omega$-quantale $Q$, for all $0 \leq i < j <\omega$,
\begin{enumerate}
\item $\fbDia{\alpha}_i\fDia{\beta}_j \gamma= \fbDia{\alpha}_i(\beta\cdot_j \gamma)$ and
  $\fbDia{\alpha}_i\bDia{\beta}_j \gamma = \fbDia{\alpha}_i(\gamma \cdot_j \beta) $,
\item $\fDia{\alpha}_i\fbDia{\beta}_j \gamma \leq \fDia{\alpha}_i(\dom_i(\beta) \cdot_j
  \dom_i(\gamma))$ and $\bDia{\alpha}_i\bDia{\beta}_j \gamma \leq \bDia{\alpha}_i(\cod_i(\gamma) \cdot_j \cod_i(\beta))$,
\item If $\alpha = \dom_k (\alpha)$ for some $k\leq j$, then 
  \begin{equation*}
    \fDia{\alpha}_i \fDia{\beta}_j \gamma \leq \fDia{\alpha}_i \beta
    \cdot_j \fDia{\alpha}_i \gamma\qquad\text{ and }\qquad
\langle \alpha|_i\langle \beta|_j \gamma \le \langle \alpha| \beta
\cdot_j \langle \alpha| \gamma,
\end{equation*}
\item
  $\fDia{\alpha}_j \fbDia{\beta}_i \gamma \le \fDia{\alpha}_j
  \fbDia{\dom_j(\beta)}_i\gamma$ and
  $\bDia{\alpha}_j\fbDia{\beta}_i \gamma \le
  \bDia{\alpha}_j\fbDia{\cod_j(\beta)}_i\gamma$, and equality holds if
  $Q$ is strong,
\item
  $\dom_i(\alpha) \cdot_i \fbDia{\beta}_j\gamma \leq
  \fbDia{\dom_i(\alpha) \cdot_i \beta}_j( \dom_i(\alpha) \cdot_i
  \gamma)$ and, whenever $Q$ is strong, then
  $ \fbDia{\alpha}_j\beta \cdot_i \dom_i(\gamma) \leq \fbDia{\alpha
    \cdot_i \dom_i(\gamma)}_j( \beta\cdot_i \dom_i(\gamma) )$.
\end{enumerate}
\end{lemma}
Proofs can be found in Appendix~\ref{A:proofs}.


\section{Higher Semirings}\label{S:globular-semiring}

In this section we weaken our value algebras from quantales to dioids,
which are additively idempotent semirings. Convolution semirings have
been studied widely in mathematics and computer science; the higher
globular semirings in~\cite{CalkGMS20} form a starting point for the
investigations in this article. Here we generalise from the powerset
structures in~\cite{CalkGMS20} to convolution algebras. As usual in
the construction of formal power series or convolution algebras, some
restriction on the domain algebra or the function space need to be
imposed. Beyond the relationship with~\cite{CalkGMS20}, the
consideration of $\omega$-semirings and $\omega$-Kleene algebras in
this and the following section also sheds some light on the design
space of algebras for higher-dimensional rewriting. While
$\omega$-quantales are certainly more expressive, $\omega$-semirings
and $\omega$-Kleene algebras might lead to stronger computational
properties such as decision procedures. But an exploration of this
design space in rewriting applications remains beyond the scope of
this article.

In convolution algebras on $S^C$, where $S$ is an additively
idempotent semiring, the complete lattice structure of the value
quantale is replaced by a semilattice.  To compensate for the lack of
infinite sups in convolutions
\begin{equation*}
  (f\ast g)(x)= \Sup_{x\in y\odot z} f(y)\cdot g(z),
\end{equation*}
when infinitely many pairs $(y,z)$ satisfy the ternary relation
$x\in y\odot z$, we require \emph{finitely decomposable} catoids
$C$. That is, for each $x \in C$ the fibre
$\odot^{-1}(x) = \{(y,z)\mid x \in y \odot z\}$ must be finite.  Now
algebras with finite sups such as additively idempotent semirings
suffice.  This is standard, for instance, for the incidence algebras
in combinatorics~\cite{Rota64}, where such a finiteness condition is
imposed on intervals on the real line or a poset. In our context this
means that modal value quantales can be replaced by modal value
semirings~\cite{DesharnaisS11} if catoids are finitely
decomposable. Similar adaptations have been made for concurrent
quantales and concurrent semirings~\cite{CranchDS21}.

\begin{example}[\cite{CranchDS21}]
  In the $Q$-valued shuffle language $2$-quantale from
  Example~\ref{ex:cat-quantale}, the fibres
  $\{(u,v) \mid w= uv\}$ and $\{(u,v) \mid w= u\| v\}$ are finite for
  every word $w$. It therefore suffices that $Q$ is a $2$-semiring, as
  defined below.
\end{example}

The main results in this section show how finitely decomposable
$\omega$-catoids extend to convolution $\omega$-semirings, where the
$\omega$-structure is defined precisely as for $\omega$-quantales, and
to (modal) semirings with a converse structure. These results prepare
for the study of $\omega$-Kleene algebras and (modal) Kleene
algebras with converse in the following section, and in particular for
a comparison with the higher globular Kleene algebras previously
introduced~\cite{CalkGMS20}.

We start with a summary of the structure of $\alpha$-dioids. In a
nutshell, their overall structure is the same as for
$\alpha$-quantales, except that the complete lattice in the latter is
replaced by a semilattice in the former.
\begin{itemize}
\item A \emph{dioid} is an additively idempotent semiring
  $(S,+,\cdot,0,1)$. Thus $(S,+,0)$ is a semilattice with least element
  $0$. We write $\le$ for its order.
\item A \emph{modal semiring}~\cite{DesharnaisS11} is an additively
  idempotent semiring (or \emph{dioid}) equipped with operations
  $\dom,\cod:S\to S$ that satisfy precisely the modal quantale axioms
  for domain and codomain up to notational differences: we
  write $\dom(0)=0$, replacing $\bot$ by $0$,
  $\dom(\alpha+\beta)=\dom(\alpha)+\dom(\beta)$, replacing $\sup$ by
  $+$, and likewise for $\cod$.

\item An \emph{$\alpha$-semiring} is a structure
  $(S,+,0,\cdot_i,1_i,\dom_i,\cod_i)_{0\le i < \alpha}$, for an
  ordinal $\alpha \in \{0,1,\dots,\omega\}$, such that the
  $(S,+,0,\cdot_i,1_i,\dom_i,\cod_i)$ are modal semirings and we
  impose the same interchange and domain/codomain axioms as for
  $\alpha$-quantales.

\item An $\alpha$-semiring is \emph{strong} if it satisfies the same
  domain/codomain axioms as for strong $\alpha$-quantales.
\end{itemize}
A formal list of axioms can be found in our Isabelle
theories~\cite{CalkS24} and their proof document.  Most properties
derived for $\omega$-quantales in previous sections are already
available in $\omega$-semirings, in
particular Lemmata~\ref{lemma:2-quantale-props}, \ref{lemma:mod-props},
\ref{lemma:fdia-module} and \ref{lemma:fdia-module}, the results on
chains of $C_i$ cells and $\omega$-quantales as filtrations, and the
construction of $n$-quantales by truncation.

Along the lines of Section~\ref{S:modular-quantales} we can also 
define a converse structure $(-)^\conv:S\to S$ on a dioid $S$.
\begin{itemize}
\item In an \emph{involutive dioid} $S$ we impose the axioms
  $\alpha^{\conv\conv} = \alpha$,
  $(\alpha+\beta)^\conv = \alpha^\conv + \beta^\conv$ and
  $(\alpha\beta)^\conv = \beta^\conv \alpha^\conv$.
\item In a \emph{dioid with converse} we also require the strong
  Gelfand law $\alpha\le \alpha\alpha^\conv \alpha$.
\item \emph{Modal involutive semirings} and
  \emph{modal semirings with converse} are then defined in the obvious way.
\item A \emph{modal semiring with strong converse} is a
  modal involutive semiring in which $\dom(\alpha)\le \alpha\alpha^\conv$ and
 $\cod(\alpha)\le \alpha^\conv \alpha$.  
\end{itemize}
The involutive dioid axioms imply $0^\conv = 0$ and $1^\conv = 1$, but
are too weak to relate $\dom$ and $\cod$ in the modal case. Dioids
with converse have been introduced in~\cite{BloomES95}. These provide
the domain/codomain interaction expected.

\begin{lemma}\label{lemma:dom-cod-conv}
  In every modal semiring with converse,
  \begin{enumerate}
  \item $\dom(\alpha)^\conv= \dom(\alpha)$ and $\cod(\alpha)^\conv = \cod(\alpha)$.
    \item $\dom(\alpha^\conv) = \cod(\alpha)$ and $\cod(\alpha^\conv)
      = \dom(\alpha)$.   
  \end{enumerate}
 \end{lemma}

 See Appendix~\ref{A:proofs} for proofs. Modal semirings with strong
 converse have been proposed in ~\cite{CalkGM21}. The strong Gelfand
 property holds in this setting:
 $\alpha =\dom(\alpha)\alpha\le \alpha\alpha^\conv \alpha$. The strong
 converse axioms can be seen as shadows of the explicit definitions of
 $\dom$ and $\cod$ in Dedekind quantales in the absence of $\inf$.

 For finitely decomposable $\omega$-catoids $C$, where each of the
 $\omega$ underlying catoids has this property, our main
 correspondence triangles from Sections~\ref{S:globular-convolution}
 and \ref{S:dedekind-correspondence} transfer to $\omega$-semirings
 and semirings with converse. Here we consider only the extensions
 from groupoids to convolution algebras. First we present a corollary
 to Theorem~\ref{theorem:globular-corresp1}, but to deal with the sups
 in the definition of $\Dom$ and $\Cod$, we need to impose another
 restriction: a catoid $C$ has \emph{finite valency} if for each
 $x\in C_0$ the sets $\{y\in C\mid s(y)=x\}$ and
 $\{y\in C\mid t(y)= x\}$ are finite. An $\omega$-catoid has finite
 valency if each of the underlying catoids has this property.

\begin{corollary}\label{corollary:globular-lifting-semiring}
  Let $C$ be a finitely decomposable local $\omega$-catoid of finite
  valency and $S$ a (strong)
  $\omega$-semiring. Then $S^C$ is a (strong) $\omega$-semiring.
\end{corollary}
\begin{proof}
  In the construction of the convolution algebra in the proof of
  Theorem~\ref{theorem:globular-corresp1}, it is routine to check that
  all sups remain finite when $C$ is finitely decomposable. In
  particular, all sups in the definitions of  $\Dom$ and $\Cod$ are
  finite because of finite valency. 
\end{proof}

Next we consider the converse structure. 

\begin{proposition}\label{theorem:groupoid-dioidinv}
  Let $C$ be a finitely decomposable groupoid.
  \begin{enumerate}
  \item If $S$ is a dioid with converse, then so is $S^C$.
  \item If $S$ is a modal semiring with strong converse, then so is
    $S^C$, whenever $C$ has finite valency.
  \end{enumerate}
\end{proposition}
\begin{proof}
  First note that the infinite sups in the proof of
  Proposition~\ref{prop:groupoid-involquantale} can be replaced by
  finite sups when $C$ is finitely decomposable and has finite
  valency. It remains to extend the strong Gelfand property for (1)
  and the axioms $\dom(\alpha)\le \alpha\alpha^\conv$ and
  $\cod(\alpha)\le \alpha^\conv \alpha$ for (2). In the proofs we
  still write $\Sup$, but tacitly assume that all sups used are finite
  and can therefore be represented using $+$.

  For dioids with converse,
  \begin{align*}
    f\ast (f^\conv \ast f)
         &= \Sup_{w,x,y,z\in C} f(x)f(y^\inv)^\conv g(z) [w \in x\odot
           y\odot z] \delta_w\\
     &= \Sup_{w,x,y,z\in C} f(x)f(y)^\conv g(z) [w \in x\odot
       y^\inv\odot z] \delta_w\\
     &\ge \Sup_{w\in C} f(w)f(w)^\conv g(w) [w \in w\odot
       w^\inv\odot w] \delta_w\\
    &\ge \Sup_{w\in C} f(w) [w\in \src(w)w]\delta_w\\
    &= \Sup_{w\in C} f(w)\delta_w\\
    &= f.
  \end{align*}
  
  For semirings with strong converse,
  \begin{align*}
    \Dom(f)
    &= \Sup_{x\in C} \dom(f(x))\delta_{\src(x)}\\
    &\le \Sup_{x\in C} f(x)f(x)^\conv [\src(x) \in x\odot
      x^\inv]\delta_{x}\\
      &\le \Sup_{x,y,z\in C} f(y)f(z^\inv)^\conv [x \in y\odot
        z^\inv]\delta_{x}\\
    &= \Sup_{x\in C} (f\ast f^\conv) (x)\delta_{x}\\
      &= f\ast f^\conv
  \end{align*}
  and the proof for $\Cod$ follows by opposition.
\end{proof}

Next we present two examples that separate the three converse
structures introduced above in the modal case. 

\begin{example}\label{ex:invol-mod-dom-cod}
  In the involutive modal semiring given by
  
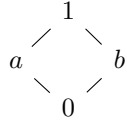
\begin{equation*}
    \begin{tikzcd}[column sep = tiny, row sep = tiny]
      & 1 &\\
      a\arrow[ur,-] && b\arrow[ul,-]\\
      &0\arrow[ul,-]\arrow[ur,-] &
    \end{tikzcd}
  \end{equation*}
  with $\cdot$ as inf, $\dom=\id=\cod$ and $a^\conv =b$, $b^\conv = a$
  we have
  \begin{equation*}
    \dom(a^\conv)=\dom(b)=b \neq a =\cod(a)
  \end{equation*}
  and likewise $\cod(a^\conv)\neq \dom(a)$.  Hence it is not a modal
  semiring with converse.
\end{example}
Similar examples show that $\dom(\alpha)^\conv$ and
$\cod(\alpha)^\conv$ need not be equal to $\dom(\alpha)$ and
$\cod(\alpha)$ in involutive modal semirings, respectively.

\begin{example}
  In the modal semiring with converse given by $0<a<1$, $aa=a$,
  $dom(a)=\cod(a)=1$ and $(-)^\conv =\id$, we have
  $\dom(a)= \cod(a) 1 > a=aa^\conv = a^\conv a$. It is therefore not a
  modal semiring with strong converse.
\end{example}

\begin{remark}
  In an involutive domain semiring or domain semiring with converse
  $S$, one can of course define codomain explicitly as
  $\cod(\alpha)=\dom(\alpha^\conv)$, from which
  $\dom(\alpha)=\cod(\alpha^\conv)$ follows, but then $S_\dom$ need
  not be equal to $S_\cod$: $\dom\circ \cod = \cod$ may hold, but
  $\cod\circ \dom=\dom$ may fail. We do not consider such alternative
  definitions of modal semirings with converse any further.
\end{remark}

\begin{remark}
  As in Section~\ref{S:n-structures}, we can consider
  $(\omega,p)$-semirings and $(\omega,p)$-Kleene algebras. We start
  with the former and outline the latter in
  Section~\ref{S:globular-ka}.

  An \emph{$(\omega,p)$-dioid} is an $\omega$-dioid $S$ equipped with
  operations $(-)^{\conv_j}:S \to S$ for all $p< j\leq n $ such that
  $(\alpha)^{\conv_j}\in S_j$ for all $\alpha\in S_j$ and such that
  the involution axioms and $\alpha\le \alpha\alpha^{\conv_j} \alpha$
  hold for $p< j \le n$.

  Corollary \ref{corollary:np-powerset} then specialises further to
  $(\omega,p)$-semirings owing to
  Corollary~\ref{corollary:globular-lifting-ka}. We can still extend
  to $(\omega,p)$-semirings, that is, whenever $C$ is a finitely
  decomposable local $(\omega,p)$-catoid of finite valency and $S$ an
  $(\omega,p)$-semiring, then $S^C$ is an
  $(\omega,p)$-semiring. Truncations at dimension $n$ work as expected
  and the globular $(n,p)$-semirings introduced in~\cite{CalkGMS20}
  arise as special cases.
\end{remark}

Finally, we compare $\omega$-semirings with the previous slightly
different axiomatisation of higher globular semirings~\cite[Definition
3.2.6]{CalkGMS20}. These have been proved sound with respect to a
powerset model of path $n$-categories~\cite[Section 3.3]{CalkGMS20},
which is the basis of higher-dimensional rewriting, but axioms have
been introduced in an ad hoc fashion with a view on higher-dimensional
rewriting proofs. Our correspondence triangles with respect to
$\omega$-catoids and $\omega$-categories yield a more systematic
structural justification of the $\omega$-quantale and
$\omega$-semiring axioms.

One difference between the globular variant and ours is that the
morphism axioms
$\dom_i(\alpha\cdot_j \beta) \le \dom_i(\alpha)\cdot_j \dom_i(\beta)$
and
$\cod_i(\alpha\cdot_j \beta) \le \cod_i(\alpha)\cdot_j \cod_i(\beta)$
are missing in the globular axiomatisation, while they are irredundant
among the $\omega$-semiring and $\omega$-quantale axioms
(Example~\ref{ex:d10-counter}). Yet they can be derived in our
construction of convolution $\omega$-quantales over any
$\omega$-catoid and hence hold in the powerset model of path
$n$-categories. A second difference is that the previous
axiomatisation contains redundant assumptions and axioms, for instance
on the bounded distributive lattice structure of $S_i$ or the
relationships between quantalic units at different dimensions, which
are now derivable.  This discussion can be summarised as follows.

\begin{proposition}\label{P:omega-sr-vs-globular-sr}
  Every strong $\omega$-semiring is a globular
  $\omega$-semiring. Every globular $\omega$-semiring, which
  satisfies
  $\dom_i(\alpha\cdot_j \beta) \le \dom_i(\alpha)\cdot_j
  \dom_i(\beta)$ and
  $\cod_i(\alpha\cdot_j \beta) \le \cod_i(\alpha)\cdot_j
  \cod_i(\beta)$, is a strong $\omega$-semiring.
\end{proposition}
The result generalises to $(\omega,p)$-semirings with or without
strong converses.


\section{Higher Kleene Algebras}\label{S:globular-ka}

In this section we extend the results in the previous one from
semirings to Kleene algebras $K$, adding a Kleene star
$(-)^\ast:K\to K$. By contrast to the quantalic Kleene star
$\alpha^\ast = \Sup_{i\ge 0} \alpha^i$ in
Section~\ref{S:modal-quantales}, it is modelled in terms of least
fixpoints. Intuitively, the Kleene star models an unbounded finite
iteration or repetition of a computation or a sequence of rewrite
steps. In the quantale of binary relations in
Example~\ref{ex:catoid-quantale}, for instance, the Kleene star of a
binary relation models its reflexive-transitive closure, in particular
the iterated execution of a rewrite relation. In the same example, the
Kleene star on the quantale of paths models the iterative gluing of
paths in a given set, thus in particular of rewriting sequences. The
same observations can be made about Kleene algebras of relations or
sets of paths.  For evidence that Kleene stars are crucial for
reasoning algebraically about higher-dimensional rewriting properties
see~\cite{CalkGMS20}.

The axioms of $\alpha$-Kleene algebras are not as straightforward as
those for $\alpha$-dioids; additional identities capturing the
interaction between higher stars, domain and codomain are needed. For
convolution algebra constructions, a star on convolution Kleene
algebras $K^C$ has previously been defined only for finitely
decomposable catoids $C$ with a single unit that satisfy a certain
grading condition~\cite{CranchDS21}. This does not cover the
$\alpha$-semirings with multiple units in this paper. Instead of
general convolution algebras, we therefore only consider powerset
extensions in this section, where the finite decomposibility of
catoids is not needed.

A \emph{Kleene algebra} is a dioid $K$ equipped with a
star operation $(-)^\ast :K\to K$ that satisfies the star unfold and
star induction axioms
\begin{gather*}
  1+\alpha\cdot \alpha^{\ast} \le \alpha^{\ast},\qquad \gamma + \alpha\cdot\beta \le \beta
  \Rightarrow \alpha^{\ast} \cdot \gamma \le \beta
\end{gather*}
and their opposites, where the arguments in compositions have been
swapped. A \emph{modal Kleene algebra} is simply a Kleene algebra
which is also a modal semiring.

The star unfold and induction axioms of Kleene algebras thus model
$\alpha^\ast$ in terms of the least pre-fixpoints of the maps
$x\mapsto 1+\alpha \cdot x$ and $x\mapsto 1 + x\cdot \alpha$, and
hence their least fixpoints.  Every quantale is a Kleene algebra: the
quantalic Kleene star defined in Section~\ref{S:modal-quantales}
satisfies the Kleene algebra axioms.

An adaptation of the definition of
$\omega$-quantales to Kleene algebras is straightforward -- except for
the last two axioms.

For an ordinal $\alpha\in \{0,1,\dots,\omega\}$, an
\emph{$\alpha$-Kleene algebra} is an $\alpha$-semiring $K$
equipped with Kleene stars $(-)^{\ast_i}:K\to K$ that satisfy the
usual star unfold and star induction axioms, for all $0\le i< j <\alpha$,
\begin{gather*}
  1_i+\alpha\cdot_i \alpha^{\ast_i} \le \alpha^{\ast_i},\qquad \gamma + \alpha\cdot_i \beta \le \beta \Rightarrow \alpha^{\ast_i} \cdot_i \gamma \le \beta,\\
  1_i+\alpha^{\ast_i}\cdot_i \alpha\le \alpha^{\ast_i},\qquad
  \gamma + \beta\cdot_i \alpha \le \beta \Rightarrow \gamma \cdot_i
  \alpha^{\ast_i} \le \beta,\\
    \dom_i (\alpha) \cdot_i \beta^{\ast_j} \le (\dom_i (\alpha)\cdot_i
  \beta)^{\ast_j},\qquad \alpha^{\ast_j}\cdot_i \cod_i (\beta) \le
  (\alpha\cdot_i \cod_i (\beta))^{\ast_j}.
\end{gather*}
An $\alpha$-Kleene algebra is \emph{strong} if the underlying
$\alpha$-semiring is, and for all $0\le i<j<\alpha$, 
\begin{gather*}
  \dom_j (\alpha) \cdot_i \beta^{\ast_j} \le (\dom_j (\alpha)\cdot_i
  \beta)^{\ast_j},\qquad \alpha^{\ast_j}\cdot_i \cod_j (\beta) \le
  (\alpha\cdot_i \cod_j (\beta))^{\ast_j}.
\end{gather*}

\begin{remark}\label{remark:strong-remark}
  The axioms mentioning domain and codomain are derivable in (strong)
  $\alpha$-quantales. In $\alpha$-Kleene algebras we have neither
  proofs nor counterexamples to show redundancy or irredundancy of
  these axioms with respect to the remaining ones, yet these laws are
  needed for coherence proofs in higher-dimensional
  rewriting~\cite{CalkGMS20}.
\end{remark}

The converse structure on Kleene algebras is inherited from
dioids. This leads immediately to \emph{involutive Kleene algebras},
\emph{Kleene algebras with converse}, their modal variants and
\emph{modal Kleene algebras with strong converse}.

\begin{lemma}\label{lemma:involutive-ka}
  In every involutive Kleene algebra, $\alpha^{\ast\conv} =
 \alpha^{\conv\ast}$. 
\end{lemma}
A proof can be found in Appendix~\ref{A:proofs}, as usual.

We now present a correspondence result for higher Kleene algebras --
in the special case of powerset extensions. As usual, we consider
$\omega$-structures. As every quantale is a Kleene algebra and the
Kleene star on a set is simply a union of powers, it is an immediate
consequence of Corollary~\ref{corollary:globular-lift}.

\begin{corollary}\label{corollary:globular-lifting-ka}~
  \begin{enumerate}
  \item  Let $C$ be a  local $\omega$-catoid. Then $\Pow C$
    is an $\omega$-Kleene algebra.
  \item Let $X$ be a set and $\Pow X$ an $\omega$-Kleene
    algebra in which $\id_i\neq 0$. Then $X$ can be equipped
      with a local $\omega$-catoid structure.
      \item If $C$ is strong, then $\Pow C$ is strong and vice versa.
  \end{enumerate}
\end{corollary}

Similarly, the following correspondence result is a specialisation of
Corollary~\ref{corollary:groupoid-lift}.

\begin{corollary}\label{corollary:groupoid-lift2} ~
  \begin{enumerate}
  \item Let $C$ be a groupoid. Then $\Pow C$ is a modal Kleene algebra with
    strong converse.
  \item Let $X$ be a set and $\Pow X$ a modal Kleene algebra
    with strong converse in which $\id_0\neq \emptyset$ and all atoms
    are functional. Then $X$ can be equipped with a
    groupoid structure.
   \end{enumerate}
 \end{corollary}
 \begin{proof}
   For (1), it remains to extend the strong converse axioms, as the
   modular law is not available. This has been done in
   Lemma~\ref{lemma:groupoid-dom-lift}.

   For (2), note that the strong converse axioms are used in the proof
   of Theorem~\ref{theorem:conv-quantale-groupoid} instead of the
   Dedekind or modular law, but the strong Gelfand law would be too
   weak. For a direct proof,
   $x\odot x^\inv = \{x\}\cdot \{x\}^\conv =\dom(\{x\}) = \{\src(x)\}$
   and likewise for the target axiom, using the strong converse axiom
   for domain in the second step (as an equation because atoms are
   functional).
\end{proof}

Typical examples of powerset $\omega$-Kleene algebras are subalgebras
of powerset $\omega$-quantales generated by some finite set and closed
under the operations of $\omega$-Kleene algebras. Languages over a
finite alphabet $\Sigma$, for instance, form a quantale, whereas the
regular languages generated by $\Sigma$ form a Kleene algebra, but not
a quantale.

\begin{remark}
The definition of \emph{$(\omega,p)$-Kleene algebra} is a
straightforward extension of that of $(\omega,p)$-semiring. The fact
that $\alpha^{\conv_ j} = \alpha$ for all $\alpha\in S_{\dom_i}$ (or
$K_j$) with $i< j$ now follows from Lemma~\ref{lemma:dom-cod-conv}.

Corollary \ref{corollary:np-powerset} now specialises to
$(\omega,p)$-Kleene algebras owing to
Corollary~\ref{corollary:globular-lifting-ka}. In particular, an
$(\omega,p)$-Kleene algebra with strong converse is needed to recover
the $(\omega,p)$-catoid among the atoms.  For the more general
convolution algebra constructions, the limitations mentioned in
Section~\ref{S:globular-ka} remain; an extension to Kleene algebras
requires further thought. Truncations at dimension $n$ work as
expected. 
\end{remark}

Finally, we briefly compare $\omega$-Kleene algebras with globular
Kleene algebras~\cite[Definitions 3.2.7]{CalkGMS20}. As both are based
on the higher semirings discussed in the previous section, we focus on
the star axioms. On one hand, as discussed in
Remark~\ref{remark:strong-remark}, we need a strong $\omega$- Kleene
algebra to derive two of the globular $\omega$-Kleene algebra
axioms. On the other hand, two globular $\omega$-Kleene algebra axioms
are derivable in the context of $\omega$-Kleene algebras, as the
following lemma shows.

\begin{lemma}\label{lemma:nKA-star}
  In every $\omega$-Kleene algebra,
  $(\alpha \cdot_j \beta)^{\ast_i} \le \alpha^{\ast_i} \cdot_j
  \beta^{\ast_i}$ for all $0\le i <j <\omega$.
\end{lemma}
See Appendix~\ref{A:proofs} for a proof. The following result
summarises this discussion.

\begin{proposition}\label{P:omega-ka-vs-globular-ka}
  Every strong $\omega$-Kleene algebra is a globular
  $\omega$-Kleene algebra. Every
  globular $\omega$-Kleene algebra, which satisfies
  $\dom_i(\alpha\cdot_j \beta) \le \dom_i(\alpha)\cdot_j
  \dom_i(\beta)$ and
  $\cod_i(\alpha\cdot_j \beta) \le \cod_i(\alpha)\cdot_j
  \cod_i(\beta)$, is a strong $\omega$-Kleene algebra.
\end{proposition}
This proposition extends to strong $(\omega,p)$-Kleene algebras and
globular $(\omega,p)$-Kleene algebras, as introduced
in~\cite[Definition 4.4.2]{CalkGMS20} for the case $(n,p)$,
using either the notion of converses or strong converses discussed
in the previous section.


\section{Conclusion}\label{S:conclusion}

This paper combines two lines of research on higher globular algebras
for higher-dimensional rewriting and on the construction of
convolution algebras on catoids and categories. More specifically, we
have introduced $\omega$-catoids and $\omega$-quantales, and
established correspondence triangles between them. These extend and
justify the axioms of higher globular Kleene
algebras~\cite{CalkGMS20}, which have previously been used for
coherence proofs in higher-dimensional rewriting, up to some
modifications. We have also introduced several extensions and
specialisations of these constructions, in particular to
$(\omega,p)$-catoids and $(\omega,p)$-quantales and to variants of
such quantales based on semirings or Kleene algebras. While the
technical focus has been on convolution algebras, which often lead to
interesting applications in quantitative program semantics and
verification in computer science, we currently have no use for
convolution $\omega$-quantales or $(\omega,p)$-quantales with value
quantales different from the quantale of booleans. In the latter case,
however, we can use power set $\omega$-quantales instead of globular
Kleene algebras to reason about higher-dimensional rewriting systems
with greater flexibility and expressive power.  A detailed
introduction to the relationship with higher-dimensional rewriting, to
globular Kleene algebras and to their use in coherence proofs such as
higher Church-Rosser theorems and higher Newman's lemmas can be found
in~\cite{CalkGMS20}.

Our results add a new perspective to higher-dimensional rewriting and
they yield new tools for reasoning with higher categories. But further
work is needed for exploring them in practice. A first line of
research could consider categorical variations of the strict globular
$\omega$-catoids and $\omega$-quantales introduced in this work to
provide a categorical framework for the construction of coherence
proofs and more generally of polygraphic resolutions, which are
cofibrant replacements in higher categories:
\begin{itemize}
\item It is worth considering cubical versions of $\omega$-catoids and
  $\omega$-quantales, as confluence diagrams in higher-dimensional
  rewriting have a cubical shape~\cite{Lucas2017,Lucas20} and proofs
  of higher confluence properties can be developed more naturally in a
  cubical setting.  A single-set axiomatisation of cubical categories
  has recently been developed~\cite{MalbosMassacrierStruth2024}. It
  remains to generalise it to cubical catoids, to introduce cubical
  quantales and to study their correspondences along the lines of the
  globular case.  Cubical catoids and quantales may also be beneficial
  for explicit constructions of resolutions.  Beyond that, we are
  interested in applications to precubical sets and higher-dimensional
  automata, where Kleene algebras or quantales describing their
  languages remain to be
  defined~\cite{FahrenbergJSZ21b,FahrenbergJSZ22b}, and to cubical
  $\omega$-categories with connections~\cite{AlAglBS02}, where
  single-set formalisations might be of interest.
\item To expand the constructions in this article to algebraic
  rewriting such as string, term, linear or diagrammatic rewriting,
  the development of catoids and quantales internal to categories of
  algebras over an operad would be needed.
\item The homotopic properties of algebraic rewriting require
  weakening the exchange law of higher categories
  \cite{Forest21,ForestMimram2022} and, in our context, Gray-variants
  of the $\omega$-catoids and $\omega$-quantales, where interchange law
  holds only up to coherent isomorphism~\cite{GordonPowerRoss95}.
\end{itemize}

A second line of research could investigate the link of the
correspondence triangles between catoids, value quantales and
convolution quantales with duality results for convolution algebras
beyond the Jónsson-Tarski case and similar structural results, and the
categorification of such an approach. Catoids, for instance, are
equivalent to monoids in the monoidal category $\Rel$ with the
standard tensor and unit. Related to this are questions about free
powerset or convolution $\omega$-quantales generated by polygraphs, or
computads, and about coherent rewriting properties of these internal
monoids.

A third line of research could consider the formalisation of higher
category theory, and higher-dimensional rewriting support for these, in the
single-set framework of catoids developed
in~\cite{Struth23,CalkS23,CalkS24} or by other means with proof
assistants such as Coq, Lean or Isabelle. The formalisation of the
theorems in~\cite{CalkGMS20} could serve as stepping stones towards
coherence theorems like Squier's theorem in higher
dimensions~\cite{GuiraudMalbos09,GuiraudMalbos18} or the computation
of resolutions by rewriting in categorical and homological
algebra~\cite{GuiraudMalbos12advances,GuiraudHoffbeckMalbos19,MalbosRen23}.

Finally, the study of domain, codomain and converse in the setting of
quantales leads to interesting questions about Dedekind quantales and
the interplay of these operations in variants of
allegories~\cite{FreydS90}. These will be addressed in successor
articles.


\vspace{\baselineskip}

\noindent \textbf{Acknowledgement} The authors would like to thank
James Cranch, Uli Fahrenberg, Éric Goubault, Amar Hadzihasanovic,
Christian Johanson, Tanguy Massacrier and Krzysztof Ziemia\'nski for
interesting discussions and the organisers of the GETCO 2022
conference and the Nordic Congress of Mathematicians 2023 for the
opportunity to present some of the results in this article. The fourth
author would like to thank the Plume team and the Computer Science
Department of the ENS de Lyon for supporting a short visit at the LIP
laboratory in the final stages of this work.
  

\newpage

\bibliographystyle{alpha}
\bibliography{glob-quantale}

\appendix


\section{Eckmann-Hilton-Style Collapses}\label{A:eckmann-hilton}

Strengthening the weak homomorphism axioms
\begin{equation*}
  \src_i(x\odot_j y) \subseteq \src_i(x) \odot_j \src_i(y)\qquad\text
  { and }\qquad
  \tgt_i(x\odot_j y) \subseteq \tgt_i(x) \odot_j \tgt_i(y),
\end{equation*}
for $0\le i<j<\omega$ to equations in the $\omega$-catoid axioms
collapses the structure.  The equational homomorphism laws for
$\src_j$ and $\tgt_j$ and $\odot_i$ for $i<j$ in $\omega$-categories
are therefore rather exceptional. In this Appendix,
following~\cite{FahrenbergJSZ21a}, we call
\emph{$\src\tgt$-multimagma} a catoid in which associativity of
$\odot$ has been forgotten. We extend this notion to
$\omega$-$\src\tgt$-magmas in the obvious way.

\begin{lemma}\label{lemma:collapse1}
  If the inclusions
  \begin{equation*}
    \src_i(x\odot_j y) \subseteq \src_i(x) \odot_j
    \src_i(y)\qquad\text{ and }
    \qquad\tgt_i(x\odot_j y) \subseteq \tgt_i(x) \odot_j \tgt_i(y),
    \end{equation*}for
  $0\le i<j<\omega$ are replaced by
  equations in the axiomatisation of $\omega$-$\src\tgt$-multimagmas, then
 \begin{enumerate}
 \item $\src_i = \src_j$ and $\tgt_i=\tgt_j$,
   \item $\src_i = \tgt_i$ and $\src_j = \tgt_j$. 
 \end{enumerate}
\end{lemma}
\begin{proof}
  For $\src_i = \src_j$,
  \begin{align*}
    \src_i(\src_j(x) \odot_j \src_i(\src_j(x))) &=
    \src_i(\src_j(x))\odot_j \src_i(\src_i(\src_j(x)))\\
    &= \src_j(\src_i(\src_j(x)))\odot_j \src_i(\src_j(x))\\
    &=\{\src_i(\src_j(x))\},
  \end{align*}
  thus $\src_i(\src_j(x) \odot_j \src_i(\src_j(x))) \neq \emptyset$
  and therefore $\Delta_j(\src_j(x),\src_i(\src_j(x)))$. It follows
  that
  \begin{equation*}
    \src_j(x)= \tgt_j(\src_j(x)) = \src_j(\src_i(\src_j(x)))
    = \src_j(\src_j(\src_i(x)))
    =\src_j(\src_i(x))
    =\src_i(x).
  \end{equation*}
  By opposition, therefore, $\tgt_i=\tgt_j$. Also,
  $\src_i (x)= \tgt_j(\src_i(x)) = \src_i(\tgt_j(x)) = \src_i(\tgt_i(x))
  =\tgt_i(x)$, and $\src_j=\tgt_j$ then follows from the previous properties.
\end{proof}

\begin{lemma}\label{lemma:collapse2}
  If the same replacement is made for $\omega$-categories, then
  $\odot_i=\odot_j$ for $0\le i<j\le \omega$ and both operations commute. 
\end{lemma}
\begin{proof}
  If $x\odot_i y= \emptyset$, then $x\odot_i y \subseteq x\odot_j
  y$. Otherwise, if $x\odot_i y\neq \emptyset$, then $\tgt_i(x) =
  \src_i(y)$ and
  \begin{align*}
    \{x \odot_i y\}
    &= (x\odot_j \tgt_j(x))\odot_i(\src_j(y)\odot_j y)\\
    &\subseteq (x\odot_i \src_j(y))\odot_j (\tgt_j(x)\odot_j y)\\
    &= (x\odot_i \src_i(y))\odot_j (\tgt_i(x)\odot_j y)\\
    &= (x\odot_i \tgt_i(x))\odot_j (\src_i(y)\odot_j y)\\
    &= \{x \odot_j y\}
  \end{align*}
  and therefore $x\odot_i y = x\odot_j y$ by functionality.

  Likewise, if $x\odot_i y= \emptyset$, then $x\odot_i y \subseteq y\odot_i
  x$.  If $x\odot_i y\neq \emptyset$, then $\tgt_i(x) =
  \src_i(y)$ and
  \begin{align*}
    \{x\odot_i y\}
    &= (\src_j(x)\odot_j x)\odot_i (y \odot_j \tgt_j(x))\\
    &\subseteq (\src_j(x)\odot_i y)\odot_j (x \odot_i \tgt_j(y))\\
    &= (\tgt_i(x)\odot_i y)\odot_j (x \odot_i \src_i(y))\\
    &= (\src_i(y)\odot_i y)\odot_j (x \odot_i \tgt_i(x))\\
    &= \{y \odot_j x\}.
  \end{align*}
Then $x\odot_i y= y\odot_i x$ and $x\odot_j y= y\odot_j x$ by
functionality and the previous result. 
\end{proof}

Unlike the classical Eckmann-Hilton collapse, the resulting
  structure is \emph{not} an abelian monoid: different elements can
  still have different units and $\odot_i$ (and therefore $\odot_j$)
  need not be total. There are $2$-element counterexamples for
  $\odot_0$ and $\odot_1$ in the case of $2$-categories.

  Finally we obtain a stronger collapse in the presence of an
  equational interchange law in the more general setting of
  multioperations and several units.

\begin{lemma}\label{lemma:collapse3}
  If the inclusions
  $(w\odot_ jx) \odot_i (y \odot_j z)\subseteq (w\odot_i y) \odot_j (x
  \odot_i z)$, for $0\le i<j<\omega$, are replaced by equations in
  the axiomatisation of $\omega$-$\src\tgt$-multimagmas, then
  $\odot_i$ and $\odot_j$ coincide, $\src_i=\src_j=\tgt_i=\tgt_j$ and
  $\odot_i$ (as well as $\odot_j$) is associative and commutative.
\end{lemma}

We have verified this result with Isabelle in two dimensions, but
leave a proof on paper to the reader.  Once again, the resulting
structure is not automatically an abelian monoid: different elements
can have different units and $\odot_i$ (and therefore $\odot_j$) need
neither be total nor functional. There are again $2$-element
counterexamples.


\section{Proofs}\label{A:proofs}

\begin{proof}[Proof of Lemma~\ref{lemma:lr-mgp-local}(2)-(4)]
  Item (2) is immediate from the axioms. For (3), suppose
  $xy=\{\src(x)\}$. Then $\tgt(x)=\src(y)$ and hence $\{y\}= \tgt(x)y$. Thus
  $y\in x^\inv xy = x^\inv \src(x)=x^\inv \tgt(x^\inv) = x^\inv$ and therefore
  $x^\inv = y$. Finally, (4) is immediate from (3) because
  $\src(x)\tgt(x)=\{\src(x)\}$ and $\tgt(x)\tgt(x)=\tgt(x)$ by
  Lemma \ref{lemma:mm-quiver}(3). 
\end{proof}

\begin{proof}[Proof of Lemma~\ref{lemma:lr-groupoid-props}]
  For (1), $x^\inv x= \{\tgt(x)\} = \{\src (x^\inv)\}$ implies
  $(x^\inv)^\inv = x$ by Lemma~\ref{lemma:lr-mgp-props}(2).
  
  For (2), suppose $x\in yz$. Then $\tgt(y) = \src(z)$ and
  $xz^\inv \subseteq yzz^\inv= y \src(z)= y\tgt(y) = \{y\}$. Moreover, by
  assumption, $\tgt(x)=\tgt(z)$ and therefore $\Delta(x,z^\inv)$ by
  locality. It then follows that $xz^\inv = \{y\}$ and hence
  $y\in xz^\inv$. The converse implication follows from (1).  The
  remaining equivalence follows by opposition.
\end{proof}

\begin{proof}[Proof of Lemma~\ref{lemma:lr-gpd-canc}]
  For (1), suppose $\src(x)=\tgt(z)=\src(y)$ and $zx=zy$. Then $z^\inv zx=
  z^\inv zy$, therefore $\tgt(z)x=\tgt(z)y$, $\src(x)x= \src(y)y$ and finally
  $x=y$.  (2)  follows by opposition. 
\end{proof}

\begin{proof}[Proof of Lemma~\ref{lemma:invol-props}]
  (1) and (2) are immediate consequences of sup-preservation.
  
  For (3),
  \begin{equation*}
    \gamma\le (\alpha\inf \beta)^\conv \Leftrightarrow \gamma^\conv \le \alpha\inf \beta
    \Leftrightarrow \gamma^\conv\le \alpha \land \gamma^\conv \le \beta \Leftrightarrow \gamma
    \le \alpha^\conv \land \gamma \le \beta^\conv \Leftrightarrow \gamma\le \alpha^\conv
    \inf \beta^\conv
  \end{equation*}
  implies the claim for $\inf$, and the proof for $\Inf$ is
  similar.
  
  For (4), $\top^\conv \le \top$ and hence $\top = \top^{\conv\conv}
  \le \top^\conv$, and the other proofs are equally simple.
  
  For (5), $\alpha^\conv \inf \beta =\bot \Leftrightarrow
  (\alpha^\conv \inf \beta)^\conv =\bot \Leftrightarrow \alpha\inf \beta^\conv = \bot$.
  
  For (6), we know that every quantale is a Kleene algebra an can
  use their induction axioms (see Section~\ref{S:globular-ka}). First,
  $\alpha^{\conv\ast\conv} = (1 + \alpha^\conv
  \alpha^{\conv\ast})^\conv = 1 + \alpha\alpha^{\conv\ast\conv}$ and
  hence $\alpha^\ast \le \alpha^{\conv\ast\conv}$ by star
  induction. Thus $\alpha^{\ast\conv} \le\alpha^{\conv\ast}$.  For
  the converse direction,
  $\alpha^{\ast\conv} = (1+\alpha^\ast \alpha)^\conv = 1 +
  \alpha^\conv \alpha^{\ast\conv}$ and therefore
  $\alpha^{\conv\ast} \le \alpha^{\ast\conv}$ by star induction. An
  alternative inductive proof uses the definition of the star in
  quantales.
\end{proof}

\begin{proof}[Proof of Lemma~\ref{lemma:modular-dedekind}]
  Let $Q$ be an involutive quantale. Suppose the modular law holds. We
  derive two auxiliary properties before deriving the Dedekind
  law. First, the modular law can be strengthened to
  \begin{equation*}
    \alpha\beta\inf \gamma  = (\alpha \inf \gamma\beta^\conv)\beta \inf \gamma,
  \end{equation*}
  because
  $\alpha\beta\inf\gamma = \alpha\beta \inf\gamma\inf \gamma \le (\alpha
  \inf \gamma\beta^\conv)\beta \inf \gamma$ using the modular law and
  $(\alpha \inf \gamma\beta^\conv)\beta \inf \gamma \le \alpha\beta\inf
  \gamma$ by properties of infs and order preservation. Second, a dual
  equational modular law
  \begin{equation*}
    \alpha\beta\inf \gamma = \alpha (\beta \inf \alpha^\conv \gamma)
    \inf\gamma
  \end{equation*}
  then follows because
  $\alpha\beta \inf \gamma = (\beta^\conv \alpha^\conv \inf
  \gamma^\conv)^\conv = ((\beta^\conv \inf \gamma^\conv
  \alpha^{\conv\conv})\alpha^\conv \inf \gamma^\conv)^\conv = \alpha
  (\beta \inf \alpha^\conv \gamma) \inf\gamma$ using properties of
  converse and the first equational modular law in the second
  step. Therefore,
  \begin{align*}
    \alpha\beta\inf\gamma
    &= \alpha (\beta \inf \alpha^\conv \gamma) \inf\gamma\\
    &\le (\alpha \inf \gamma (\beta \inf \alpha^\conv
      \gamma)^\conv)(\beta\inf \alpha^\conv \gamma)\\
    &= (\alpha \inf \gamma (\beta^\conv \inf \gamma^\conv
      \alpha))(\beta\inf \alpha^\conv \gamma)\\
    &\le (\alpha \inf \gamma \beta^\conv)(\beta\inf \alpha^\conv \gamma),
  \end{align*}
  using properties of converse and order preservation as well as the
  dual equational modular in the first and the modular law in the second
  step. This proves the Dedekind law.
  
  Finally,
  $\alpha\beta\inf \gamma \le (\alpha \inf \gamma \beta^\conv)(\beta\inf
  \alpha^\conv \gamma) \le (\alpha \inf \gamma \beta^\conv)\beta$ yields
  the modular law from the Dedekind law using properties of inf and
  order preservation.
  
  These proofs are standard in relation algebra.
\end{proof}

\begin{proof}[Proof of Lemma~\ref{lemma:modular-props}]
  For the strong Gelfand property, $\alpha =1\alpha\inf \alpha \le
  \top \alpha \inf \alpha \le (\top\inf \alpha\alpha^\conv)\alpha =
  \alpha\alpha^\conv\alpha$ using the modular law in the third step.

  For Peirce's law, suppose $\alpha\beta\inf\gamma^\circ = \bot$. Then
  $(\alpha\beta\inf\gamma^\circ)\beta^\circ \inf \alpha = \bot$ and
  therefore $\beta\gamma\inf \alpha^\circ =\bot$ using the first
  equational modular law from the proof of
  Lemma~\ref{lemma:modular-dedekind} and properties of convolution.
  The converse implication is similar, using the dual equational
  modular law from Lemma~\ref{lemma:modular-dedekind}.

  For the first Schröder law,
  $\alpha\beta\inf \gamma = \bot \Leftrightarrow \beta \inf
  \alpha^\conv \gamma=\bot$, $\alpha\beta\inf\gamma = \bot$ implies
  $\beta^\conv \alpha^\conv \inf \gamma^\conv = \bot$ by properties of
  converse and therefore $\beta \inf \alpha^\conv \gamma=\bot$ by
  Peirce's law. The proof of the converse direction is similar. 

  The second Schröder law, $\alpha\beta\inf \gamma = \bot \Leftrightarrow \alpha \inf
  \gamma\beta^\conv = \bot$ follows in a similar way from Peirce's
  law, the first Schröder law and properties of converse.

  These proofs are once again standard in relation algebra.
\end{proof}

 \begin{proof}[Proof of Proposition~\ref{prop:modular-modal}]
   The proofs of $\dom(\alpha)\le 1$ and $\dom(\bot)=\bot$ are
   trivial.

   For $\alpha \le \dom(\alpha)\alpha$,
   $\alpha =\alpha\inf \alpha1 \le (1\inf \alpha\alpha^\conv)\alpha =
   \dom(\alpha)\alpha$ by the modular law, and the converse inequality
   follows from the second axiom.

   For $\dom(\alpha\dom(\beta))=\dom(\alpha\beta)$, we have
   $\dom (\alpha\dom(\beta)) = 1 \inf \alpha\dom(\beta)\top = 1\inf
   \alpha\beta\top = \dom(\alpha\beta)$.

   For $\dom(\alpha\sup\beta)=\dom(\alpha)\sup\dom(\beta)$, note that
   $\dom(\dom(\alpha))=\dom(\alpha)$ is immediate from the previous
   axiom. We first show that
   $\dom(\dom(\alpha)\sup\dom(\beta))=\dom(\alpha)\sup\dom(\beta)$. Indeed,
   \begin{align*}
     \dom(\dom(\alpha)\sup\dom(\beta))
     &= 1 \inf
       (\dom(\alpha)\sup\dom(\beta))(\dom(\alpha)\sup\dom(\beta))^\conv\\
     &= 1 \inf
       (\dom(\alpha)\sup\dom(\beta))(\dom(\alpha)\sup\dom(\beta))\\
     &= 1 \inf (\dom(\alpha)\sup\dom(\beta))\inf
       (\dom(\alpha)\sup\dom(\beta))\\
     &= 1 \inf (\dom(\alpha)\sup\dom(\beta))\\
         &= (1 \inf \dom(\alpha))\sup (1\inf \dom(\beta))\\
     &= \dom(\alpha)\sup\dom(\beta),
   \end{align*}
   where the third step uses the fact that multiplication of domain
   elements is commutative and idempotent. Proofs can be found in our
   Isabelle theories. Using this property with the alternative
   definition of domain,
   \begin{align*}
     \dom(\alpha \sup \beta)
     &= 1 \inf (\alpha \sup\beta)\top\\
       &= 1 \inf (\alpha\top \sup\beta\top)\\
     &= 1 \inf (\dom(\alpha)\top\sup\dom(\beta)\top))\\
     &= 1 \inf (\dom(\alpha)\sup\dom(\beta))\top\\
     &= \dom(\dom(\alpha)\sup\dom(\beta))\\
     &=\dom(\alpha)\sup\dom(\beta).
   \end{align*}

   For the compatibility property note that, by the strong Gelfand
   property and previous domain axioms
   $\dom(\alpha)\dom(\alpha)\dom(\alpha)=\dom(\alpha)$, therefore
   $\dom(\alpha)\dom(\alpha)=\dom(\alpha)$ and likewise for
   $\cod(\alpha)$. The proof of $\cod(\dom(\alpha))=\dom(\alpha)$ is
   similar.
\end{proof}

\begin{proof}[Proof of Lemma~\ref{lemma:quantale-dom-props}]
  For (1), $1\inf\alpha\top = 1 \inf \alpha(\top \inf \alpha^\conv 1) = 1 \inf
  \alpha\alpha^\conv = \dom(\alpha)$ using the modular law.
  
  For (2), $\alpha\top = \dom(\alpha)\alpha\top \le \dom(\alpha)\top\top = \dom(\alpha)\top$
  and $\dom(\alpha)\top = (1\inf \alpha\top)\top \le \top \inf \alpha\top\top =
  \alpha\top$.
  
  Item (3) and (4) are  trivial. 
  
  Finally, (4) is immediate from the definition of $\dom$. 
\end{proof}

\begin{proof}[Proof of Lemma~\ref{lemma:omega-cat-props}]
  By the morphism axiom for $\src_i$, $\src_i(x\odot_j y) \subseteq
  \src_i(x)\odot_j\src_i(y)$. The right-hand side must not be empty
  whenever the left-hand side is defined, which is assumed. Hence it
  must be equal to $\{\src_i(x)\}$ because compositions of lower
  cells in higher dimensions are trivial. 
\end{proof}

\begin{proof}[Proof of Lemma~\ref{lemma:dom-conv}]
  By the strong Gelfand property,
  \begin{equation*}
    \dom(\alpha)\le \dom(\alpha)\dom(\alpha)^\conv\dom(\alpha) \le
    1\dom(\alpha)^\conv 1 = \dom(\alpha)^\conv,
  \end{equation*}
  from which $\dom(\alpha)^\conv \le \dom(\alpha)$ follows using the
  adjunction in Remark~\ref{rem:backhouse}. The proof for $\cod$
  follows by duality.
\end{proof}

\begin{proof}[Proof of Proposition~\ref{prop:reduced-props}]
  Before the proof proper, we derive
  $\src_i\circ \tgt_j \circ \src_i = \src_i$ and
  $\tgt_i\circ \src_j \circ \src_i = \src_i$.  We have
  \begin{align*}
    \{\src_i(x)\} &= \src_i(x) \odot_i \src_i(x)\\
                  &= (\src_j(\src_i(x)) \odot_j \src_i(x))\odot_i (\src_i(x)\odot_j \tgt_j(\src_i(x)))\\
                  &\subseteq (\src_j(\src_i(x)) \odot_i \src_i(x))\odot_j (\src_i(x)\odot_i \tgt_j(\src_i(x))),
  \end{align*}
  hence $\Delta_j (\src_j(\src_i(x)) \odot_i \src_i(x),
  \src_i(x)\odot_i \tgt_j(\src_i(x)))$. Thus
  $\Delta_i(\src_j(\src_i(x)),\src_i(x))$ and
  $\Delta_i(\src_i(x),\tgt_j(\src_i(x)))$, and therefore
  $\tgt_i(\src_j(\src_i(x))) = \src_i(\src_i(x))=\src_i(x)$ as well as
  $\src_i(\tgt_j(\src_i(x))) = \tgt_i(\src_i(x))=\src_i(x)$. 
  
  Next we derive the missing $n$-catoid axioms.  First
  we consider $\src_j\circ \src_i = \src_i$, $\src_j\circ \tgt_i = \tgt_i$,
  $\tgt_j\circ \src_i = \src_i$, and $\tgt_j\circ \tgt_i = \tgt_i$.  For the
  first one,
  \begin{align*}
    \{\src_i (x)\} &= \src_i(x)\odot_i \src_i(x) \\
                   &= (\src_j(\src_i(x))\odot_j \src_i(x))\odot_i (\src_i(x)\odot_j
                     \tgt_j(\src_i(x)))\\
                   & \subseteq (\src_j(\src_i(x))\odot_i \src_i(x))\odot_j (\src_i(x)\odot_i
                     \tgt_j(\src_i(x)))\\
                   &= (\src_j(\src_i(x))\odot_i \tgt_i(\src_j(\src_i(x))))\odot_j (\src_i(\tgt_j(\src_i(x)))\odot_i\tgt_i(
                     \tgt_j(\src_i(x))))\\
                   &= \src_j(\src_i(x))\odot_j \tgt_j(\src_i(x))
  \end{align*}
  and therefore $\src_i(x)= \src_j(\src_i(x))$ as well as
  $\src_i(x)= \tgt_j(\src_i(x))$. The remaining identities hold by
  opposition.
  
  Second, we derive the identities
  $\src_i \circ \src_j = \src_j\circ \src_i$,
  $\src_i \circ \tgt_j = \tgt_j\circ \src_i$,
  $\tgt_i\circ \src_i = \src_j \circ \tgt_i$ and
  $\tgt_i\circ \tgt_j = \tgt_j\circ \tgt_i$.  For the first,
  $\{\src_j(x)\} = \src_j(\src_i(x))\odot_i \src_j(x)
  =\src_i(x)\odot_i \src_j(x)$ and therefore
  $\Delta_i(\src_i(x),\src_j(x))$. It follows that
  $\src_j(\src_i(x))=
  \src_i(x)=\tgt_i(\src_i(x))=\src_i(\src_j(x))$. The remaining
  proofs are similar.
  
  Note that none of the proofs so far requires an associativity law.
  
  It remains to derive $\src_i(x\odot_j y) \subseteq \src_i(x) \odot_j\src_i(y)$
  and $\tgt_i(x\odot_j y) \subseteq \tgt_i(x) \odot_j \tgt_i(y)$.  We prove the
  first inclusion by cases. If $\src_i(x\odot_j y) =\emptyset$, then the
  claim is trivial. Otherwise, if $\src_i(x\odot_j y) \neq\emptyset$,
  then $\tgt_j(x)=\src_i(y)$ and thus
  \begin{align*}
    \src_i(x\odot_j y)
    &= \src_i(\src_j(x\odot_j y))\\
    &= \src_i(\src_j(x\odot_j \src_j(y)))\\
    &= \src_i(\src_j(x\odot_j \tgt_j(y)))\\
    &= \src_i(\src_j(x))\\
    &= \{\src_i(x)\}\\
    &=\src_i(x)\odot_j \tgt_j(\src_i(x))\\
    &= \src_i(x)\odot_j \src_i(\tgt_j(x))\\
    &= \src_i(x)\odot_j \src_i(\src_j(x))\\
    &= \src_i(x)\odot_j \src_i(x).
  \end{align*}
  This uses weak locality laws that are available in all catoids if
  the composition under the source operation is defined. See
  Lemma~\ref{lemma:msg-props1} and~\cite{FahrenbergJSZ21a} for
  details.  The second inclusion follows by opposition.
\end{proof}

\begin{proof}[Proof of Lemma~\ref{lemma:2-quantale-props}]
  For (1)
  $\cod_j\circ \dom_i = \cod_j\circ \dom_j\circ \dom_i = \dom_j\circ
  \dom_i=\dom_i$, and the remaining identities in (1) follow by
  opposition.

For (2), $1_j\cdot_i 1_j = d_j(1_j)\cdot_i d_j(1_j) = d_j(1_j\cdot_i
1_j)\le 1_j$ if $Q$ is strong and
\begin{equation*}
  1_j= 1_j\cdot_i 1_i = (1_j\cdot_j 1_j)\cdot_i
(1_i\cdot_j 1_j)\le (1_j\cdot_i 1_i)\cdot_j (1_j\cdot_i 1_j) =
1_j\cdot_j (1_j\cdot_i 1_j) = 1_j\cdot_i 1_j.
\end{equation*}
Further,
$1_i \cdot_j 1_i = d_i(1_i)\cdot_j d_i(1_i) = d_j(d_i(1_i))\cdot_j
d_i(1_i) = d_i(1_i) = 1_i$.

For (3), $1_i = 1_i \cdot_i 1_i = (1_j\cdot_j 1_i) \cdot_i (1_i\cdot_j
1_i) \le (1_j\cdot_i 1_i) \cdot_j (1_i\cdot_i 1_j) = 1_j\cdot_j 1_j =
1_j$.

For (4), $d_j(1_i)= d_j(d_i(1_i)) = d_i(1_i)=1_i$, and $d_i(1_j)\le
1_i$ as well as $1_i = d_i(1_i)\le d_i(1_j)$. The other identities in
(4) then follow by opposition.

For (5), $\dom_i\circ \dom_j = \dom_j \circ  \dom_i$ follows from
\begin{align*}
  \dom_i(\dom_j(\alpha))
  &= \dom_i(\dom_j(\dom_i(\alpha)\cdot_i \alpha))\\
  &= \dom_i(\dom_j(\dom_i(\alpha))\cdot_i \dom_j(\alpha))\\
  &= \dom_i(\dom_i(\alpha)\cdot_i \dom_j(\alpha))\\
  &= \dom_i(\alpha)\cdot_i \dom_i(\dom_j(\alpha))\\
  &= \dom_j(\dom_i(\alpha))\cdot_i\dom_i(\dom_j(\alpha))\\
  &\le  \dom_j(\dom_i(\alpha))\cdot_i 1_i\\
     &=  \dom_j(\dom_i(\alpha))\\
\end{align*}
and
\begin{align*}
  \dom_j(\dom_i(\alpha))
  &= \dom_i(\alpha)\\
  &= \dom_i(\dom_j(\alpha)\cdot_j \alpha)\\
  &\le \dom_i(\dom_j(\alpha))\cdot_j \dom_i(\alpha)\\
  &\le \dom_i(\dom_j(\alpha))\cdot_j 1_i\\
  &\le \dom_i(\dom_j(\alpha))\cdot_j 1_j\\
     &=\dom_i(\dom_j(\alpha)).
\end{align*}
Moreover, $\dom_i\circ \cod_j = \cod_j \circ \dom_i$ follows from
\begin{align*}
\dom_i(\cod_j(\alpha))
  &=\dom_i(\cod_j(\dom_i(\alpha)\cdot_i \alpha))\\
  &=\dom_i(\cod_j(\dom_i(\alpha))\cdot_i \cod_i(\alpha))\\
  &=\dom_i(\dom_i(\alpha)\cdot_i \cod_i(\alpha))\\
  &=\dom_i(\alpha)\cdot_i \dom_i(\cod_i(\alpha))\\
  &=\cod_j(\dom_i(\alpha))\cdot_i \dom_i(\cod_i(\alpha))\\
  &\le \cod_j(\dom_i(\alpha))\cdot_i 1_i\\
    &= \cod_j(\dom_i(\alpha))
\end{align*}
and
\begin{align*}
  \cod_j(\dom_i(\alpha))
  &= \dom_i(\alpha)\\
  &= \dom_i(\alpha \cdot_j \cod_j(\alpha))\\
  &\le \dom_i(\alpha) \cdot_j \dom_i(\cod_j(\alpha))\\
  & \le 1_i\cdot_j \dom_i(\cod_j(\alpha))\\
  & \le 1_j\cdot_j \dom_i(\cod_j(\alpha))\\
  &= \dom_i(\cod_j(\alpha)).
\end{align*}
The remaining two identities in (6) then follow by opposition.

Finally, for (6),
\begin{align*}
  \dom_i(\alpha\cdot_j \beta) &=\dom_i(\dom_j(\alpha\cdot_j \beta))\\
                              &=\dom_i(\dom_j(\alpha\cdot_j \dom_j(\beta)))\\
  &= \dom_i(\alpha\cdot_j\dom_j(\beta))
\end{align*}
and the second identity in (7) then follows by opposition.
\end{proof}

\begin{proof}[Proof of Lemma~\ref{lemma:2-quantale-star-props}]
  For the first property in (1) and $k=i$, we first use induction on
  $k$ to prove
  $\dom_i (\alpha) \cdot_i \beta^{k_1} \le (\dom_i(\alpha)\cdot_i
  \beta)^{k_j}$, where $(-)^{k_j}$ indicates that powers are taken
  with respect to $\cdot_j$. The base case follows from
  $\dom_i(\alpha)\cdot_i 1_j \le 1_i\cdot_i 1_j =1_j$.  For the
  induction step, suppose
  $\dom_i(\alpha) \cdot_i \beta^{k_j} \le (\dom_i(\alpha) \cdot_i
  \beta)^{k_j}$. Then
\begin{align*}
  \dom_i (\alpha) \cdot_i \beta^{(k+1)_j}
  &= \dom_j(\dom_i(\alpha))\cdot_i (\beta\cdot_j \beta^{k_j})\\
  &= (\dom_j(\dom_i(\alpha))\cdot_j \dom_j(\dom_i(\alpha)))\cdot_i (\beta\cdot_j
    \beta^{k_j})\\
  &= (\dom_i(\alpha)\cdot_j \dom_i(\alpha))\cdot_i (\beta\cdot_j \beta^{k_j})\\
  &\le (\dom_i(\alpha)\cdot_i \beta)\cdot_j (\dom_i(\alpha)\cdot_i \beta^{k_j})\\
  &\le (\dom_i(\alpha)\cdot_i \beta)\cdot_j (\dom_i(\alpha)\cdot_i \beta)^{k_j}\\
    &= (\dom_i(\alpha)\cdot_i \beta)^{(k+1)_j}.
\end{align*}
Using this property yields
\begin{align*}
  \dom_i(\alpha)\cdot_i \beta^{\ast_j}
  &= \dom_i(\alpha)\cdot_i\Sup_{k\ge 0} \beta^{k_j}\\
  &=\Sup\{\dom_i(\alpha)\cdot_i \beta^{k_j}\mid k\ge 0\} \\
  &\le \Sup_{k\ge 0} (\dom_i(\alpha)\cdot_i \beta)^{k_j}\\
  &= (\dom_i(\alpha)\cdot_i \beta)^{\ast_j}.
\end{align*}
The proof of the second property in (1) for $k=i$ is dual.

The proofs for (2) are very similar, but in the base case,
$\dom_j(x)\cdot_i 1_j \le 1_j\cdot_i 1_j$ needs to be shown, which
follows from Lemma~\ref{lemma:2-quantale-props}(2) and needs a strong
$\omega$-quantale.

For (3), we first use induction on $k$ to show that $(\alpha\cdot_j
\beta)^{k_i} \le \alpha^{k_i} \cdot_j\beta^{k_i}$. In the base case,
$1_i=1_i\cdot_j 1_i$. For the induction step, suppose $(\alpha\cdot_j
\beta)^{k_i} \le \alpha^{k_i} \cdot_j \beta^{k_i}$. Then
\begin{align*}
  (\alpha\cdot_j \beta)^{(k+1)_i}
  &= (\alpha\cdot_j \beta) \cdot_i  (\alpha\cdot_j \beta)^{k_i}\\
  &\le (\alpha\cdot_j \beta) \cdot_i  \alpha^{k_i}\cdot_j \beta^{k_i}\\
  &\le (\alpha\cdot_i \alpha^{k_i}) \cdot_j (\beta\cdot_i \beta^{k_i})\\
    & = \alpha^{{k+1}_i} \cdot_j \beta^{(k+1)_i}.
\end{align*}
Using this property, we get
$(\alpha\cdot_j \beta)^{k_i} \le \alpha^{\ast_i} \cdot_j \beta^{\ast_i}$ for all
$k\ge 0$ and thus
$(\alpha\cdot_j \beta)^{\ast_i} \le \alpha^{\ast_i} \cdot_j \beta^{\ast_i}$ by
properties of sup.
\end{proof}

\begin{proof}[Proof of Lemma~\ref{lemma:mod-props}]
  For (1),
  \begin{align*}
    \dom_0(\alpha)\cdot_1 \dom_0(\alpha)
    &= \dom_1(\dom_0(\alpha))\cdot_1\dom_1(\dom_0(\alpha))\\
    &=\dom_1(\dom_0(\alpha)\\
    &= dom_0(\alpha)
  \end{align*}
  and the proof for
  codomain follows by opposition.
  
  For (2),
  \begin{align*}
    \dom_i (\alpha \cdot_j \beta)
    &= \dom_i(\dom_j(\alpha\cdot_j \beta))\\
    &=\dom_i(\dom_j(\alpha\cdot_j \dom_j(\beta)))\\
    &= \dom_j(\alpha\cdot_j
    \dom_j(\beta)),
    \end{align*}
  and the proof for codomain follows by opposition.

  The proofs of (3) are similar to those of (2), inserting $\cod_j$
  instead of $\dom_j$ in the first step.

  For (4),
  \begin{align*}
    \dom_i (\alpha \cdot_i \beta)
    &= \dom_i (\alpha \cdot_i \dom_i(\beta))\\
    &= \dom_i (\alpha \cdot_i \dom_j(\dom_i(\beta)))\\
    &=  \dom_i (\alpha\cdot_i \dom_i(\dom_j(\beta)))\\
    &=  \dom_i (\alpha\cdot_i \dom_j(\beta))
  \end{align*}
  and the proof for codomain follows by opposition.

  For (5),
  \begin{align*}
    \dom_i(x\cdot_i y)
    &= \dom_i(\cod_j(x\cdot_i y))\\
    & \le \dom_i(\cod_j(x)\cdot_i \cod_j(y))\\
    &= \dom_i(\cod_j(x)\cdot_i \dom_i(\cod_j(y)))\\
    &= \dom_i(\cod_j(x)\cdot_i \dom_i(y))\\
      &= \dom_i(\cod_j(x)\cdot_i y).
  \end{align*}
  The second step uses the morphism law for $\cod_j$. If the quantale
  is strong, it can be replaced by an equational step. The proofs for
  $\cod_i$ follow by opposition. 

   Item (6) is immediate from (1) and interchange.

   Items (7) and (8) are immediate consequences of the weak
   homomorphism laws and the homomorphism laws, respectively.

   For (9),
   \begin{align*}
     \dom_i(\alpha)\cdot_j \dom_i(\beta)
     &= \dom_i (\dom_i(\alpha)\cdot_j \dom_i(\beta)) \cdot_0 (\dom_i(\alpha)\cdot_j
       \dom_i(\beta))\\
     &\le (\dom_i (\dom_i(\alpha))\cdot_j \dom_i(\dom_i(\beta)))) \cdot_0 (\dom_i(\alpha)\cdot_j
       \dom_i(\beta))\\
     &\le (\dom_i(\alpha)\cdot_j 1_i) \cdot_0 (1_i\cdot_j \dom_i(\beta))\\
     &\le (\dom_i(\alpha)\cdot_j 1_j) \cdot_0 (1_i\cdot_j \dom_j(\beta))\\
             &\le \dom_i(\alpha) \cdot_i\dom_i(\beta),
   \end{align*}
   where the first step uses domain absorption, the second a weak
   homomorphism law, and the remaining steps are straightforward
   approximations. For the converse direction, 
   \begin{align*}
     \dom_i\alpha) \cdot_i \dom_i(\beta)
     &= (\dom_i(\alpha) \cdot_j \dom_i(\alpha))\cdot_i (\dom_i(\beta) \cdot_j
       \dom_i(\beta))\\
       &\le (\dom_i(\alpha) \cdot_i \dom_i(\beta))\cdot_j (\dom_i(\alpha) \cdot_i
         \dom_i(\beta))\\
     &\le (\dom_i(\alpha) \cdot_i 1_i)\cdot_j (1_i \cdot_i \dom_i(\beta))\\
     &= \dom_i(\alpha) \cdot_j \dom_i(\beta),
   \end{align*}
where the first step uses (1), the second step the interchange law,
and the remaining steps are obvious.

For (10), using in particular (9), 
\begin{align*}
  &(\dom_i(\alpha)\cdot_j \dom_i(\beta))\cdot_i (\dom_i(\gamma)\cdot_j\dom_i(\delta))\\
  &= (\dom_i(\alpha)\cdot_i \dom_i(\beta))\cdot_i (\dom_i(\gamma)\cdot_i\dom_i(\delta))\\
  &= (\dom_i(\alpha)\cdot_i \dom_i(\gamma))\cdot_i (\dom_i(\beta)\cdot_i\dom_i(\delta))\\
    &= \dom_i(\dom_i(\alpha)\cdot_i \dom_i(\gamma))\cdot_i
      \dom_i(\dom_i(\beta)\cdot_i\dom_i(\delta))\\
  &= \dom_i(\dom_i(\alpha)\cdot_i \dom_i(
    \gamma))\cdot_j
     \dom_i(\dom_i(\beta)\cdot_i\dom_i(\delta))\\
     &= (\dom_i(\alpha)\cdot_i \dom_i(\gamma))\cdot_j
       (\dom_i(\beta)\cdot_i\dom_i(\delta)).\qedhere
       \end{align*}
     \end{proof}

     \begin{proof}[Proof of Corollary~\ref{corollary:globular-lift}]
       For
  (1), we show explicit proofs for the globular structure, starting
  with the interchange laws.  For $W,X,Y,Z\subseteq C$,
\begin{align*}
  a \in (W \odot_j X) \odot_i (Y \odot_j Z)
& \Leftrightarrow \exists w\in W, x\in X,y\in Y,z\in Z.\ a\in (w\odot_j x)
                                              \odot_i (y\odot_j z)\\
&\Rightarrow \exists w\in W, x\in X,y\in Y,z\in Z.\ a \in (w\odot_i y)
                                              \odot_j (x\odot_i z)\\
& \Leftrightarrow a \in  (W \odot_i X) \odot_j (Y \odot_i Z).
\end{align*}
This requires only the interchange law in $C$.  It remains
to extend the globular laws.

For $\dom_j(x\cdot_i y) \le \dom_j(x)\cdot_i \dom_j(y)$ and
$X,Y\subseteq X$, 
\begin{align*}
  a \in \src_j(X \odot_i Y) 
& \Leftrightarrow \exists b\in X\odot_i Y.\ a = \src_j(b)\\
& \Leftrightarrow \exists b, c\in X, d\in Y.\ a = \src_j(b) \land b \in
  c\odot_i d\\
& \Leftrightarrow \exists c\in X, d\in Y.\ a \in \src_j(c\odot_i d)\\
& \Rightarrow \exists c\in X, d\in Y.\ a \in \src_j(c)\odot_i
                                                                       \src_j(d)\\
& \Leftrightarrow \exists c\in \src_j(X), d\in \src_j(Y).\ a \in
                                                                                   c\odot_i d\\
& \Leftrightarrow a\in \src_j(X), \cdot_i\src_j(Y).
\end{align*}
The implication in the fourth step can be replaced by $\Leftrightarrow$
if $C$ is an $\omega$-category.  The proofs of
$\cod_j(x\cdot_i y) \le \cod_j(x)\cdot_i \cod_j(y)$ and its strong
variant follows by opposition.  These proofs require only the
respective morphism laws in $C$.

For $\dom_i(x\cdot_j y) \le \dom_i(x)\cdot_j \dom_i(y)$, note that
$\src_i(x\odot_j y) \subseteq \src_i(x) \odot_j \src_i(y)$ is
derivable even in the reduced axiomatisation of $n$-catoids
(Proposition~\ref{prop:reduced-props}). The proof is then similar to
the previous one.

Finally, for $\dom_j\circ \dom_i = \dom_i$, note that
$\src_j\circ \src_i = \src_i$ is derivable in the reduced
axiomatisation of $n$-catoids
(Proposition~\ref{prop:reduced-props}). Hence let $X\subseteq C$. Then
\begin{align*}
  a\in \src_j(\src_i(X)) \Leftrightarrow \exists b\in
  X.\ a=\src_j(\src_i(b)) \Leftrightarrow \exists b\in X.\ a=
  \src_i(b)\Leftrightarrow a\in \src_i(X).
\end{align*}

For (2), we consider the morphism laws for $\src_i$ and $\tgt_i$. We
consider these correspondences one by one between quantales and
catoids and thus mention the laws for $\src_i$ and $\tgt_i$ for the
context without locality. As for
Theorem~\ref{theorem:globular-corresp1}, we start with an explicit
proof of interchange for multioperations and powersets.
  \begin{align*}
     (w\odot_j x)\odot_i (y\odot_j z) &= (\{w\}\cdot_j \{x\}) \cdot_i
                                        (\{y\}\cdot_j \{z\}) \\
                                      &\subseteq (\{w\}\cdot_i \{y\}) \cdot_j (\{x\}\cdot_i \{z\})\\
    &= (w\odot_i y)\odot_j (x\odot_i z).
   \end{align*}
For the morphism law for $\src_j$, 
\begin{equation*}
  \src_j(x\odot_i y) = \dom_j(\{x\}\cdot_i \{y\}) \subseteq
  \dom_j(\{x\})\cdot_i \dom_j(\{y\}) = \src_j(x)\odot_i \src_j(y).
\end{equation*}
The inclusion in the second step becomes an equality if $Q$ is strong.
The proofs for the morphism laws for $\tgt_j$ are dual. The
proofs of the laws for $\src_i$ and $\tgt_i$ are very similar.
\end{proof}

\begin{proof}[Proof of Lemma~\ref{lemma:groupoid-dom-lift}]
  The proofs of (1) and (2) are obvious; (3) is immediate from (2):
  $X = \src(X)X\subseteq XX^\inv X$.  For (4),
  \begin{align*}
    a \in C_0\cap XX^\inv
    &\Leftrightarrow \src(a) = a \land \exists b\in X,c\in X^\inv \
      \src(a) \in b\odot c\\
    &\Leftrightarrow \src(a) = a \land \exists b\in X.\
      \src(a) = \src(b)\\
    &\Leftrightarrow a\in \src(X)
  \end{align*}
  and the proof for $\tgt$ is dual. The proofs for (5) are very similar:
   \begin{align*}
    a \in C_0 \cap X\top
    &\Leftrightarrow \src(a) = a \land \exists b\in X,c\in C.\
      \src(a) \in b\odot c\\
    &\Leftrightarrow \src(a) = a \land \exists b\in X.\
      \src(a) = \src(b)\\
    &\Leftrightarrow a\in \src(X)
   \end{align*}
   and likewise for $t$. Finally, for (6),
   $X\top = \src(X)X\top \subseteq \src(X)\top\top\subseteq
   \src(X)\top$ and
   $\src(X)\top \subseteq XX^\inv\top \subseteq X\top\top\subseteq
   X\top$.
 \end{proof}

 \begin{proof}[Proof of Lemma~\ref{lemma:fdia-glob}]

 The equations in (1) follow immediately from
 Lemma~\ref{lemma:mod-props}(2)-(5).
 
 For (2), $\dom_i (\alpha\cdot_i \dom_j (\beta\cdot_j \gamma)) \le \dom_i(\alpha\cdot_i
 (\dom_i(\beta) \cdot_j \dom_i(\gamma)))$ follows from the above using the weak
 morphism axiom, and a dual property holds for $\cod_i$ and
 $\cod_j$.

 For (3),
 \begin{align*}
   &\dom_i(\dom_k(\alpha)\cdot_i \dom_j(\beta\cdot_j \gamma))\\
   & \le dom_i(\dom_k(\alpha)\cdot_i (\dom_i(\beta)\cdot_j\dom_i(\gamma)))\\
   & \le dom_i((\dom_k(\alpha)\cdot_i \dom_i(\beta))\cdot_j (\dom_k(\alpha)\cdot_i
     \dom_i(\gamma)))\\
   & \le dom_i(\dom_k(\alpha)\cdot_i \dom_i(\beta))\cdot_j
     \dom_i(\dom_k(\alpha)\cdot_i \dom_i(\gamma))\\
   & = dom_i(\dom_k(\alpha)\cdot_i \beta)\cdot_j \dom_i(\dom_k(\alpha)\cdot_i \gamma).
 \end{align*}
 The first step uses an approximation from (2), the second
 Lemma~\ref{lemma:mod-props}(6), the third a weak morphism axiom
 and the last domain locality. The proof for codomains is similar. 

 For (4),
 \begin{align*}
\dom_j(\alpha \cdot_j \dom_j(\dom_i(\beta \cdot_i \gamma)))
&= \dom_j(\alpha \cdot_j \dom_i(\dom_j(\beta \cdot_i \gamma)))\\
&\le \dom_j(\alpha \cdot_j \dom_i(\dom_j(\beta) \cdot_i \dom_j(\gamma)))\\
   &= \dom_j(\alpha \cdot_j \dom_i(\dom_j(\beta) \cdot_i \gamma)).
 \end{align*}
 The first step uses Lemma~\ref{lemma:2-quantale-props}(5), the second
 applies a morphism axiom, the third~\ref{lemma:mod-props}(4).  The
 second step becomes an equality if $Q$ is strong.  The remaining
 proofs are similar.
 
 For (5),
 \begin{align*}
   \dom_0(\gamma)\cdot_0 \dom_1 (\alpha \cdot_1 \beta)
   & = \dom_1(\dom_0(\gamma))\cdot_0 \dom_1 (\alpha \cdot_1 \beta)\\
   & = \dom_1(\dom_0(\gamma)\cdot_0 (\alpha \cdot_1 \beta))\\
   &\le \dom_1 ((\dom_0(\gamma)\cdot_0
   \alpha) \cdot_1 (\dom_0(\gamma)\cdot_0 \beta))
 \end{align*}
 The first step uses an axiom, the second a strong morphism axiom and
 the third Lemma~\ref{lemma:mod-props}(6).  The second proof is
 similar.
\end{proof}

\begin{proof}[Proof of Lemma~\ref{lemma:dom-cod-conv}]
  For (1),
  $\dom(\alpha) \le \dom(\alpha)\dom(\alpha)^\conv \dom(\alpha) \le \dom(\alpha)^\conv$, and
  thus $\dom(\alpha)^\conv \le \dom(\alpha)^{\conv\conv} = \dom(\alpha)$ as
  converse is order preserving. The property for $\cod$ then follows by
  opposition.

  For (2),
  \begin{align*}
    \dom(\alpha^\conv)
    &= \dom((\alpha\cod(\alpha))^\conv)\\
    &= \dom(\cod(\alpha)^\conv x^\conv)\\
    &=\dom(\cod(\alpha) \alpha^\conv)\\
    &=\cod(\alpha)\dom(\alpha^\conv)\\
    &= \cod(\alpha\dom(\alpha^\conv))\\
    &= \cod(\alpha^{\conv\conv}\dom(\alpha^\conv)^\conv)\\
    &= \cod((\dom(\alpha^\conv)\alpha^\conv)^\conv)\\
    &= \cod(\alpha^{\conv\conv})\\
&=\cod(\alpha).
  \end{align*}
  The export steps in the fourth and fifth lines of the proof work
  because of the compatibility laws $\dom\circ \cod = \cod$ and
  $\cod\circ \dom =\dom$ of modal semirings.  The second property
  then follows by opposition.
\end{proof}

\begin{proof}[Proof of Lemma~\ref{lemma:involutive-ka}]
  First,
  $\alpha^{\conv\ast\conv} = (1 + \alpha^{\conv\ast} \alpha^\conv)^\conv = 1 + \alpha
  \cdot \alpha^{\conv\ast\conv}$ by star unfold and properties of
  involution. Thus $\alpha^\ast \le \alpha^{\conv\ast\conv}$ by star induction
  and therefore $\alpha^{\ast\conv} \le \alpha^{\conv\ast}$.

  Second,
  $\alpha^{\ast\conv} = (1+\alpha\alpha^\ast)^\conv = 1 + \alpha^\conv
  \alpha^{\ast\conv}$ by star unfold and properties of
  involution. Thus $\alpha^{\conv\ast} \le \alpha^{\ast\conv}$ by star
  induction.
\end{proof}

\begin{proof}[Proof of Lemma~\ref{lemma:nKA-star}]
  It suffices to show that
  $1_i + (\alpha \cdot_j \beta)\cdot_i (\alpha^{\ast_i} \cdot_j
  \beta^{\ast_i}) \le (\alpha^{\ast_i} \cdot_j \beta^{\ast_i})$ by
  star induction.  Then
  $1_i \le (\alpha^{\ast_i} \cdot_j \beta^{\ast_i})$ holds because
  $1_i = 1_i \cdot_j 1_i$ and $1_i\le \alpha^{\ast_i}$,
  $1_i\le \beta^{\ast_i}$ by standard Kleene algebra. Also, we have
  $(\alpha \cdot_j \beta)\cdot_i (\alpha^{\ast_i} \cdot_j
  \beta^{\ast_i}) \le (\alpha^{\ast_i} \cdot_j \beta^{\ast_i})$ by
  interchange and $\alpha\cdot_i \alpha^{\ast_i} \le \alpha^{\ast_i}$
  and $\beta\cdot_i \beta^{\ast_i} \le \beta^{\ast_i}$, again by
  standard Kleene algebra.
\end{proof}

\section{Diagrams for Main Structures}\label{A:diagrams}


\definecolor{cbblue}{RGB}{100,134,255}
\definecolor{cbred}{RGB}{220,37,127}
\definecolor{cborange}{RGB}{225,176,0}

The main structures used in this article are related in the following
two diagrams. The diagrams are drawn by analogy to the Hasse diagrams
of order theory. An increasing \textcolor{cbblue}{blue} line means
that the class at the lower node is a subclass of the class at the
higher node. For instance, every (single-set) category is a local
catoid. An increasing \textcolor{cborange}{orange} line means that
every element in the lower class can be obtained by truncation from an
element of the higher one. For instance, every $2$-catoid can be
obtained by truncation from an $\omega$-catoid. An increasing
\textcolor{cbred}{red} line means that the upper and lower class are
related by correspondence in the sense of Jónsson-Tarski duality. This
the case, for instance, for local $\omega$-catoids and
$\omega$-quantales.

\vspace{\baselineskip}

\begin{center}
  \begin{tikzpicture}[x=1cm, y=1cm]
    \node (C) at (6,0) {$\textcolor{gray}{\mathbf{C}}$};
    \node (2C) at (3,1) {$\textcolor{gray}{2\mathbf{C}}$};
    \node (oC) at (0,2) {$\textcolor{gray}{\omega\mathbf{C}}$};
    \node (Ct) at (6,3) {$\textcolor{gray}{\mathbf{Ct}}$};
    \node (2Ct) at (3,4) {$\textcolor{gray}{2\mathbf{Ct}}$};
    \node (oCt) at (0,5) {$\omega\mathbf{Ct}$};
    \node (lCt) at (7,1) {$\textcolor{gray}{\mathbf{lCt}}$};
    \node (loCt) at (1,3) {$\mathbf{l}\omega\mathbf{Ct}$};
    \node (Q) at (6,5) {$\textcolor{gray}{\mathbf{Q}}$};
    \node (iQ) at (3,6) {$\textcolor{gray}{\mathbf{iQ}}$};
    \node (mQ) at (7,3) {$\textcolor{gray}{\mathbf{mQ}}$};
    \node (oQ) at (1,8) {$\omega\mathbf{Q}$};

    \draw[-, line width=.8pt,cbblue] (C) -- (Ct);
    \draw[-,line width=.8pt,cbblue] (2C) -- (2Ct);
    \draw[-,line width=.8pt,cbblue] (oC) -- (oCt);
    \draw[-,line width=.8pt,cbblue] (C) -- (lCt);
    \draw[-,line width=.8pt,cbblue] (Ct) -- (lCt);
    \draw[-,line width=.8pt,cbblue] (oCt) -- (loCt);
    \draw[-,line width=.8pt,cbblue] (oC) -- (loCt);
    \draw[-,line width=.8pt,cborange] (C) -- (2C);
    \draw[-,line width=.8pt,cborange] (oC) -- (2C);
    \draw[-,line width=.8pt,cborange] (Ct) -- (2Ct);
    \draw[-,line width=.8pt,cborange] (oCt) -- (2Ct);
    \draw[-,line width=.8pt,cborange] (lCt) -- (loCt);
    \draw[-,line width=.8pt,cbblue] (mQ) -- (Q);
     \draw[-,line width=.8pt,cborange] (Q) -- (iQ);
    \draw[-,line width=.8pt,cbred] (2Ct) -- (iQ);
    \draw[-,line width=.8pt,cbred] (Ct) -- (Q);
    \draw[-,line width=.8pt,cbred] (lCt) -- (mQ);
    \draw[-,line width=.8pt,cbred] (loCt) -- (oQ);
    \draw[-,line width=.8pt,cborange] (iQ) -- (oQ);
    \draw[-,line width=.8pt,cborange] (Q) -- (oQ);
    \draw[-,line width=.8pt,cborange] (mQ) -- (oQ);   
  \end{tikzpicture}
\end{center}

\vspace{\baselineskip}

In the diagram above we write $\mathbf{C}$, $2\mathbf{C}$ and
$\omega\mathbf{C}$ for the classes of (single-set) categories, strict
2-categories and strict $\omega$-categories, respectively.  Further
$\mathbf{Ct}$, $\mathbf{lCt}$, $2\mathbf{Ct}$, $\omega\mathbf{Ct}$
$\mathbf{l}\omega\mathbf{Ct}$ stand for the classes of catoids, local
catoids, 2-catoids, $\omega$-catoids and local $\omega$-catoids,
respectively.  Finally, $\mathbf{Q}$, $\mathbf{iQ}$, $\mathbf{mQ}$,
$\omega\mathbf{Q}$, indicate the classes of quantales, interchange
quantales, modal quantales and $\omega$-quantales,
respectively. Structures introduced in this articles are shown in
black, all others are shown in \textcolor{gray}{gray}.

\vspace{\baselineskip}

\begin{center}
  \begin{tikzpicture}[x=1cm, y=1cm]
    \node (G) at (6,0) {$\textcolor{gray}{\mathbf{G}}$};
    \node (C) at (3,.5) {$\textcolor{gray}{\mathbf{C}}$};
    \node (oC) at (0,1) {$\textcolor{gray}{\omega\mathbf{C}}$};
       \node (Ct) at (3,2.5) {$\textcolor{gray}{\mathbf{Ct}}$};
    \node (oCt) at (0,3) {$\mathbf{l}\omega\mathbf{Ct}$};
    \node (oQ) at (0,5) {$\omega\mathbf{Q}$};
    \node (DQ) at (6,4) {$\textcolor{gray}{\mathbf{DQ}}$};
    \node (opC) at (4.5,1.5){$\textcolor{gray}{(\omega,p)\mathbf{C}}$};
    \node (lopCt) at (4.5,3.5) {$\mathbf{l}(\omega,p)\mathbf{Ct}$};
     \node (opQ) at (4.5,5.5) {$(\omega,p)\mathbf{Q}$};

     \draw[-, line width=.8pt,cbblue] (G) -- (C);
     \draw[-, line width=.8pt,cbblue] (C) -- (Ct);
     \draw[-, line width=.8pt,cbblue] (oC) -- (oCt);
     \draw[-, line width=.8pt,cbblue] (opC) -- (lopCt);
     \draw[-, line width=.8pt,cbblue] (oC) -- (opC);
     \draw[-, line width=.8pt,cbblue] (oCt) -- (lopCt);
     \draw[-, line width=.8pt,cbblue] (oCt) -- (lopCt);
     \draw[-, line width=.8pt,cbblue] (oQ) -- (opQ);
     \draw[-, line width=.8pt,cbblue] (G) -- (opC);
     \draw[-, line width=.8pt,cbblue] (DQ) -- (opQ);
     \draw[-, line width=.8pt,cborange] (oC) -- (C);
     \draw[-, line width=.8pt,cborange] (Ct) -- (oCt);
     \draw[-, line width=.8pt,cbred] (oCt) -- (oQ);
     \draw[-, line width=.8pt,cbred] (G) -- (DQ);
      \draw[-, line width=.8pt,cbred] (lopCt) -- (opQ);
  \end{tikzpicture}
\end{center}

\vspace{\baselineskip}

In the second diagram we further write $\mathbf{G}$ for the class of
groupoids, $(\omega,p)\mathbf{C}$ for that of $(\omega,p$-categories,
$\mathbf{l}(\omega,p)\mathbf{Ct}$ for the class of
$(\omega,p)$-catoids, $\mathbf{DQ}$ for that of Dedekind
quantales and
$(\omega,p)\mathbf{Q}$ for that of $(\omega,p)$-quantales.

\end{document}